\documentclass[10pt,conference]{IEEEtran} 

\usepackage[T1]{fontenc}

\usepackage{amsmath,amsthm,amsfonts,amssymb}
\usepackage{hyperref}
\usepackage[utf8]{inputenc}
\usepackage{url}
\usepackage{graphicx}
\usepackage{subcaption}
\usepackage{color}
\usepackage{xcolor}
\usepackage{xspace}
\usepackage{comment}
\usepackage{cryptocode}
\usepackage{msc}
\usepackage{amsthm} 
\usepackage{pgf}

\usepackage{tikz}
\usetikzlibrary{matrix,shapes,arrows,positioning,chains, calc}
\usetikzlibrary{arrows,automata}
\usetikzlibrary{shapes,positioning,arrows,arrows.meta,calc}
\usepackage{tikz-qtree}
\usepackage{multirow}
\pgfarrowsdeclarecombine{boebelivoll}{boebelivoll}{Circle}{>}{Circle}{<}
\pgfarrowsdeclarecombine{boebelileer}{boebelileer}{Circle[open]}{>}{Circle[open]}{<}
\setmsckeyword{}
\drawframe{no}
\usepackage[normalem]{ulem} 
\theoremstyle{remark}

\newtheorem{lemma}{Lemma}

\newtheorem{mydef}{Definition}
\newtheorem{theorem}{Theorem}
\newtheorem{assumption}{Assumption}

\newtheorem{requirement}{Requirement}
\theoremstyle{remark}
\newtheorem*{justremark}{Remark}
\newcommand{\remarksymbol}{$\triangle$}

\newcommand{\chan}[1]{\ensuremath{\mathrel{\xrightarrow{#1}}} }
\newcommand{\insecChan}[1]{\ensuremath{\mathrel{\circ\mkern-5mu\xrightarrow{#1}\mkern-5mu\circ}}}
\newcommand{\authChan}[1]{\ensuremath{\mathrel{\bullet\mkern-7mu\xrightarrow{#1}\mkern-5mu\circ}}}
\newcommand{\confChan}[1]{\ensuremath{\mathrel{\circ\mkern-7mu\xrightarrow{#1}\mkern-5mu\bullet}}}
\newcommand{\secChan}[1]{\ensuremath{\mathrel{\bullet\mkern-7mu\xrightarrow{#1}\mkern-5mu\bullet}}}

\newcommand{\insecChanDefault}{\ensuremath{\mathrel{\circ\mkern-5mu\xrightarrow{\default}\mkern-5mu\circ}}}

\newcommand{\insecChanReliable}{\ensuremath{\mathrel{\circ\mkern-5mu\xrightarrow{\reliable}\mkern-5mu\circ}}}
\newcommand{\authChanReliable} {\ensuremath{\mathrel{\bullet\mkern-7mu\xrightarrow{\reliable}\mkern-5mu\circ}}}
\newcommand{\confChanReliable} {\ensuremath{\mathrel{\circ\mkern-7mu\xrightarrow{\reliable}\mkern-5mu\bullet}}}
\newcommand{\secChanReliable} {\ensuremath{\mathrel{\bullet\mkern-7mu\xrightarrow{\reliable}\mkern-5mu\bullet}}}

\newcommand{\insecChanObservable}{\ensuremath{\mathrel{\circ\mkern-5mu\xrightarrow{\observable}\mkern-5mu\circ}}}
\newcommand{\authChanObservable} {\ensuremath{\mathrel{\bullet\mkern-7mu\xrightarrow{\observable}\mkern-5mu\circ}}}
\newcommand{\confChanObservable} {\ensuremath{\mathrel{\circ\mkern-7mu\xrightarrow{\observable}\mkern-5mu\bullet}}}
\newcommand{\secChanObservable} {\ensuremath{\mathrel{\bullet\mkern-7mu\xrightarrow{\observable}\mkern-5mu\bullet}}}
\newcommand{\edgeInsec}{{Circle[open]}-boebelileer}
\newcommand{\edgeAuth}{Circle-boebelileer}
\newcommand{\edgeSec}{Circle-boebelivoll}

\newcommand{\accountabilityCite}{\emph{K{\"{u}}sters et al.}~\cite{accountability}\xspace}

\newcommand{\platform}{\textit{P}\xspace} 
 
\newcommand{\server}{\textit{Auth}\xspace} 
\newcommand{\BB}{\textit{BB}\xspace}
\newcommand{\human}{\textit{H}\xspace}
\newcommand{\device}{\textit{D}\xspace} 
\newcommand{\auditor}{\textit{A}\xspace}

\newcommand{\provoterCast}{\textit{VoterC}\xspace}
\newcommand{\provoterAbstain}{\textit{VoterA}\xspace}
\newcommand{\proauth}{\textit{AuthP}\xspace}
\newcommand{\timeliness}{\textit{TimelyP}\xspace}
\newcommand{\protime}{\textit{TimelyP}\xspace}
\newcommand{\uniqueness}{\textit{Uniqueness}\xspace}
\newcommand{\indivVerif}{$\textit{IndivVerif}$\xspace}

\newcommand{\functional}{$\textit{Func}$\xspace}
\newcommand{\TalliedAsRecorded}{\textit{Tallied-as-Recorded}\xspace}
\newcommand{\verifEnd}{\textit{End2EndVerif}\xspace}
\newcommand{\receiptfree}{\textit{Receipt-freeness}\xspace}

\newcommand{\DRprop}[5]{\textit{DR}(#1,#2,#3,#4,#5)}
\newcommand{\eligibilityVerif}{\textit{EligVerif}\xspace}
\newcommand{\EligibilityVerif}{\textit{EligVerif}\xspace}

\newcommand{\adv}{\textit{adv}\xspace}
\newcommand{\advModel}[2]{ #1^{#2}}
\newcommand{\Ballot}{\textit{Ballot}}
\newcommand{\ballot}{\textit{b}}
\newcommand{\ballotPrime}{\tilde{b}}
\newcommand{\bs}{[b]}
\newcommand{\channel}{c}
\newcommand{\channeltype}{\textit{chanType}\xspace}
\newcommand{\CastBy}{\textit{castBy}}

\newcommand{\confirmation}{\textit{c}}
\newcommand{\corresponds}{\textit{Corr}}
\newcommand{\drequal}{\textit{DRequal}}
\newcommand{\drsimilar}{\textit{DRsimilar}}
\newcommand{\End}{\textit{End}}
\newcommand{\encp}[3]{\{#1\}_{#2}^{#3}}

\newcommand{\Evidence}{\textit{Ev}\xspace} 

\newcommand{\evidence}{\textit{ev}} 
\newcommand{\false}{\textit{false}}

\newcommand{\fresh}{\textit{fresh}}
\newcommand{\fullTime}{\textit{TimelyDR}\xspace}

\newcommand{\honest}{\textit{hon}}

\newcommand{\honestNetwork}{\textit{honestNetw}}

\newcommand{\Hs}{[H]}

\newcommand{\knows}{\textit{knows}}

\newcommand{\pair}[1]{\langle #1\rangle}
\newcommand{\prot}{\textit{Pr}}

\newcommand{\pk}{\textit{pk}}
\newcommand{\pkD}{\textit{pk}_{\device}}
\newcommand{\pkDused}{\textit{pk}_{\device}^u} 

\newcommand{\pkP}{\textit{pk}_P}
\newcommand{\pkS}{\textit{pk}_{\server}}
\newcommand{\pubTrace}{\textit{pubtr}}
\newcommand{\pubChan}{\textit{Pub}}

\newcommand{\receive}{\textit{rec}}
\newcommand{\reliable}{\textit{r}} 
\newcommand{\observable}{\textit{u}} 
\newcommand{\default}{\textit{d}}
\newcommand{\deccp}{\textit{decp}}

\newcommand{\sd}{.\;} 
\newcommand{\sone}{\textit{D1}\xspace}
\newcommand{\stwo}{\textit{D2}\xspace}

\newcommand{\send}{\textit{send}}

\newcommand{\sk}{\textit{sk}}
\newcommand{\skD}{\textit{sk}_\textit{D}}

\newcommand{\skS}{\textit{sk}_{\server}}

\newcommand{\smallertopo}{\sqsubseteq}
\newcommand{\smallertopoNotEqual}{\sqsubset}

\newcommand{\subterm}{\vdash}

\newcommand{\Tally}{\textit{Tally}}

\newcommand{\topo}{T} 
\newcommand{\topos}{T_S} 
\newcommand{\topoi}{T_I} 
\newcommand{\topom}{T_M} 
\newcommand{\tr}{\mathit{tr}}
\newcommand{\trP}{\mathit{tr}} 
\newcommand{\true}{\textit{true}}
\newcommand{\trust}{t}
\newcommand{\trustassumption}{\textit{trustType}\xspace}
\newcommand{\trustedforward}{\textit{trustFwd}\xspace}
\newcommand{\trustedreply}{\textit{trustRpl}\xspace}
\newcommand{\TR}{\textit{TR}}

\newcommand{\verdict}{\textit{Faulty}}
\newcommand{\verdictNoVote}{\textit{Faulty}}
\newcommand{\ver}{\textit{ver}} 
\newcommand{\verlist}{\textit{ver}_\textit{L}} 

\newcommand{\verifyIV}{\textit{verifyH}} 
\newcommand{\verifyCast}{\textit{verifyC}} 
\newcommand{\verifyNoVote}{\textit{verifyA}} 
 
\newcommand{\vote}{\textit{v}} 
\newcommand{\Vote}{\textit{Vote}} 
 
\newcommand{\vs}{[v]} 
\newcommand{\zkp}{\textit{zkp}} 
\newcommand{\zkproof}{\textit{p}}

\newcommand{\BBpkS}{\textit{BB}_{pk}} 
\newcommand{\BBpkP}{\textit{BB}_{pkP}} 
\newcommand{\BBpkD}{\textit{BB}_{pkD}} 
\newcommand{\BBballotWoSignature}{\textit{BB}_{\textit{woS}}}
\newcommand{\BBrecorded}{\textit{BB}_{\textit{rec}}}
\newcommand{\BBtallied}{\textit{BB}_{\textit{tal}}}

\newcommand{\BBH}{\textit{BB}_{\textit{H}}}

\newcommand{\verifyProof}{\textit{verify}_\textit{A2}} 
\newcommand{\verifyValidSig}{\textit{verify}_\textit{A1}}
\newcommand{\verifyzkp}{\textit{ver}_\textit{zk}}

\newcommand{\hShH}{\server^+\human^+}
\newcommand{\mShH}{\human^+}
\newcommand{\hSmH}{\server^+}
\newcommand{\protocolnameInst}{\textit{MixVote}\xspace} 
\newcommand{\protInst}{\prot_\textit{MV}\xspace}
\newcommand{\topoInst}{\topo_\textit{MV}\xspace}

\begin{document}

\title{Dispute Resolution in Voting}
\author{\IEEEauthorblockN{David Basin}
\IEEEauthorblockA{\textit{Department of Computer Science} \\
\textit{ETH Zurich}\\
basin@inf.ethz.ch}
\and
\IEEEauthorblockN{Sa\v{s}a Radomirovi\'{c}}
\IEEEauthorblockA{\textit{Department of Computer Science} \\
\textit{Heriot-Watt University}\\
sasa.radomirovic@hw.ac.uk}
\and
\IEEEauthorblockN{Lara Schmid}
\IEEEauthorblockA{\textit{Department of Computer Science} \\
\textit{ETH Zurich}\\
schmidla@inf.ethz.ch}
}

\maketitle

\begin{abstract}

In voting, disputes arise when a voter claims that the voting authority is dishonest and did not correctly process his ballot while the authority claims to have followed the protocol.
A dispute can be resolved if any third party can unambiguously determine who is right. 
We systematically characterize all relevant disputes for a generic, practically relevant, class of voting protocols.
Based on our characterization, we propose a new definition of dispute resolution for voting that accounts for the possibility that both voters and the voting authority can make false claims and that voters may abstain from voting.

A central aspect of our work is timeliness: a voter should possess the evidence required to resolve disputes no later than the election's end.
We characterize what assumptions are necessary and sufficient for timeliness in terms of a communication topology for our voting protocol class. 
We formalize the dispute resolution properties and communication topologies symbolically.
This provides the basis for verification of dispute resolution for a broad class of protocols.
To demonstrate the utility of our model, we analyze a mixnet-based voting protocol and prove that it satisfies dispute resolution as well as verifiability and receipt-freeness.
To prove our claims, we combine machine-checked proofs with traditional pen-and-paper proofs.

\end{abstract}
 \smallskip

\section{Introduction}
For a society to accept a voting procedure, the public must believe that the system implementing it works as intended, that is, the system must be \emph{trustworthy}.
This is essential as elections involve participants from opposing political parties that may neither trust each other nor the election authority. 
Nevertheless, there must be a consensus on the final outcome, including whether the election is valid.
This requires that voters and auditors can \emph{verify} that the protocol proceeds as specified and detect any manipulations, even if they do not trust the authority running the election.
To achieve this, the information relevant for checking verifiability may be published in a publicly accessible database, known as the \emph{bulletin board}.

\paragraph{The need for dispute resolution}
Ballots are cast privately in elections.
Thus only the voters themselves know if and how they voted.
If a voter claims that his ballot is incorrectly recorded or that he was hindered in recording his ballot, no other party can know, a priori, whether the voter is lying or if there was a problem for which the voting authority is responsible.
We call such unresolved situations \emph{disputes}.

When a dispute occurs, the honest parties must be protected.
That is, an honest voter who detects some manipulation must be able to convince third parties that the authority was dishonest.\footnote{Here dishonesty includes all deviations of the authority from the protocol specification, both due to corruption or to errors.}
In particular, when a voter checks whether his cast ballot is correctly recorded, then either this is the case (respectively, no ballot is recorded when he abstained from voting) or he can convince others that the authority was dishonest. 
Another problem is when a voter cannot even proceed in the protocol to perform such checks, for instance when he is not provided with a necessary confirmation. 
Hence, a timeliness guarantee must ensure by the election's end that an honest voter's ballot is correctly recorded or there is evidence that proves to any third party that the authority is dishonest.
Finally, in addition to protecting the honest voters from a dishonest authority, the honest authority must be protected from voters making false accusations.
That is, when the authority is honest, no one should be able to convince others of the contrary.

\paragraph{State of the art}
The vast majority of formal analyses of remote e-voting protocols do not consider dispute resolution at all, e.g.,~\cite{helios,BeleniosRF,civitas,belenios}.
Works that recognize the importance of dispute resolution~\cite{pubEvidence,BMDsStark} or that take aspects of it into account when proposing poll-site~\cite{Forensics,vvote2,scantegrity2b,scantegrity2a,vvote,audiotegrity,scantegrity3} or remote~\cite{remotegrity} voting protocols, reason about it only informally.
The most closely related prior works define different notions of \emph{accountability}~\cite{accountabilityBruni,accountabilityCSF19,accountability} that formalize which agents should be held accountable when a protocol fails to satisfy some properties.
These definitions are very general, but have 
been instantiated for selected voting protocols~\cite{accountabilityBruni,accountability,clash,sElect}. 
The accountability properties satisfied by these protocols do not guarantee the resolution of all disputes considered by our dispute resolution properties.
We provide a detailed comparison of accountability and our properties in Section~\ref{sec:relatedWork}.

\paragraph{Contributions}
Our work provides a new foundation for characterizing, reasoning about, and establishing dispute resolution in voting.
First, we systematically reason about what disputes can arise in voting for a generic, practically relevant, class of voting protocols. 
Our class comprises both remote and poll-site voting protocols that can be electronic or paper-based.
We then focus on disputes regarding whether the published recorded ballots correctly represent the ballots cast by the voters.
Based on our classification, we formally define dispute resolution properties in a symbolic formalism amenable to automated verification using the Tamarin tool~\cite{tamarin,tamarin2}.
This enables the analysis of a broad class of protocols with respect to dispute resolution. 
Moreover, we identify an important new property, which we call \emph{timeliness}, requiring that when a voter's ballot is recorded incorrectly he has convincing evidence of this by the election's end.
This property ensures the resolution of disputes that could not be resolved unambiguously in prior work.

Second, we demonstrate that timeliness can only be guaranteed under strong assumptions (for example, some messages must not be lost on the network) by systematically analyzing what communication channels and trust assumptions are necessary and sufficient to satisfy this property.
The result is a complete characterization of all topologies in our voting protocol class for which timeliness holds for some protocol.
Such a characterization can guide the design of new voting protocols where timeliness should hold, e.g., by identifying and thereby eliminating settings where timeliness is impossible.
We formally verify the claimed properties using proofs constructed by Tamarin and pen-and-paper proofs.

Finally, 
to simplify establishing dispute resolution in practice, we introduce a property, called \uniqueness, that can be checked by everyone and guarantees that each recorded ballot was cast by a unique voter. 
We prove for protocols where voters can cast at most one ballot that \uniqueness implies guarantees for voters who abstain from voting.
This has the practical consequence that in many protocols, the corresponding guarantees can be proven more easily.
We then present as a case study a mixnet-based voting protocol with dispute resolution and prove that our introduced properties hold, as well as standard voting properties such as verifiability and receipt-freeness.

Overall, our results can be used as follows.
In addition to specifying what messages are exchanged between the different agents, a voting protocol in our class specifies (i) how the election's result is computed, (ii) which verification steps are performed, and (iii) when the authority conducting the election is considered to have behaved dishonestly.
For (i), it is required that each protocol specifies a function $\Tally$.
For (ii), as voters must be able to check that no ballots were wrongly recorded for them, a function $\CastBy$ must map each ballot to the voter that has (presumably) cast it. 
Only if this is defined can a voter notice when a ballot was recorded for him that he has not cast.
Finally, dispute resolution requires that a protocol defines a dispute resolution procedure such that everyone can agree on (iii). 
For this purpose, a protocol may specify a set of executions $\verdict$ where the authority is considered to have behaved dishonestly and which 
only depends on public information and can therefore be evaluated by everyone.

Given a protocol with a dispute resolution procedure and a
communication topology, our topology characterization can be used to quickly conclude if the given topology is insufficient to achieve the timeliness aspect of dispute resolution. 
When this is the case, one can immediately conclude that not all dispute resolution properties can be satisfied.
When this is not the case, our formal definitions can be used to analyze whether all dispute resolution properties indeed hold in the protocol. 
Thereby, in protocols where voters can cast at most one ballot, the guarantees for voters who abstain can be established directly or by showing \uniqueness and inferring them by our results.

\paragraph{Organization}
We describe our protocol model in Section~\ref{sec:protocolModel} and the class of voting protocols for which we define our properties in Section~\ref{sec:class}.
In Section~\ref{sec:disputeResolution}, we classify disputes and define our dispute resolution properties.
We then analyze in Section~\ref{sec:Timeliness} the communication topologies where timeliness can be achieved.
In Section~\ref{sec:uniqueness}, we show how dispute resolution can be established in practice, introduce \uniqueness, and present our case study.
Finally, we discuss related work in Section~\ref{sec:relatedWork} and conclude in Section~\ref{sec:conclusion}.

\section{Protocol model and system setup}\label{sec:protocolModel}

As is standard in model-checking, we model the protocol and adversary as a (global) transition system. 
Concretely, we use a formalism that also serves as the input language to the Tamarin tool~\cite{tamarin}. 
Our model uses abstractions that ease the specification of communication channels with security properties and trust assumptions.
These kinds of abstractions are now fairly standard in protocol specifications.
We complement existing abstractions \cite{hisp} (e.g., authentic and secure channels and parties that satisfy different kinds of trust assumptions) with new abstractions that are relevant for dispute resolution (e.g., reliable channels described in Section~\ref{subsubsec:channelTypes}).
Our protocol model is inspired by the model in \cite{alethea} used for e-voting.
We first introduce some terminology relevant for voting protocols and then our protocol model.

\paragraph{Terminology}
We distinguish between 
\emph{votes} and \emph{ballots}.
Whereas a vote is a voter's choice in plain text, a ballot contains the vote and possibly additional information. 
The ballots' exact design depends on the voting protocol, but it usually consists of the vote cryptographically transformed to ensure the vote's authenticity or confidentiality.
When a ballot is sent by the voter, we say it is \emph{cast}.
We denote by the \emph{(voting) authority} the entity responsible for collecting and tallying all voters' ballots. 
Usually, both the list of collected ballots, called the \emph{recorded ballots}, and the \emph{votes in the final tally}
are published on a \emph{public bulletin board} that can be accessed by voters and auditors to verify the election's result.

\subsection{Notation and term algebra}
\label{subsec:notationAndTermAlgebra}
We write $[x_i]_{i\in\{1,\dots,n\}}$ to denote a list of $n$ messages of the same kind. 
Similarly, we write $[f(x_i,y_i)]_{i\in\{1,\dots,n\}}$ for a list whose elements have the same form, but may have different values.
When the index set is clear from context, we omit the indices, e.g., we write $[x]$ and $[f(x,y)]$ for the above lists, and we write $[x]_i$ and $[f(x,y)]_i$ for the $i$th element in the lists.
Also, we write $x:=y$ for the assignment of $y$ to~$x$.

Our model is based on a term algebra $\mathcal{T}$ that is generated from the application of functions in a signature $\Sigma$ to a set of names $\mathcal{N}$ and variables $\mathcal{V}$. 
We use the standard notation and equational theory, given in~\cite[Appendix A]{hisp}. 
The symbols we use here are $\langle p_1,p_2 \rangle$ for pairing two terms $p_1$ and $p_2$, $\textit{fst}(p)$ and $\textit{snd}(p)$ for the first and second projection of the pair $p$, 
$\pk(x)$ for the public key (or the verification key) associated with a private key (or signing key)~$x$, and $\{m\}_{\sk}$ for a message $m$ signed with the signing key $\sk$.
The equational theory contains standard equations, for example 
pairing and projection obey $\textit{fst}(\langle p_1,p_2 \rangle)=p_1$ and
$\textit{snd}(\langle p_1,p_2 \rangle)=p_2$.
We sometimes omit the brackets $\langle \rangle$ when tupling is clear from the context.

We extend the term algebra from~\cite{hisp} with the following function symbols and equations.
We use $\ver(s,k)$ for signature verification, where $s$ is a signed message and $k$ the verification key.
When the signature in $s$ is verified with the (matching) verification key $k$, the function returns the underlying signed message~$m$ and otherwise it returns a default value $\perp$.
This is modeled by the equations $\ver(s,k)=m$, if $s=\{m\}_{\sk}$ and $k=\pk(\sk)$, and $\ver(s,k)=\perp$ otherwise. 

Moreover, we use the function $\Tally$ to model the tallying process in voting. 
Given a list of ballots $\bs$, $\Tally(\bs)$ denotes the computation of the votes $\vs$ in the final tally, possibly including pre-processing steps such as filtering out invalid ballots. The exact definition of $\Tally$ depends on the protocol.

Finally, $\CastBy(\ballot)$ denotes the voter who is considered to be the sender of a ballot $\ballot$.
As with $\Tally$, $\CastBy(\ballot)$ depends on the protocol, in particular on the ballots' design.
For example, in a voting protocol where a ballot $\ballot$ contains a voter's identifier (e.g., a code or pseudonym), $\CastBy(\ballot)$ maps the ballot $\ballot$ to the voter with the identifier included in $\ballot$.
In contrast, in a voting protocol where ballots contain a signature associated with a voter, $\CastBy(\ballot)$ maps each ballot to the voter associated with the signature contained in $\ballot$.

As $\Tally$ and $\CastBy$ are protocol dependent, each concrete protocol specification must define the equations that they satisfy, i.e., extend the term algebra's equational theory with equations characterizing their properties.
Note that the functions may not be publicly computable. 
For example, if only a voter $H$ knows which identifier or signature
belongs to him, then other parties are not able to conclude that $\CastBy(\ballot)=H$.

\subsection{Protocol specification}
A protocol consists of multiple \emph{(role) specifications} that define the behavior of the different communicating roles.
We model protocols as transition systems that give rise to a trace semantics.
Each role specification defines the role's sent and received messages and \emph{signals} that are recorded, ordered sequentially.
A signal is a term with a distinguished top-most function symbol.
Signals have no effect on a protocol's execution. They merely label events in executions to facilitate specifying the protocol's security properties.
We distinguish \emph{explicit signals} that are defined in the specification and \emph{implicit signals} that are recorded during the protocol execution but are not explicitly included in the specification.
We explain how we depict protocols as message sequence charts in Section~\ref{subsubsec:mixnet:protocol}.

Roles may possess terms in their \emph{initial knowledge}, which is denoted by the explicit signal $\knows$.
We require that in a specified role $R$, any message sent by $R$ must either occur in $R$'s initial knowledge or be deducible from the messages that $R$ initially knows or received in a previous protocol step. 
Deducibility is defined by the equational theory introduced above.
As is standard, whenever $R$ is specified to receive a term that it already has in its knowledge, an agent who instantiates this role will compare the two terms and proceed with the protocol only when they are equal.

\subsection{Adversary model and communication topology}
We depict the system setup as a \emph{topology graph} $G=(V,E)$, where the set of vertices $V$ denotes the roles and the set of (directed) edges $E \subseteq (V\times V)$ describes the available communication channels between roles (see e.g. Figure~\ref{fig:GenericSystemSetup}).
For two graphs $G_1=(V_1,E_1)$ and $G_2=(V_2,E_2)$, we define the standard subgraph relation $G_1\subseteq_G G_2$ as $V_1 \subseteq V_2 \wedge E_1\subseteq E_2$.

By default, we consider a Dolev-Yao adversary~\cite{DolevYao} who has full control over the network, learns all messages sent over the network, and can construct and send messages herself.
Additionally, the adversary can \emph{compromise} participating agents to learn their secrets
and control their behavior.
In a concrete system model, we limit the adversary by trust and channel assumptions. 
A \emph{(communication) topology}~\cite{hisp}
$\topo=(V,E,\trust,\channel)$ specifies the system setup by a graph $G(\topo)=(V,E)$,
trust assumptions by a function $\trust:V\mapsto \trustassumption$ mapping vertices to trust types,
and channel assumptions by a function $\channel:E\mapsto \channeltype$ mapping edges to channel types,
which denote
a channel with certain properties. 
The types \trustassumption and \channeltype are specified in Section~\ref{subsec:TrustAssumumptionAndChannelTypes}, after the execution model.
When $\channel$ is applied to an edge $(A,B)$, we omit duplicate brackets and write $\channel(A,B)$ instead of $\channel((A,B))$.

\subsection{Execution model, signals, properties, and assumptions}
\label{subsec:executionModel}
During protocol execution, roles are instantiated by agents (i.e., the parties involved in the protocol) and we consider all possible interleavings of agents' runs in parallel with the adversary.
A \emph{trace} $\tr$ is a finite sequence of multisets of the signals associated with an interleaved execution.
We denote by $\TR(\prot,\topo)$ the set of all traces of a protocol~$\prot$ that is \emph{run in the topology $\topo$}, i.e., run in parallel with the adversary defined by the topology $\topo$.\footnote{$\topo$ may specify channels that are never used by $\prot$. Also, $\prot$ may specify channels that are not available in $\topo$. In the latter case, the corresponding protocol steps cannot be executed and will not occur in the execution.}
We write $\tr_1\cdot \tr_2$ for the concatenation of two traces $\tr_1$ and $\tr_2$.

As previously explained, a trace may contain implicit signals, which are recorded during execution but omitted from the protocol specification for readability, and explicit signals (containing auxiliary information) that we explicitly add to the protocol specification.
Implicitly, when an agent $A$ sends a message $m$ (presumably) to $B$, the signal $\send(A,B,m)$ is recorded in the trace.
Similarly, when an agent $B$ receives $m$ (presumably) from $A$, $\receive(A,B,m)$ is recorded.
Furthermore, the signal $K(m)$ denotes the adversary's knowledge and is recorded whenever the adversary learns a term $m$ and $\honest(A)$ is recorded when an honest agent $A$ instantiates a role.

Furthermore, we use the explicit signal 
$\verifyCast( \human,\ballot)$ to indicate that an honest agent checks whether the ballot $\ballot$, which was cast by the honest voter $H$, is recorded correctly ($C$ stands for \emph{cast} and indicates that $H$ cast a ballot). 
In protocols where voters can cast multiple ballots, this signal can occur multiple times for the same voter.
Moreover, the signal $\verifyNoVote(\human,b_H)$ is recorded when an honest agent checks for an honest voter $\human$ who has cast the set of ballots $b_H$, that no ballots other than those in $b_H$ are recorded for $H$.
When this check is done for $H$ who abstained, then $b_H=\emptyset$ ($A$ stands for the fact that $H$ \emph{abstained} from voting).
$\verifyCast$ and $\verifyNoVote$ may be defined such that they can be computed by a machine but not by a human voter, e.g., when they require cryptographic computations.
We thus leave it open whether they are performed by the voter $H$ or by another agent such as a helper device.

The explicit signal $\knows(A,x)$ is recorded when an agent $A$ has a term $x$ in its initial knowledge.
The explicit signals $\Vote(H,\vote)$ and $\Ballot(H,\ballot)$ respectively record an honest voter $H$'s vote $\vote$ and cast ballot $\ballot$. 
The former is recorded when $H$ decides what to vote for and the latter is recorded when $H$ casts his ballot.
Finally, the explicit signal $\BB(m)$ denotes that a message $m$ is published on the bulletin board.
We use subscripts to distinguish the signals recorded when different messages are published on the bulletin board.
For example, the signals $\BBrecorded(\bs)$ and $\BBtallied(\vs)$ denote that the recorded ballots $\bs$ and the votes in the final tally $\vs$ are published.
We will introduce further signals as we need them.

We next define two kinds of trace properties.
The first are classical \emph{security properties}, which are specified as sets of traces.
A protocol $\prot$ run in the topology $\topo$, satisfies a security property $\mathcal{S}_S$, if $\TR(\prot,\topo) \subseteq \mathcal{S}_S$.
To reason about functional requirements, we additionally define \emph{functional properties}.
For example, the empty protocol satisfies many security properties but it is useless for voting because, even in the absence of the adversary, a voter's ballot is never recorded.
We will thus require a functional property stating that a protocol must at least have one execution where a voter's ballot is correctly recorded.
We describe a functional property by a set of traces $\mathcal{S}_F$, for example containing all traces where a voter's ballot is recorded. 
We then define that a protocol $\prot$ run in the topology $\topo$ satisfies the property $\mathcal{S}_F$ if $\TR(\prot,\topo) \cap \mathcal{S}_F \neq \emptyset$.
Finally, we define protocol \emph{assumptions} as sets of traces. 
That is, we define so-called \emph{(trace) restrictions} by giving a set of traces and then only consider the traces in the intersection of this set and $\TR(\prot,\topo)$ (see e.g. Section~\ref{subsubsec:channelTypes}).

\subsection{Trust and channel types}
\label{subsec:TrustAssumumptionAndChannelTypes}

\subsubsection{Trust types}
In the topologies, we consider four types of trust on roles that reflect the honesty of the agents that execute the role. 
A \emph{trusted} role means we assume that the agents who instantiate this role are always \emph{honest} and thus strictly follow their role specification.
In contrast, an \emph{untrusted} role can be instantiated by \emph{dishonest} agents (i.e., compromised by the adversary) who behave arbitrarily.
Dishonest agents model both corrupt entities and entities that unintentionally deviate from their specification, for example due to software errors.
We model dishonest agents by sending all their secrets to the adversary and by modeling all their incoming and outgoing channels as insecure (see the channel types below).

In addition, we consider the types \trustedforward and \trustedreply, which assume \emph{partial trust}.
The agents who instantiate a role of type \trustedforward or \trustedreply do not strictly follow their role specification but, respectively, always correctly forward messages or reply upon receiving correct messages.
Such assumptions turn out to often be necessary for the timeliness property that we introduce shortly, as otherwise dishonest agents that are expected to forward or answer certain messages can fail to do so and thereby block other protocol participants.
Thus, these trust types enable fine grained distinctions to be made about which assumptions are necessary for certain properties to hold.

In summary, we consider the set of trust types $\trustassumption:=\{\textit{trusted},\trustedforward, \trustedreply,\textit{untrusted}\}$.
In the topologies, we denote trusted roles by nodes that are circled twice (see e.g., $\BB$ in Figure~\ref{fig:topoT1}, p.~\pageref{fig:topoT1}) and the partial trust types \trustedforward and \trustedreply by two dashed circles 
(see e.g., $\platform$ in Figure~\ref{fig:topoT1}).
In our protocol class, there is no role that can be mapped both to type \trustedforward and to type \trustedreply; thus the interpretation will always be unambiguous.
All remaining roles are untrusted.

\subsubsection{Channel types}
\label{subsubsec:channelTypes}
In addition to the trust assumptions, a communication topology states channel assumptions.
Channels, which are the edges in the topology graph, denote which parties can communicate with each other. 
Also, channels define assumptions, for example that limit the adversary by stating who can change or learn the messages sent over a given channel.
Following Maurer and Schmid~\cite{maurerChannels}, we use the notation $A\insecChan{}B$, $A\authChan{}B$, $A\confChan{}B$, and $A\secChan{}B$ to denote a channel from (instances of) role $A$ to role $B$ that is respectively insecure, authentic, confidential, and secure.
For a formal semantics for these channels, see~\cite{hisp}.

We introduce two additional channel assumptions that are useful for dispute resolution. These assumptions concern whether a channel reliably delivers messages and whether external observers can see the communication on a channel.
Usually, it is assumed that the above channels are \emph{unreliable} in that the adversary can drop messages sent. 
We make such assumptions explicit and also allow for \emph{reliable} variants of channels.
On a reliable channel, the adversary cannot drop messages and thus all messages sent are received by the intended recipient.
We will see in Section~\ref{sec:Timeliness} that such channels are needed to achieve timeliness properties.

For dispute resolution, it is sometimes required that external observers can witness the communication an agent is involved in to later judge whether this agent is telling the truth. 
For example, it may be required that witnesses can observe when a voter casts his physical ballot by placing it into a voting box.
Such communication cannot later be denied, e.g. when others witness that the voter has cast his ballot then the voter cannot later deny this.
Whereas it is in reality sufficient if several witnesses, e.g., a subset of all voters, can observe such communication, we model this by channels that specify that \emph{any} honest agent can observe such communication.
Similarly, we will also model the fact that sufficiently many parties can decide who is right in a dispute by specifying that \emph{any} party can resolve disputes (see Section~\ref{sec:disputeResolution}).
Concretely, we model communication that can be observed by others by \emph{undeniable} channels where any honest agent $C\notin\{A,B\}$ learns the communication between $A$ and $B$.
This is in contrast to the default \emph{deniable} channels, where an honest agent $C\notin\{A,B\}$ cannot determine that $A$ and $B$ are communicating with each other.
 
A \emph{channel type} can be built from any combination of the three channel assumptions introduced above.
For example, on an insecure reliable channel, the adversary can learn all messages and write messages herself, but she cannot drop messages sent from $A$ to $B$.
However, for dispute resolution not all combinations
are useful.
In particular, an undeniable channel provides evidence that a message was sent, but this is only useful together with the guarantee that the message is also received.
Hence, we only consider undeniable channels that are also reliable.
We thus distinguish the following channel types, named after their most significant property:
The \emph{default channels} $\insecChan{\default}$, $\authChan{\default}$, $\confChan{\default}$, $\secChan{\default}$, which are neither reliable nor undeniable,
the \emph{reliable channels} 
$\insecChan{\reliable}$, $\authChan{\reliable}$, $\confChan{\reliable}$, $\secChan{\reliable}$, 
which are reliable but not undeniable, 
and the \emph{undeniable channels} 
$\insecChan{\observable}$, $\authChan{\observable}$, $\confChan{\observable}$, $\secChan{\observable}$, which are both reliable and undeniable.

We model the guarantees for senders and receivers that use a reliable or undeniable channel by stating that each message sent on such a channel is also received.
We only require this property when both the sender and the receiver of a message are trusted or partially trusted and formally express it by the following restriction.
\begin{alignat*}{1}
 &\{\tr | \forall A, B, m \sd
 \trust(A), \trust(B) \in \{\textit{trusted},\trustedforward, \trustedreply\}
 \\\multicolumn{2}{r}{$
 \wedge \,\channel(A,B)
 \in \{\insecChanReliable, \authChanReliable,\confChanReliable, \secChanReliable,
 \insecChanObservable, \authChanObservable, \confChanObservable, \secChanObservable\}$}
 \\\multicolumn{2}{r}{$\wedge\; \send(A,B,m) \in \tr \implies \receive(A,B,m) \in \tr \}. $} 
\end{alignat*}
To model the additional guarantee that undeniable channels provide, additional signals are recorded in the trace when agents communicate over such channels.
That is, whenever an agent $A$ sends a message $m$ to $B$ over an undeniable channel, in addition to the signals $\send(A,B,m)$ and $\receive(A,B,m)$, the signal $\pubChan(A,B,m)$ is recorded.
We formalize this by the following restriction.
\begin{alignat*}{1}
 &\{\tr | \forall A, B, m \sd
 \channel(A,B)
 \in \{\insecChanObservable, \authChanObservable, \confChanObservable, \secChanObservable\} \wedge\;
 \\\multicolumn{2}{r}{$
 (\send(A,B,m) \in \tr \vee \receive(A,B,m) \in \tr)
 $} 
 \\\multicolumn{2}{r}{$
 \implies \pubChan(A,B,m) \in \tr \}. 
 $}
\end{alignat*}
In the rest of this paper, these two restrictions are always stipulated.
That is, whenever we use $\TR(\prot,\topo)$ to refer to all traces of the protocol $\prot$ run in the topology $\topo$, we actually mean all traces in the intersection of $\TR(\prot,\topo)$ and the above two sets of trace restrictions.

\section{Class of voting protocols}\label{sec:class}
Formal reasoning about dispute resolution in voting requires a language for specifying voting protocols and their properties.
We provide such a language by presenting a class of voting protocols for which we subsequently define dispute resolution properties.
Our class comprises both remote and poll-site voting protocols that can be electronic or paper-based. However, we require a public bulletin board, which is, most of the time, realized by digital means.
We define the class by stating natural restrictions that communication topologies and protocols must satisfy to be in our class.
Afterwards, we show that many well-known voting protocols belong to this class.

\subsection{Communication topologies considered}
\label{subsec:consideredCommunicationTopologies}
\subsubsection{Topology graph}
\begin{figure}
 \begin{center}
 \scalebox{0.75}{
\begin{tikzpicture}[->,>=stealth',shorten >=1pt,auto,node distance=3cm and 1.5cm,semithick]

 \node[state] (H) {$\human$};
 \node[state] (D)[left=0.5cm of H] {$\device$}; 
 \node[state] (P)[right=0.7cm of H] {$\platform$};
 \node[state] (S) [right=0.7cm of P] {$\server$};
 
 \node[state] (BB)[below=0.2cm of P] {$\BB$};
 \node[state] (A) [below=0.2cm of H] {$\auditor$};

 \path (H) edge[->,bend left] node[left] {} (D) 
 (H) edge[->,bend left]node[above left] { } (P) 
 (D) edge[->,bend left] node[left] {} (H) 
 (P) edge[->,bend left]node[above right] { } (S) 
 (P) edge[-> ]node[above right] { } (H)
 
 (S) edge[-> ] node[below] { } (BB)
 (S) edge[->] node[right] {} (P)
 (BB)edge[-> ] node[below] { } (H)
 edge[-> ] node[below left] { } (A)
 ;
\end{tikzpicture}
,
\begin{tikzpicture}[->,>=stealth',shorten >=1pt,auto,node distance=3cm and 1.5cm,semithick]

 \node[state] (H) {$\human$};
 \node[state] (D)[left=0.5cm of H] {$\device$}; 
 \node[state] (S) [right of=H] {$\server$};

 \node[state] (A) [below=0.2cm of H] {$\auditor$};
 \node[state] (BB)[right=0.6 of A] {$\BB$};

 \path (H) edge[->,bend left] node[left] {} (D) 
 (D) edge[->,bend left] node[left] {} (H) 
 (H) edge[->,bend left]node[above right] { } (S) 
 (S) edge[-> ] node[below] { } (BB)
 (S) edge[->] node[right] {} (H)
 (BB)edge[-> ] node[below] { } (H)
 edge[-> ] node[below left] { } (A)
 ;
\end{tikzpicture}}
 \end{center}
 \caption{The topology graphs $G_S$ (left) and $G_U$ (right).
 We allow for any topology where $G(\topo)\subseteq_G G_S$ or $G(\topo)\subseteq_G G_U$.
 }\label{fig:GenericSystemSetup}
 \end{figure}
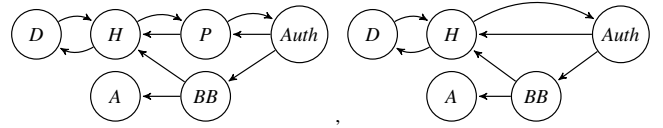
The topology graphs in Figure~\ref{fig:GenericSystemSetup} depict all possible roles and communication channels that we consider.
That is, we allow for any topology $\topo$ whose $G(\topo)$ is a subgraph of $G_S$ or $G_U$ in Figure~\ref{fig:GenericSystemSetup}.
The node $\human$ describes two roles for the human voters, one for voters who vote and one for those voters who abstain.
Also, there are roles for the devices $D$ and $\platform$, the voting authority $\server$, a public bulletin board $\BB$, and the auditor $\auditor$.
In a concrete protocol, each role, except for $\server$ and $\BB$, can be instantiated by multiple agents.

We consider two kinds of setups, $G_S$ and $G_U$ in Figure~\ref{fig:GenericSystemSetup}, for two kinds of protocols that differ in how ballots are cast.
$G_S$ provides the necessary channels for protocols where each voter
$H$ knows his ballot and sends it to $\server$ using a platform $\platform$.
It models remote and poll-site voting. 
In the former case $\platform$ could be the voter's personal
computer, and in the latter case $\platform$ could be a ballot box, or an optical scanner
that forwards $H$'s ballot $b$ to the authority for tallying. 
$G_U$ models setups for protocols where a trusted platform $P$
computes (e.g., encrypts) and casts the ballot for $H$. 
Often, such protocols do not distinguish between $H$ and $P$ and $P$
operates ``in the name of $H$''. Therefore, we model the setup of such
protocols by unifying the roles $H$ and $P$ into a single role
\emph{voter $H$}.

In some protocols, voters also have a personal off-line device $\device$. In contrast to $\platform$, $\device$ has limited capabilities and is not connected to the authority.
This models, for example, off-line trusted digital devices 
or letters containing codes that may be used to compute ballots.

$\server$ denotes the authority that is responsible for setting up elections and collecting and tallying the ballots.
Even though some voting protocols describe the authority in terms of several distinct roles, we collectively describe all these relevant functionalities in a single role, except for the publication of information, which is described by the bulletin board role $\BB$.
We then also consider just one agent in the role $\server$.
We will argue in Section~\ref{sec:disputeResolution} that this is sufficient for our dispute resolution properties.
As depicted in Figure~\ref{fig:GenericSystemSetup}, $\server$ can publish information on $\BB$, which can be read by the auditors and voters.
An auditor performs checks on the published information to ensure that the election proceeded correctly.
By modeling the auditor as a separate role, this role can be instantiated by anyone, including the voters.

\subsubsection{Topology assumptions}
\label{paragraph:toplogyAssumptions}
We further restrict the considered communication topologies by making some minimal channel and trust assumptions.
As is common for many voting protocols~\cite{alethea,civitas,belenios,jcj}, we model a secure bulletin board and consider its realization as a separate problem. Such a bulletin board can be used to send messages authentically and consistently from the authority to all voters and auditors.
We thus assume that the roles $\BB$ and $A$ exist and are trusted and that the channels from $\server$ to $\BB$ as well as from $\BB$ to $\human$ and $\auditor$ exist and are default authentic channels.
Furthermore, we only use the following partial trust assumptions.
\server can be trusted to always reply with a confirmation upon receiving a correct message (type \trustedreply) and \platform can be trusted to always forward messages correctly (type \trustedforward), e.g., a voting machine can be trusted to forward the entered ballots to a remote server.
The remaining channel and trust types can respectively be assigned to any channels and roles.
Note that we support protocols using anonymous channels (e.g. Civitas~\cite{civitas}) since, for the properties we consider, anonymous channels can be modeled as our default channels.\footnote{Distinguishing between anonymous and default channels is relevant when analyzing observational equivalence properties such as coercion resistance. 
However, for our possibility results, we only consider reachability properties.
}

For dispute resolution, certain guarantees should hold for an honest voter $H$, even with an untrusted or partially trusted authority and even if all other voters are untrusted.
Similarly, the guarantees for the honest authority should hold even when all voters are untrusted.
We therefore only consider topologies $\topo$ where the roles $\human$ and $\server$ are untrusted or partially trusted and analyze dispute resolution with respect to three variations of $\topo$.
We introduce the following notation.
We single out a distinguished voter $H$ for whom the security properties are analyzed.
Given a topology $T$, $\advModel{T}{\hShH}$ denotes the same topology but where the trust assumptions about $\server$ and the distinguished $H$ are defined by $\trust(\server)=\textit{trusted}$ and $\trust(H)=\textit{trusted}$, 
$\advModel{T}{\hSmH}$ is as $T$ but with $\trust(\server)=\textit{trusted}$, and 
$\advModel{T}{\mShH}$ is as $T$ but with $\trust(H)=\textit{trusted}$. 
Note that in all variations, the trust assumptions about the voters other than $H$ are as in~$T$.

\subsection{Voting protocols considered}\label{subsec:GenericVotingProtocol}
We next define the voting protocols considered in terms of the protocols' structure and which equations must be specified in the term algebra.
Our definition allows for protocols with re-voting, that is, where voters can send several ballots (e.g.,~\cite{civitas}).
As explained in Section~\ref{subsec:consideredCommunicationTopologies}, we allow protocols where the voter $H$ knows and casts his ballot or where a trusted platform $P$ casts the ballot, in which case we unify the roles $H$ and $P$.

\subsubsection{Required functions and equations}\label{subsubsec:necessaryFunctionsAndEquations}
A protocol specification must define the equation satisfied by $\Tally(\bs)$, defining how the election's result is computed from the list of recorded ballots $\bs$.
Similarly, a protocol must define the equation satisfied by $\CastBy(\ballot)$, which must map each ballot $\ballot$ to a voter, thereby specifying that this voter is considered to have cast the ballot.

\subsubsection{Protocol's start and end}
We assume that the protocol's setup can specify any number of voters and devices and any relation between them, for example that each voter is associated with a unique trusted device.
Also, at the protocol's start some public information may be posted on the bulletin board.
For example, this might be some election parameters or the list of all eligible voters, denoted by $\BBH(\Hs)$.
Furthermore, some agents may know some terms such that these terms (or associated terms) are initially published on the bulletin board or known to other agents.
For example, $\server$'s public key $\pkS:=\pk(\skS)$ can be posted on the bulletin board at the protocol's start whereby $\server$ has the corresponding private key $\skS$ in its initial knowledge.
We require, however, that at the protocol's start no honest agent knows a voter $H$'s ballot other than the voter himself.

We also assume that an election has two publicly known deadlines that determine the \emph{voting phase's end}, i.e., when ballots can no longer be cast, and the \emph{moment when all relevant information is published}.
We denote the latter by the explicit signal $\End$ in the $\BB$ role, which is recorded right after the last message relevant for the election is published.

\subsubsection{Tallying and publication of results}
We assume that ballots are collected and tallied by the authority $\server$ and that the protocol allows voters to abstain from voting.
Thus, $\server$ starts the tallying process after the voting phase, even if not all voters have cast a ballot. 
By the election's end, all valid ballots that were received by the authority have been published on the bulletin board together with the votes in the final tally.
In protocols with \emph{re-voting}, all ballots are published in the list of recorded ballots and the tallying process is responsible for removing duplicates.
Finally, we assume that all messages sent to the bulletin board are immediately published. That is, whenever $\BB$ receives a message $m$, the signal $\BB(m)$ is recorded in the trace.

\subsection{Examples of protocols in our voting class}
Our class comprises well known voting protocols such as Helios~\cite{helios}, Belenios~\cite{belenios}, and Civitas~\cite{civitas}.
In these protocols, voters can abstain from voting, the bulletin board is assumed to be secure, and the recorded ballots are published on the bulletin board as they were received by the authority.
Moreover, even though these protocols all specify different roles and setups, they can each be understood as instantiations of the setups
in Figure~\ref{fig:GenericSystemSetup}.
Belenios and Civitas both have many authority roles, such as registrars and different trustees, which can be understood as our role \emph{$\server$}.
In Belenios, the \emph{Bulletin Board} also performs some checks and computations. 
Thus, to cast it in our protocol class, we must additionally interpret those parts of Belenios's \emph{Bulletin Board} as part of our role \emph{$\server$} and just the published part of Belenios's \emph{Bulletin Board} as our role \emph{bulletin board $\BB$}.

Note that there are voting protocols, such as BeleniosRF~\cite{BeleniosRF}, where the recorded ballots are re-encrypted before being published on the bulletin board to achieve stronger privacy properties. 
Such protocols are not in our class.

\section{Dispute Resolution} 
\label{sec:disputeResolution}
In voting, the authority conducting the election should behave as expected.
That is, if the authority is dishonest, it must be held accountable for this.
For elections that are conducted by multiple parties, we require that it is unambiguously detectable when \emph{any} of these parties misbehave, but we do not require that it is detectable \emph{which} of these parties misbehave. 
This is sufficient to determine when ``the system'' running the election does not proceed as expected and to take recovery measures when this is the case, such as declaring the election to be invalid.
Thus, except for the bulletin board, we model all of the parties involved in conducting an election as one role (and agent) \emph{authority} $\server$ and require that this agent is held accountable if it does not behave as expected, i.e. does not follow its role specification.

In contrast to the authority, we should not and cannot require that all other parties, notably the voters, behave as expected. 
In fact, a well-designed voting protocol should still satisfy its expected properties for the honest voters, even when other voters misbehave. 
We therefore only consider disputes with respect to claims that the voting authority is dishonest.

We first explain why dispute resolution is needed in elections and characterize all relevant disputes.
Afterwards, to formalize our dispute resolution properties, we extend our protocol model with additional signals and functions.
We then motivate the required properties using our classification and formalize them using the model extensions.

\subsection{Relevant disputes}
After an election, all honest participants should agree on the election's outcome.
A protocol where any manipulation by the authority can be detected by suitable checks is called \emph{verifiable}.
For voting, the gold standard is \emph{end-to-end verifiability} where the final tally consists of the honest voters' votes, tallied correctly and this can all be checked.
Often, this property is divided into \emph{universal} and \emph{individual} verifiability.
\emph{Universal verifiability} properties denote that some guarantees hold (e.g., the tally is computed correctly) if an auditor performs appropriate checks on bulletin-board data. 
Any voter or independent third party can serve as an auditor and do such checks.
Therefore, if the universal verifiability checks fail, all honest protocol participants will agree on this fact and such checks never give rise to disputes.

\begin{figure*}
 \begin{center} 
 \scalebox{0.9}{ \begin{tabular}{|l |l | l| l |l |l| l l l l}
 \hline
 & \multicolumn{1}{c|}{\textbf{Voter} $H$ claims that $H$ }
 & \multicolumn{1}{c|}{\textbf{Authority} $\server$ claims that $H$ }
 & \multicolumn{1}{c|}{Properties protecting $H$}
 & \multicolumn{1}{c|}{Properties protecting $\server$}
 \\\hline\hline
 \sone& \textit{\textbf{cast} ballot} $\ballot$ & \textit{\textbf{did not cast} ballot} $\ballot$
 & $\provoterCast(\server)$, $\timeliness(\server) $
 & $\proauth(\server)$ \\\hline
 
 \stwo& \textit{\textbf{did not cast} ballot} $\ballot$ &\textit{\textbf{cast} ballot} $\ballot$
 & $\provoterAbstain(\server)$ 
 & $\proauth(\server)$ \\\hline
 \end{tabular}}
 \end{center}
 \caption{Possible disputes in voting.
 The authority's claim is captured by the information on the bulletin board.
 The respective disputes can be resolved when all properties in the third and fourth columns hold.
 }\label{fig:disputeSituations}
\end{figure*}

\emph{Individual verifiability} denotes that each voter can verify that his own ballot has been correctly considered in the list of recorded ballots. 
As only the voter knows which ballot he has cast, this property relies on checks that must (and can only) be done by \emph{each voter himself}.
Hence, individual verifiability checks give rise to the following three problems, where a voter claims that the authority is dishonest while other protocol participants cannot determine whether the voter is lying.

(1) A voter is hindered from taking the protocol step where he casts his ballot, in particular he cannot complete one of the preceding protocols steps.
There may be technical as well as social reasons for this.
For example, the voter may fail to be provided with the necessary credentials in a setup phase or he cannot access a polling station.
For generality, we therefore consider disputes regarding the inability to cast a ballot as out of scope of this paper and focus in the following on disputes concerning whether the recorded ballots correctly reflect the ballots cast by the voters.

(2) A voter who successfully cast his ballot 
is hindered from reaching the verifiability step. For instance, this can happen when recording a ballot requires receiving a confirmation from the authority, which is sent to the voter too late or not at all.

(3) A voter's check whether his ballot is recorded correctly fails. 
This can happen when a voter detects that one of his cast ballots was not recorded correctly or when he detects that there is a ballot recorded for him that he never sent.

As a result of the above reasoning, based on (3) we distinguish two possible disputes that must be considered in voting protocols, which are depicted in Figure~\ref{fig:disputeSituations}.
In both disputes, a voter $\human$'s and the authority's claim about $\human$'s cast ballot differ, where the authority's claim is denoted by the information on the bulletin board.
We take the standpoint that the authority is responsible for setting up a working channel to the bulletin board. That is, if messages are not on the bulletin board that should be there, we consider this to be the authority's fault.
In the dispute~\sone, a voter $\human$ claims that he cast a ballot $\ballot$, while the authority $\server$ claims that $\human$ did not cast $\ballot$ and in the dispute~\stwo their claims are reversed.
Note that when $\human$ claims to have cast $\ballot$ and $\server$ claims that $\human$ cast $\ballot'$, this constitutes both a dispute~\sone with respect to the ballot $\ballot$ and a dispute~\stwo with respect to $\ballot'$.

We require that when a voter learns that the authority $\server$ did not record a ballot that he cast, he can convince the other honest participants that $\server$ is dishonest.
This is a prerequisite needed to take recovery measures when such manipulations occur.
The same must hold when a voter learns that $\server$ recorded a ballot for him that he did not cast.
We respectively denote these properties by $\provoterCast(\server)$ (in dispute~\emph{D1}) and $\provoterAbstain(\server)$ (in dispute~\emph{D2}).

As explained in (2), it is also a problem when a voter who casts a ballot is hindered from reaching his verifiability check in due time.
We thus require that a voter who casts a ballot has some timeliness guarantees, namely that by the election's end either his ballot is correctly recorded or he has evidence to convince others that the authority $\server$ is dishonest. We denote this property by $\timeliness(\server)$. 

Finally, it is possible that voters lie.
Therefore, we require that an honest authority $\server$ is protected from false convictions in any dispute. We denote this by $\proauth(\server)$.

Some protocols support re-voting, where
voters can send a set of ballots, all of which are recorded on the bulletin board. 
In this case, the dispute~\sone denotes that $H$ claims that at least one of his cast ballots is not listed by $\server$. We thus require that $\provoterCast(\server)$ and $\protime(\server)$ hold for each cast ballot.
Dispute~\stwo means that $H$ claims not to have cast some of the ballots that are recorded for him.
In such a dispute, the voter must be able to convince everyone that too many ballots are recorded for him and that the authority $\server$ is dishonest.
This guarantee generalizes the property $\provoterAbstain(\server)$, which we will define so that it covers both situations.
As before, the disputes~\sone and~\stwo can occur simultaneously, for example when $H$ claims he cast the ballots $\ballot_1$ and $\ballot_2$ and $\server$ claims that $H$ cast $\ballot_2$ and $\ballot_3$.

\subsection{Protocol model for dispute resolution}
To formalize dispute resolution for our class of voting protocols, we extend our protocol model from Section~\ref{sec:protocolModel}.

It may be required that agents collect evidence to be used in disputes.
We use the signal $\Evidence(\ballot,\evidence)$ to model that the evidence $\evidence$ concerning the ballot $\ballot$ is collected. 
We model the \emph{forgery of such evidence} by allowing any dishonest agent to claim that any term in its knowledge is evidence. That is, we allow the adversary to perform an action that records $\Evidence(\ballot,\evidence)$ for any terms $\ballot$ and $\evidence$ such that $K(\langle \ballot, \evidence\rangle)$.

As we have argued that all honest agents should be able to agree on the outcome of disputes, we do not specify which agents resolve disputes and how the collected evidence must be communicated to them to file disputes.
We merely define that a voting protocol can generate evidence such that \emph{any} third party who obtains this evidence can, together with public information, judge whether the authority $\server$ is dishonest.
Recall that in poll-site settings, undeniable channels are used to model that 
\emph{sufficiently many} witnesses can observe the relevant communication in practice. 
In this context, requiring that any third party can judge whether the authority is dishonest models that sufficiently many parties can decide this in practice.
Abstracting away from these details allows us to focus on which evidence and observations are required to resolve disputes, independently of how undeniable channels are realized in practice.

We thus model the verdict of whether $\server$ should be considered dishonest by a publicly verifiable property $\verdict$, which can be specified as part of each voting protocol, independently of any role.
$\verdict(\server,\ballot)$ defines a set of traces where the agent $\server$ is considered to have behaved dishonestly with respect to the ballot $\ballot$, i.e., $\ballot$ has presumably not been processed according to the protocol in these traces.
For example, $\verdict(\server,\ballot) := \{\tr | \exists B, \bs \sd \pubChan(\server,B,\ballot) \wedge \BBrecorded(\bs) \in \tr \wedge \ballot \notin \bs\}$ specifies that $\server$ is considered dishonest in all traces where the ballot $\ballot$ was sent from $\server$ to an agent $B$ on an undeniable channel but not included in the recorded ballots~$\bs$ published on the bulletin board.

To ensure that the verdict whether a trace is in the set $\verdict(\server,\ballot)$ is publicly verifiable, the specification of $\verdict(\server,\ballot)$ must depend just on evidence and public information.
We thus state the following requirement.
\begin{requirement}\label{requirement:verdict}
 $\verdict(\server,\ballot)$ may only be defined based on the signals $\BB$, $\Evidence$, and $\pubChan$.
\end{requirement}
\noindent
Whereas the above example satisfies this requirement, the set $\{\tr | \exists A, B,m,\bs \sd \send(A,B,m) \wedge \BBrecorded(\bs) \in \tr \wedge m \notin \bs\}$ is not a valid definition of $\verdict(\server,\ballot)$, as $\send$ is not one of the admissible signals.

As a consequence of the above requirement, not all signals in a trace $\tr$ are relevant for evaluating whether $\tr$ satisfies a given $\verdict$ definition.
In particular, let $\pubTrace(\tr)$ be a projection that maps a trace $\tr$ to the signals in $\tr$ whose top-most function symbol is one of $\BB$, $\Evidence$, or $\pubChan$, while maintaining the order of these signals.
Then, it follows from Requirement~\ref{requirement:verdict} that for all traces $\tr_1$ and $\tr_2$ such that $\pubTrace(\tr_1)=\pubTrace(\tr_2)$, it holds that $\tr_1 \in \verdict(\server,\ballot)$ iff $\tr_2\in\verdict(\server,\ballot)$.

\subsection{Formal dispute resolution properties}\label{subsec:ProVoterAndProAuth}
We now use our extended model to define the dispute resolution properties for our class of voting protocols. We formalize each property from Figure~\ref{fig:disputeSituations} as a set of traces.

First, we consider the property $\provoterCast(\server)$ that protects an honest voter who detects that one of his cast ballots is not recorded correctly by the authority $\server$. 
Intuitively, we require that if this happens,
the voter can then convince others that the authority is dishonest.
Specifically, the property states that whenever an honest voter $H$ (or one of his devices) 
reaches the step where he believes that one of his ballots $\ballot$ should be recorded on $\BB$, then either this ballot is correctly included in the list of recorded ballots on $\BB$ or everyone can conclude that the authority $\server$ is dishonest.
We define $\provoterCast(\server)$ as follows ($C$~denotes that a ballot has been cast).
\begin{mydef}\label{def:proVoter}
\begin{alignat*}{1}
\provoterCast(\server) := \; & \{\tr \mid
 \verifyCast( H, \ballot) \in \tr
 \\\multicolumn{2}{r}{$ 
 \implies (\exists \bs \sd 
 \BBrecorded(\bs)\in \tr 
 \wedge b \in \bs ) 
 \vee \tr \in \verdict(\server,\ballot)\}.
 $}
 \end{alignat*}
\end{mydef}
\noindent
Note that, for notational simplicity, here and in the rest of the paper, when using set comprehension notation like $\{x| F(x,\bar{y})\}$, all free variables $\bar{y}$ different from $x$ are universally quantified, i.e., $\{x|\forall \bar{y}\sd F(x,\bar{y})\}$.

The next property, $\timeliness(\server)$, states that an honest voter $H$ who casts a ballot $\ballot$ cannot be prevented from proceeding in the protocol such that his ballot is recorded or, if he is prevented, then he can convince others that the authority $\server$ is dishonest.
In particular, a voter's ballot must be recorded or there must exist evidence that the authority is dishonest \emph{within a useful time period}.
Note that we do not require that the resolution of disputes must take place before the election's end and there can be a complaint period afterwards.
However, we require that the necessary evidence exists by this fixed deadline as otherwise it could be received after the complaint period ended.
We now define $\timeliness(\server)$.
 \begin{mydef}\label{def:proTime}
 \begin{alignat*}{3}
 &\timeliness(\server) := 
 \{\tr \mid \exists \tr',\tr''\sd \tr = \tr'\cdotp\tr'' 
 \\\multicolumn{2}{r}{$ 
 \wedge \Ballot(H,b) \in \tr' 
 \wedge\; \End\in \tr''
 $}
 \\\multicolumn{2}{r}{$ 
 \implies (\exists \bs. \BBrecorded(\bs)\in \tr' \wedge b \in \bs )
 \vee \tr \in \verdict(\server,\ballot) \}.
 $}
 \end{alignat*}
 \end{mydef}
 \noindent
 The difference to $\provoterCast(\server)$ (Definition~\ref{def:proVoter}) is that we not only require the property when a verifiability check is reached, but whenever all the relevant information is published on the bulletin board (indicated by $\End$) and a voter's ballot was cast before this deadline.
 We illustrate this difference on an example in Section~\ref{subsubsec:mixvote:DR}.

For abstaining voters we define $\provoterAbstain(\server)$. It states
that when an honest voter $H$ who abstains from voting, or one of $H$'s devices, checks that no ballot is recorded for $H$, then either this is the case or everyone can be convinced that the authority $\server$ is dishonest.
We define this property such that it can also be used in protocols with re-voting, where a voter who cast a set of ballots $b_H$ checks that no additional ballots are wrongly recorded for him.
We define $\provoterAbstain(\server)$ as follows. 
\begin{mydef}\label{def:provoterAbstain}
\begin{alignat*}{1}
\provoterAbstain(\server) := \; &
\{\tr \mid \verifyNoVote(H,b_H) \in \tr
\\\multicolumn{2}{r}{$ \implies \neg ( \exists \bs, b\sd \BBrecorded(\bs) \in \tr \wedge b \in \bs \wedge\; \CastBy(\ballot) = H 
 $}
\\\multicolumn{2}{r}{$ \wedge\; b \notin b_H)
\vee\; \exists \ballot \sd \tr \in \verdictNoVote(\server,\ballot) \}. $}
 \end{alignat*}
\end{mydef}
\noindent
Note that $\CastBy$ is just a claim that $H$ has cast a ballot and does not imply that $H$ has actually cast it.
For example, in a protocol where ballots contain a voter's identity in plain text and $\CastBy(\ballot)$ is defined to map each ballot to the voter whose identity it contains, everyone can construct a ballot $\ballot$ such that $H=\CastBy(\ballot)$, even when $H$ has not cast it.

It is possible, of course, that a voter who claims that the authority is dishonest is lying. Thus, for dispute resolution to be fair, it must not only protect the honest voters but also an honest authority. 
We formalize by $\proauth(\server)$ that traces where the authority $\server$ is honest should not be in $\verdict(\server,\ballot)$ for any ballot $\ballot$.
\begin{mydef}\label{def:proauth}
\begin{alignat*}{3}
\proauth(\server) := \; & \{\tr \mid
 \honest(\server) \in \tr
 \\& \multicolumn{2}{r}{$
\implies \forall \ballot \sd \tr \notin \verdict(\server,\ballot) \}.
 $}
 \end{alignat*}
\end{mydef}

Even though the above properties are stated independently of any adversary model, $\provoterCast(\server)$, $\timeliness(\server)$, and $\provoterAbstain(\server)$ are guarantees for an honest voter $\human$ and must hold even when the authority $\server$ and other voters are dishonest. 
Similarly, $\proauth(\server)$ constitutes a guarantee for $\server$ and must hold even if all voters are dishonest.
Therefore, we define a \emph{dispute resolution property} by stating what property must be satisfied by a protocol 
(1) for the honest voter $H$, i.e., when the protocol is run in a topology where $H$ is honest,
and (2) for the honest authority $\server$.
Additionally, it is usually required that a protocol satisfies some functional requirement when the agents are honest.
Thus a dispute resolution property also specifies (3) which functional requirement must hold when both $\server$ and the voter $H$ are honest.
\begin{mydef}\label{def:useful}
Let $p_H$ and $p_\server$ be two security properties that must hold respectively for an honest voter $H$ and the honest authority $\server$ and let $p_f$ be a functional property that must hold when both agents are honest.
A protocol $\prot$, executed in a topology $\topo$, 
satisfies the dispute resolution property $\DRprop{\prot}{\topo}{p_H}{p_\server}{p_f}$ iff
\begin{alignat*}{3}
&
\TR(\prot, \advModel{\topo}{\hShH}) \subseteq p_H \cap p_\server
\wedge \TR(\prot, \advModel{\topo}{\mShH}) \subseteq p_H 
\\\multicolumn{2}{r}{$
\wedge 
\TR(\prot, \advModel{\topo}{\hSmH}) \subseteq p_\server 
\wedge 
 \TR(\prot, \advModel{\topo}{\hShH}) \cap p_f \neq \emptyset
 . $} 
\end{alignat*}
\end{mydef}
\noindent
For example, $\DRprop{\prot}{\topo}{\provoterCast(\server) \cap\timeliness(\server) \cap \provoterAbstain(\server)}{\proauth(\server)}{f}$
denotes that the protocol $\prot$ run in the topology $\topo$ satisfies all previously introduced properties in the required adversary models.
That is, it satisfies the properties $\provoterCast(\server)$, $\timeliness(\server)$, and $\provoterAbstain(\server)$ for an honest voter $H$, the property $\proauth(\server)$ for an honest authority $\server$, and the functional property $f$ for an honest voter and the honest authority (see the next section for an example of a functional property).
Another dispute resolution property that we consider in the next Section is
$\DRprop{\prot}{\topo}{\timeliness(\server)}{\proauth(\server)}{f}$, which states that the timeliness property $\timeliness(\server)$ should hold for an honest voter $H$ while $\proauth(\server)$ is preserved for the honest authority $\server$.

\section{Communication topologies and timeliness}\label{sec:Timeliness}

For $\timeliness(\server)$, it is a problem when messages are lost as some protocol participants may be waiting for these messages and thus cannot proceed in the protocol.
Intuitively, timeliness only holds under strong assumptions.
We investigate this next by systematically characterizing the assumptions needed for $\timeliness(\server)$ to hold in our protocol class.

\subsection{Problem scope}
\subsubsection{Dispute resolution property}
We aim at achieving timeliness guarantees for the voters while also maintaining the $\proauth(\server)$ property for the authority $\server$.
Furthermore, we are only interested in protocols where a voter's ballot can actually be recorded.
To express the third requirement, we formalize a functional property
stating that for a given protocol and topology, there must be an execution where an honest voter $\human$ casts a ballot $\ballot$ and where this ballot is published in the list of recorded ballots $\bs$ on the bulletin board before the last relevant information is published (indicated by $\End$).
Moreover, this property must hold when all agents and the network behave honestly.
We denote by $\honestNetwork$ the set of traces where all agents follow the protocol and messages are forwarded on all channels unchanged.
The required functional property is defined as the following set of traces.
\begin{mydef}\label{def:functional}
\begin{alignat*}{3}
 &\textit{\functional}\; := \; \{\tr \mid 
 \exists \tr', \tr'',\human, \ballot, \bs \sd
 \tr = \tr'\cdotp \tr''\wedge 
 \\&\multicolumn{2}{r}{$ 
 \Ballot(H,b) \in \tr'
 \wedge \BBrecorded(\bs)\in \tr'
 \wedge b \in \bs\wedge \End \in \tr'' 
 $}
 \\&\multicolumn{2}{r}{$ 
 \wedge \, \tr \in \honestNetwork \}.
 $}
\end{alignat*}
\end{mydef}

Given a protocol $\prot$ and a topology $\topo$, we would like the dispute resolution property $\DRprop{\prot}{\topo}{\,\timeliness(\server)}{\,\proauth(\server)}{\,\functional}$, which we write for conciseness as $\fullTime(\prot,\topo)$.

\subsubsection{Additional assumptions} 
In standard protocol models, honest agents can stop executing their role at any time.
For timeliness, this might be a problem as other agents may wait for their messages and cannot proceed in the protocol.
We thus state the additional assumption that honest agents do not abort the protocol execution prematurely.
Similarly, we assume that partially trusted agents execute the required action once they can.
Note that agents can still be blocked, e.g., when waiting for messages that are dropped on the network.
\begin{assumption}
 Honest agents always execute all protocol steps possible 
 and agents who instantiate a role that is trusted to forward or answer messages, i.e., partially trusted, perform this respective action once they can.
\end{assumption}
\noindent 
The assumption implies, for example, that when a role specifies a send event after a receive event, the agent instantiating the role will always perform the second step right after the first one.
However, an agent can also be blocked between two consecutive protocol steps, for example when multiple receive events are specified after each other and the agent must wait for all of them.

Under the above assumption, we characterize all topologies from our protocol class for which there exists a protocol that satisfies \fullTime, i.e., the set $\{\topo | \exists \prot \sd \fullTime(\prot,\topo) \}$.
As others~\cite{hisp},
we introduce a partial order on topologies such that a possibility result (i.e., the existence of a protocol such that $\fullTime(\prot,\topo)$ holds) for a weaker topology implies a possibility result for a stronger topology.
We then characterize the above set by providing the ``boundary'' topologies, i.e., the minimal topologies satisfying \fullTime.

\subsection{Communication topology hierarchy}
We define a partial order on topologies that, given two topologies, orders them with respect to their trust and system assumptions.
We first define a partial order on our trust and channel types.
For $t,t' \in \textit{\trustassumption}$, $t' \smallertopo t$ denotes that $t$ is a stronger assumption than $t'$. 
We thus have that $\textit{untrusted} \smallertopo \trustedforward \smallertopo \textit{trusted}$ and $\textit{untrusted} \smallertopo \trustedreply \smallertopo \textit{trusted}$.
We also define for two channel types $c$ and $c'$ that $c$ makes stronger assumptions than $c'$.
Formally, let $\smallertopo_0$ be the relation where $\insecChan{x} \smallertopo_0 \confChan{x} \smallertopo_0 \secChan{x}$ and $\insecChan{x} \smallertopo_0 \authChan{x} \smallertopo_0 \secChan{x}$ for all $x \in \{\default,\reliable,\observable \}$ and 
$\chan{\default} \smallertopo_0 \chan{\reliable} \smallertopo_0 \chan{\observable}$
for all $\chan{}\, \in \{ \insecChan{} , \authChan{}, \confChan{}, \secChan{}\}$.
We overload the symbol $\smallertopo$ and, for channel types, we write $\smallertopo=\smallertopo_0^*$, i.e., $\smallertopo$ is the reflexive transitive closure of $\smallertopo_0$.

Using the above, for two topologies $\topo_1=(V_1,E_1,\trust_1,\channel_1)$ and $\topo_2=(V_2,E_2,\trust_2,\channel_2)$ we say that $\topo_2$ makes at least as strong assumptions as $\topo_1$ if $\topo_2$ uses channel and trust assumptions that are at least as strong as those in $\topo_1$ and if $\topo_2$'s topology graph includes all the roles and communication channels that exist in $\topo_1$:
\begin{alignat*}{3}
\topo_1 \smallertopo \topo_2 &:= 
G(\topo_1) \subseteq_G G(\topo_2) 
\wedge \forall (v_a,v_b) \in E_1 \sd 
\\\multicolumn{2}{r}{$\channel_1(v_a,v_b) \smallertopo \channel_2(v_a,v_b) \wedge \forall v \in V_1 \sd \trust_1(v)\smallertopo \trust_2(v).$}
\end{alignat*}

We show next that defining the relation this way is useful for relating possibility results for different topologies.
In particular, if for a topology $\topo$ it is possible to satisfy $\fullTime(\prot,\topo)$ with some protocol, then for all topologies that make stronger assumptions, there is also a protocol that satisfies the property.
The lemma is proven in
Appendix~\ref{appendix:proofTopologyHierarchy}.
\begin{lemma}\label{lemma:possAndImposs}
 Let $\topoi\smallertopo\topos$ be topologies in our class. 
\begin{alignat*}{3}
&\exists \prot \sd \fullTime(\prot,\topoi)
 \implies \exists \prot'\sd
 \fullTime(\prot',\topos).
\end{alignat*}
\end{lemma}

\subsection{Characterization of topologies enabling $\fullTime$}
We next present the minimal topologies satisfying $\fullTime$ in our voting protocol class.
In combination with the above hierarchy, this allows us to fully characterize all topologies $\topo$ that enable $\fullTime(\topo,\prot)$ for some protocol $\prot$.

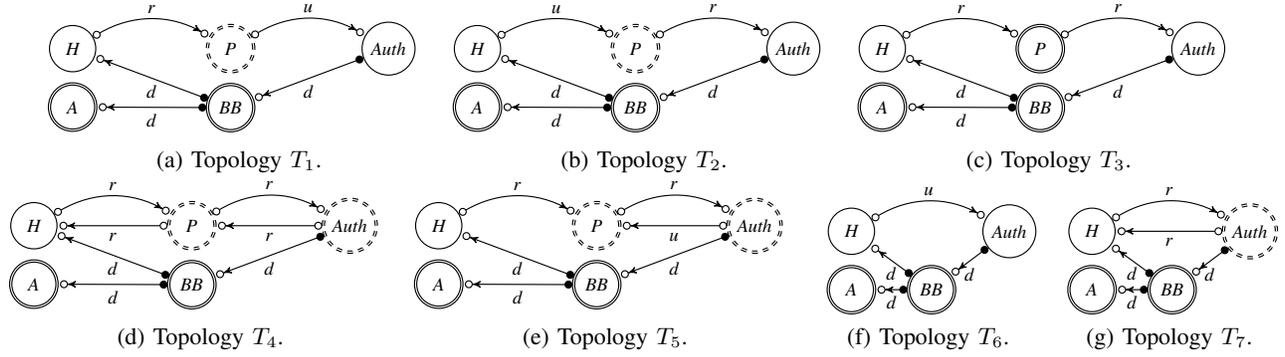
\begin{figure*} 
 \centering \begin{subfigure}[b]{0.29\textwidth}
 \scalebox{0.7}{
\begin{tikzpicture}[->,>=stealth',shorten >=1pt,auto,node distance=3cm and 1.5cm,semithick]

 \node[state] (H) {$\human$};
 \node[state, accepting,style=dashed](P)[right of=H] {$\platform$};
 \node[state] (S) [right of=P] {$\server$};
 
 \node[state, accepting
 ] (BB)[below=0.2cm of P] {$\BB$};
 \node[state, accepting
 ] (A) [below=0.2cm of H] {$\auditor$};

 \path 
 (H) edge[\edgeInsec,bend left]node[above] {\reliable} (P) 
 
 (P) edge[\edgeInsec,bend left]node[above] {\observable} (S) 
 
 (S) edge[\edgeAuth
 ] node[below] {\default} (BB)
 (BB)edge[\edgeAuth
 ] node[below] {\default} (H)
 edge[\edgeAuth
 ] node[below] {\default} (A)
 ;
\end{tikzpicture}} 
\caption{Topology $\topo_1$.}
\label{fig:topoT1}
\end{subfigure} 
\begin{subfigure}[b]{0.29\textwidth}
 \scalebox{0.7}{
\begin{tikzpicture}[->,>=stealth',shorten >=1pt,auto,node distance=3cm and 1.5cm,semithick]

 \node[state] (H) {$\human$};
 \node[state, accepting,style=dashed](P)[right of=H] {$\platform$};
 \node[state] (S) [right of=P] {$\server$};
 
 \node[state, accepting
 ] (BB)[below=0.2cm of P] {$\BB$};
 \node[state, accepting
 ] (A) [below=0.2cm of H] {$\auditor$};

 \path 
 (H) edge[\edgeInsec,bend left]node[above] {\observable} (P) 
 (P) edge[\edgeInsec,bend left]node[above] {\reliable} (S) 
 
 (S) edge[\edgeAuth
 ] node[below] {\default} (BB)
 (BB)edge[\edgeAuth
 ] node[below] {\default} (H)
 edge[\edgeAuth
 ] node[below] {\default} (A)
 ;
\end{tikzpicture}} 
\caption{Topology $\topo_2$.}
\end{subfigure}
\begin{subfigure}[b]{0.29\textwidth}
 \scalebox{0.7}{
\begin{tikzpicture}[->,>=stealth',shorten >=1pt,auto,node distance=3cm and 1.5cm,semithick]

 \node[state] (H) {$\human$};
 \node[state, accepting](P)[right of=H] {$\platform$};
 \node[state] (S) [right of=P] {$\server$};
 
 \node[state, accepting
 ] (BB)[below=0.2cm of P] {$\BB$};
 \node[state, accepting
 ] (A) [below=0.2cm of H] {$\auditor$};

 \path 
 (H) edge[\edgeInsec,bend left]node[above] {\reliable} (P) 
 
 (P) edge[\edgeInsec,bend left]node[above] {\reliable} (S) 
 
 (S) edge[\edgeAuth
 ] node[below] {\default} (BB)
 (BB)edge[\edgeAuth
 ] node[below] {\default} (H)
 edge[\edgeAuth
 ] node[below] {\default} (A)
 ;
\end{tikzpicture}} 
\caption{Topology $\topo_3$.}
\end{subfigure}
 \begin{subfigure}[b]{0.29\textwidth}
 \scalebox{0.7}{
\begin{tikzpicture}[->,>=stealth',shorten >=1pt,auto,node distance=3cm and 1.5cm,semithick]

 \node[state] (H) {$\human$};
 \node[state, accepting,style=dashed](P)[right of=H] {$\platform$};
 \node[state, accepting,style=dashed] (S) [right of=P] {$\server$};
 
 \node[state, accepting
 ] (BB)[below=0.2cm of P] {$\BB$};
 \node[state, accepting
 ] (A) [below=0.2cm of H] {$\auditor$};

 \path 
 (H) edge[\edgeInsec,bend left]node[above] {\reliable} (P) 
 (P) edge[\edgeInsec]node[below] {\reliable} (H) 
 
 (P) edge[\edgeInsec,bend left]node[above] {\reliable} (S) 
 (S) edge[\edgeInsec]node[below] {\reliable} (P) 
 
 (S) edge[\edgeAuth
 ] node[below] {\default} (BB)
 (BB)edge[\edgeAuth
 ] node[below] {\default} (H)
 edge[\edgeAuth
 ] node[below] {\default} (A)
 ;
\end{tikzpicture}} 
\caption{Topology $\topo_4$.}
\label{fig:topo4}
\end{subfigure}
 \begin{subfigure}[b]{0.29\textwidth}
 \scalebox{0.7}{
\begin{tikzpicture}[->,>=stealth',shorten >=1pt,auto,node distance=3cm and 1.5cm,semithick]

 \node[state] (H) {$\human$};
 \node[state, accepting,style=dashed](P)[right of=H] {$\platform$};
 \node[state, accepting,style=dashed] (S) [right of=P] {$\server$};
 
 \node[state, accepting
 ] (BB)[below=0.2cm of P] {$\BB$};
 \node[state, accepting
 ] (A) [below=0.2cm of H] {$\auditor$};

 \path 
 (H) edge[\edgeInsec,bend left]node[above] {\reliable} (P) 
 
 (P) edge[\edgeInsec,bend left]node[above] {\reliable} (S) 
 (S) edge[\edgeInsec]node[below] {\observable} (P) 
 
 (S) edge[\edgeAuth
 ] node[below] {\default} (BB)
 (BB)edge[\edgeAuth
 ] node[below] {\default} (H)
 edge[\edgeAuth
 ] node[below] {\default} (A)
 ;
\end{tikzpicture}} 
\caption{Topology $\topo_5$.}
\end{subfigure}
\begin{subfigure}[b]{0.17\textwidth}
 \begin{center}
 \scalebox{0.7}{
\begin{tikzpicture}[->,>=stealth',shorten >=1pt,auto,node distance=3cm and 1.5cm,semithick]

 \node[state] (H) {$\human$};
 \node[state] (S) [right of=H] {$\server$};
 
 \node[state, accepting
 ] (BB)[right=0.6 of A] {$\BB$};
 \node[state, accepting
 ] (A) [below=0.2cm of H] {$\auditor$};

 \path 
 
 (H) edge[\edgeInsec,bend left]node[above] {\observable} (S) 
 
 (S) edge[\edgeAuth
 ] node[below] {\default} (BB)
 (BB)edge[\edgeAuth
 ] node[below] {\default} (H)
 edge[\edgeAuth
 ] node[below] {\default} (A)
 ;
\end{tikzpicture}}\end{center}\vspace{-8pt}
\caption{Topology $\topo_6$.}
\end{subfigure}
\begin{subfigure}[b]{0.17\textwidth}
 \begin{center}
\scalebox{0.7}{
\begin{tikzpicture}[->,>=stealth',shorten >=1pt,auto,node distance=3cm and 1.5cm,semithick]

 \node[state] (H) {$\human$};
 \node[state, accepting,style=dashed] (S) [right of=H] {$\server$};
 
 \node[state, accepting
 ] (BB)[right=0.6 of A] {$\BB$};
 \node[state, accepting
 ] (A) [below=0.2cm of H] {$\auditor$};

 \path 
 
 (H) edge[\edgeInsec,bend left]node[above] {\reliable} (S) 
 (S) edge[\edgeInsec]node[below] {\reliable} (H) 
 
 (S) edge[\edgeAuth
 ] node[below] {\default} (BB)
 (BB)edge[\edgeAuth
 ] node[below] {\default} (H)
 edge[\edgeAuth
 ] node[below] {\default} (A)
 ;
\end{tikzpicture}} \end{center}\vspace{-8pt}
\caption{Topology $\topo_7$.}
\end{subfigure}
\caption{The minimal topologies for which there exists a protocol such that $\fullTime$ can be achieved.
The channels' labels denote whether the channels are default ($d$), reliable ($r$), or undeniable ($u$). The nodes' lines denote whether the roles are untrusted (circled once), trusted (circled twice), or partially trusted (dashed circles), where a partially trusted $P$ is of type $\trustedforward$ and a partially trusted $\server$ is of type $\trustedreply$.
}\label{fig:possResults}
\end{figure*}
The minimal topologies are depicted in Figure~\ref{fig:possResults} and denoted by $\topo_1,\dots,\topo_7$.
Recall that the agents $\BB$ and $\auditor$ as well as their incoming and outgoing channels 
have fixed trust assumptions. 
In all topologies, there are roles for $\human$, $\platform$ and $\server$ (respectively for $\human$ and $\server$ in $\topo_6$ and $\topo_7$), as this is required to satisfy the functional property ($\human$ must cast a ballot, $\platform$ must forward it, and $\server$ must publish it on $\BB$). 
All topologies have a reliable path from $\human$ to $\server$ and additional trust assumptions, such as (partially) trusted roles or undeniable channels.
We present some possible real-world interpretations of these topologies in Section~\ref{subsec:topologiesprovidingfulltime}.

We now state the main theorem for our voting protocol class:
The set of topologies for which there exists a protocol that establishes $\fullTime$ consists of all topologies that make at least the assumptions that are made by one of the seven topologies in Figure~\ref{fig:possResults}.
\begin{theorem}\label{theorem:completeChar}
Let $\topo_1,\dots,\topo_7$ be the topologies depicted in Figure~\ref{fig:possResults} and $T$ be a topology in our voting protocol class.
 \begin{alignat*}{3}
 &(\exists \prot \sd \fullTime(\prot,\topo) )
 \Leftrightarrow
 (\exists i\in \{1,\dots,7\} \sd \topo_i \smallertopo \topo).
 \end{alignat*}
\end{theorem}
We only explain the high level idea of the proof here and refer to 
Appendix~\ref{appendix:proofoftheorem:completeChar} 
for the details.
\begin{proof}[Proof Sketch]
First, we establish necessary requirements for topologies to enable $\fullTime$, by showing by pen-and-paper proofs that any topologies that do not meet these requirements cannot satisfy $\fullTime$ with any protocol.
Next, we show that these requirements, which are met by the topologies $\topo_1,\dots,\topo_7$ in Figure~\ref{fig:possResults} are sufficient. 
In particular, we prove by automated proofs in Tamarin (see~\cite{tamarinfiles}) that for each topology $\topo_i \in \{\topo_1,\dots,\topo_7\}$ there exists a simple protocol $\prot_i$ for which $\fullTime(\topo_i,\prot_i)$.
Finally, we show (by hand) that the topologies $\topo_1,\dots,\topo_7$ in Figure~\ref{fig:possResults} are the only \emph{minimal} topologies satisfying the necessary and sufficient requirements.
It follows that all topologies in our class are either stronger than one of the topologies $\topo_1,\dots,\topo_7$ and, by Lemma~\ref{lemma:possAndImposs}, also also establish $\fullTime$ with some protocol
or they are weaker than one of the topologies $\topo_1,\dots,\topo_7$ and thus do not meet the necessary requirements for $\fullTime$.
\end{proof}

The theorem shows that strong assumptions are indeed necessary to achieve timeliness for dispute resolution.
In particular, unreliable channels are insufficient.
In most cases, undeniable channels or trusted platforms are required. 
This can only be avoided in those topologies where there are reliable paths both from the voters to the authority and back.
Moreover, \fullTime cannot hold when the platforms are untrusted. 
This generalizes~\cite{BMDsStark}, which states that dispute resolution (called \emph{contestability} in~\cite{BMDsStark}) cannot hold in poll-site voting protocols where ballot-marking devices can be corrupted.

Recall that, in our protocol class, we also allow for off-line devices $\device$. Our analysis shows that $\device$ is irrelevant for the question of whether or not $\fullTime$ can be achieved.
Also, it is irrelevant whether the channels between the voters and the authority are authentic, confidential, or secure. 
In particular, they can all be insecure.
Nevertheless such devices and channels are needed in voting to satisfy other properties, for example privacy.

\section{Dispute resolution in practice}
\label{sec:uniqueness}

We now give a practical interpretation of the above results and
illustrate how our formalism can be used.

\subsection{Topologies providing \fullTime}
\label{subsec:topologiesprovidingfulltime}
Consider the topologies $\topo_1,\dots,\topo_7$ in Figure~\ref{fig:possResults}.
In $\topo_1$, there is an undeniable channel from $\platform$ to $\server$. 
When the platforms are physical ballot boxes, this can be interpreted as the assumption that sufficiently many witnesses see all ballots in the boxes and observe that they are forwarded and considered in the tallying process.
The undeniable channel between $\human$ and $\platform$ in $\topo_2$ could, for example, model that 
witnesses 
at each polling station 
observe 
voters' 
attempts
to cast their ballot, e.g., by scanning their already encrypted ballot on a voting machine~\cite{pretAVoter}.
The trusted $\platform$ in $\topo_3$ models, for example, that everyone trusts the 
voting machines 
used to compute and cast ballots. In this case, the machines can store a trustworthy record of what ballots have been cast for dispute resolution.

In topology $\topo_4$, the paths from $\human$ to $\server$ as well as from $\server$ to $\human$ are reliable. $\platform$ and $\server$ are respectively trusted to forward and reply. 
In a remote setting, the reliable channel from $H$ to $P$ could model that voters can always successfully enter messages on their platforms, for example on a working keyboard.
The voters could try with several platforms~\cite{remotegrity} to cast their ballot remotely and receive a confirmation from $\server$ or, in the worst case, go to a physical polling station to do so.
The assumptions then model that the voters can find a working platform and website (e.g., public platforms in libraries, polling places, etc.) or they can find a polling station that issues them with a valid confirmation before the election closes.
In $\topo_5$ and $\topo_6$, the undeniable channels could model a distributed ledger on which everyone can respectively observe when confirmations are issued or ballots are cast.
Finally, $\topo_7$ could model a remote setting similar to $\topo_4$, but where ballots are cast by the trusted platforms.

\subsection{Resolving dispute \stwo in protocols without re-voting}

In practice, the properties $\provoterCast(\server)$ and $\timeliness(\server)$ can be established in a protocol that provides evidence that a ballot was received by $\server$.
For example, this can be achieved by an undeniable channel or by a confirmation that is sent back from $\server$ to the voter upon a ballot's receipt.
$\verdict(\server,\ballot)$ can then be defined as the set of traces where a ballot $\ballot$ was received by $\server$ but is not in the set of recorded ballots on the bulletin board (see Section~\ref{subsec:CaseStudy} for a concrete example).
In contrast, it is often unclear how $\provoterAbstain(\server)$ can be established as a voter who abstains cannot prove the \emph{absence} of a message.
To solve this issue, we show that $\provoterAbstain(\server)$ is, in many cases,
entailed by the \uniqueness property, defined next, that can be
achieved using standard techniques. 
We prove this for protocols without re-voting and assume such protocols in the rest of this section.

$\uniqueness(\server)$ states that whenever any recorded ballots are published and $\server$ is not considered dishonest according to $\verdict(\server,\ballot)$ for some $\ballot\neq \perp$, then 
each recorded ballot $b'$ has been sent by a unique eligible voter $H$ for which $\CastBy(b')=H$.
Thereby, the ballot can be sent as part of a larger composed message. To express that a message $m'$ is a subterm of another message $m$, we write $m \subterm m'$.
As everyone can evaluate $\verdict(\server,\ballot)$, the property's preconditions and thus $\uniqueness(\server)$ are verifiable by everyone. 
\begin{mydef}\label{def:uniqueness}
Let the length of the list $\bs$ be $n$.
\begin{alignat*}{1}
&\uniqueness(\server) := \; \{\tr \mid 
{\ballot \neq \perp \wedge}
 \tr \notin \verdict(\server,\ballot)
 \\\multicolumn{2}{r}{$ 
 \wedge \BBrecorded(\bs) \in \tr 
 \wedge j \in \{1,\dots,n\}
 \wedge i \in \{1,\dots,n\} \implies
 $}
 \\\multicolumn{2}{r}{$ 
 \exists \Hs,i',j', A_1,A_2,m_1,m_2 \sd
 \BBH(\Hs) \in \tr 
 $}
 \\\multicolumn{2}{r}{$ 
 \wedge \CastBy(\bs_i)=\Hs_{i'} 
 \wedge \CastBy(\bs_j)=\Hs_{j'}
 $}
 \\\multicolumn{2}{r}{$ 
 \wedge \send(\Hs_{i'},A_1, m_1)\in \tr 
 \wedge \send(\Hs_{j'},A_2,m_2 )\in \tr 
 $}
 \\\multicolumn{2}{r}{$ 
 \wedge m_1\subterm \bs_i
 \wedge m_2 \subterm \bs_j 
 \wedge (i \neq j \implies \Hs_{i'} \neq \Hs_{j'})
 \}. 
 $}
\end{alignat*} 
\end{mydef}
\noindent
The property's guarantees are similar to \emph{eligibility verifiability}~\cite{kremerVerif} in that both state that each element of a list on the bulletin board is associated with a unique eligible voter and we compare the two notions in more detail in 
Appendix~\ref{caseStudy:SecurityProperties}.
Note that the property can only hold for protocols where the list of eligible voters is publicly known.

Intuitively, if a protocol satisfies $\uniqueness(\server)$, then a ballot recorded for the voter $H$ implies that $H$ cast it.
Thus, for any ballot that was not cast by $H$, $\server$ cannot convincingly claim the contrary and an honest voter is thus protected in disputes \stwo.
In particular, 
the traces in the protocol also satisfy $\provoterAbstain(\server)$.
We prove the following theorem in 
Appendix~\ref{appendix:securityproperties1}.

\begin{theorem}\label{theorem:uniquenessImpliesDR}
 Let $\prot$ be a protocol in our class without re-voting and where a voter who abstains does not send any message and let $\topo$ be a topology in our class.
\begin{alignat*}{1}
& \forall \tr \in \TR(\prot, \topo) \sd 
\\\multicolumn{2}{r}{$
\tr \in \uniqueness(\server) 
\implies \tr \in \provoterAbstain(\server).
$}
\end{alignat*}
\end{theorem}
\noindent
The theorem has the practical application that, while it is often unclear how $\provoterAbstain(\server)$ can be directly realized, 
$\uniqueness(\server)$ can easily be achieved using standard techniques, such as voters signing their ballots.
We provide an example in Section~\ref{subsec:CaseStudy}.

\subsection{How to use our formalism}
\label{subsec:howToUseOurFormalism}
Given the above results, our formalism can be used to analyze whether a protocol $\prot$ and topology $\topo$ in our class of voting protocols satisfy all dispute resolution properties introduced in Section~\ref{sec:disputeResolution}.
If there is no topology $\topo_1, \dots,\topo_7$ in Figure~\ref{fig:possResults}, such that $\topo_i \smallertopo \topo$, then we can immediately conclude, by Theorem~\ref{theorem:completeChar}, that $\timeliness(\server)$ and $\proauth(\server)$ cannot hold while the protocol is also functional.
Otherwise, 
analysis is required whether the properties are indeed satisfied by $\prot$.

First, let $\prot$ be a protocol that defines a dispute resolution procedure, i.e., specifies the set $\verdict$.
Our formalism is mainly intended for this case and can directly be used to analyze whether $\provoterCast(\server)$, $\protime(\server)$, $\provoterAbstain(\server)$, and $\proauth(\server)$ hold in such a protocol.
In protocols without re-voting that satisfy $\uniqueness(\server)$ and the preconditions of Theorem~\ref{theorem:uniquenessImpliesDR}, $\provoterAbstain(\server)$ can also be proven by proving $\uniqueness(\server)$ and concluding $\provoterAbstain(\server)$ by Theorem~\ref{theorem:uniquenessImpliesDR}. 

If a protocol $\prot$ does not define a dispute resolution procedure $\verdict$ then our properties are also undefined.
Nevertheless, one can still try to define a verdict $\verdict$ using the protocol's specified signals $\BB$, $\Evidence$, and $\pubChan$ and the terms contained in these signals.
Our formalism can then be used to establish for each such $\verdict$ which properties are satisfied.
However, to prove that no definition of $\verdict$ achieves dispute resolution, all possible combinations and relations of the above signals and their terms must be considered. 
Thus, it is in general not straightforward to efficiently conclude
that no appropriate definition of $\verdict$ exists for a given protocol.

\subsection{A mixnet-based voting protocol with dispute resolution}
\label{subsec:CaseStudy}

\begin{figure*}
\begin{center}
 \scalebox{0.8}{\includegraphics{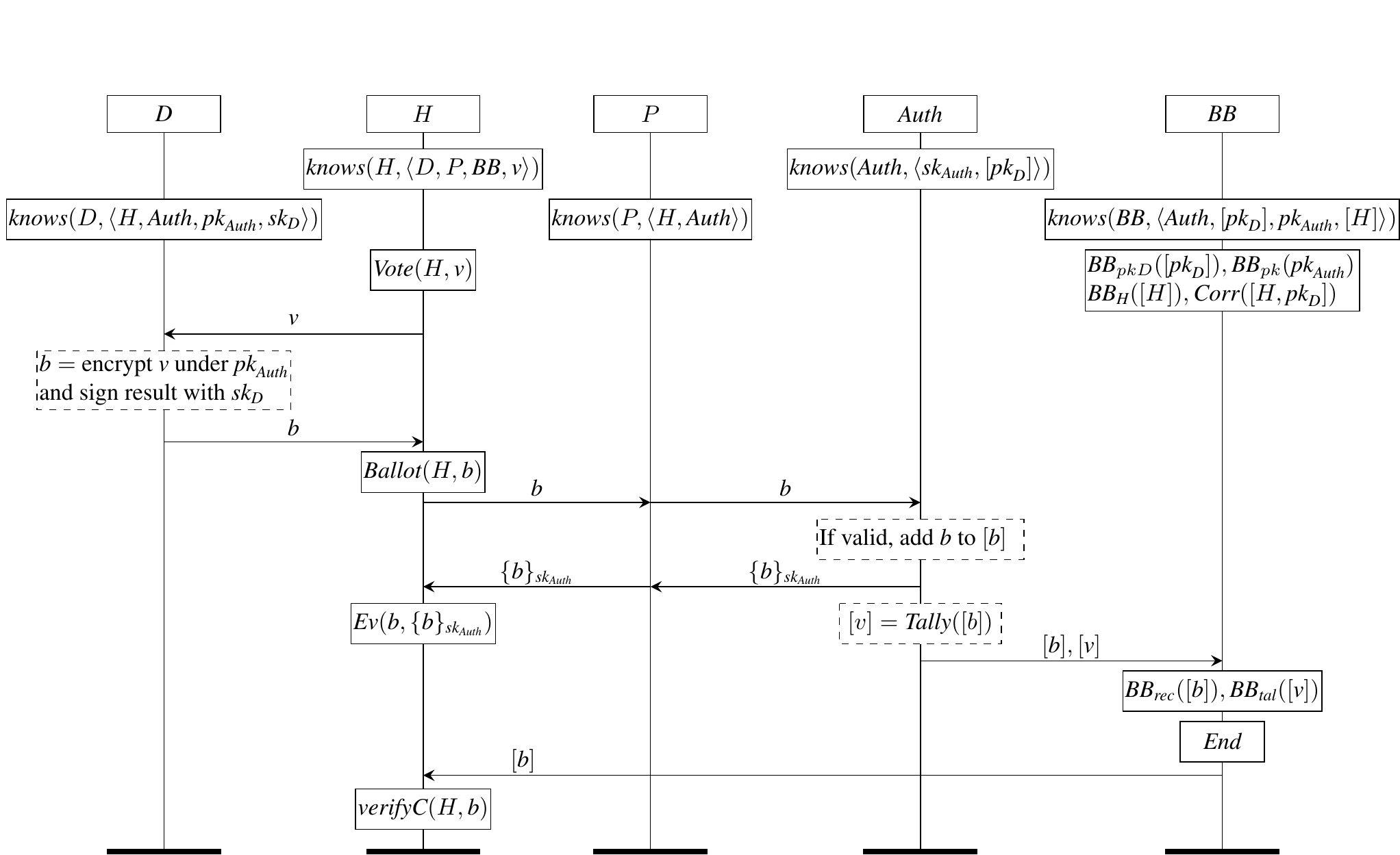}}
 \scalebox{0.9}{ \parbox[c]{\textwidth}{
 \begin{alignat*}{1}
 &\textit{For } \ballot = \perp:\; \verdict(\server, \ballot):=\{ \}.
 \\&
\textit{For }\ballot \neq \perp: \;
\verdict(\server, \ballot):=\{\tr |
(\exists [\ballot], \pkS, \confirmation \sd \BBpkS(\pkS) \in \tr \wedge \Evidence(\ballot,\confirmation) \in \tr 
\wedge \ver(\confirmation, \pkS) = \ballot
\wedge \BBrecorded([\ballot]) \in \tr
\wedge \ballot \notin [\ballot])
\\&\multicolumn{2}{r}{$
\vee 
(\exists [\ballot],[\pkD] \sd
\BBrecorded(\bs)\in \tr 
\wedge \BBpkD([\pkD])\in \tr
\wedge
\textit{ not all ballots in }\bs
\textit{ contain a signature associated with a unique key in }[\pkD])\}
$}
\end{alignat*}
}}
\end{center}
 \caption{
 Simplified protocol specification for \protocolnameInst, without the auditor role and the full function definitions.
 Here $\pkD =\pk(\skD)$, $\pkS=\pk(\skS)$, and $\CastBy(\ballot)= H$ holds iff $ \exists \pk \sd\ver(\ballot,\pk)\neq \perp \wedge \langle H,\pk\rangle \in [H,\pk] \wedge \corresponds([H,\pk])\in \tr$.
The protocol's setup specifies a single agent $\server$, that each voter $H$ is associated with a unique trusted off-line device $D$, and that there is no restriction on the relation between voters $H$ and platforms $P$.
The role for a voter $H$ who abstains consists of receiving the list of recorded ballots from the bulletin board followed by the signal $\verifyNoVote(H,\emptyset)$.
 }\label{fig:protAletheaShort}
\end{figure*}

To demonstrate the applicability of our formalism, we next analyze \protocolnameInst, a standard mixnet-based voting protocol inspired by~\cite{alethea} with a dispute resolution procedure similar to~\cite{sElect}.
In particular, we show how $\verdict$ is instantiated, how the properties $\provoterCast(\server)$ and $\timeliness(\server)$ differ in practice, and that our dispute resolution properties are compatible with standard voting properties, such as verifiability and receipt-freeness.
Due to space constraints, we only describe the protocol's main features here, omitting some details such as the auditor's role and the precise definition of some functions and equations in the term algebra.
For the detailed protocol specification, the properties' formal definitions, and the proofs we refer to
Appendix~\ref{appendix:AnalyzingAmixnetbasedVotingProtocolWithDisputeResolution}

\subsubsection{Topology}
We consider a topology $\topoInst$ that is as $\topo_4$ in Figure~\ref{fig:topo4}, except that there is also a trusted off-line device $D$, which is connected to the voter $H$ by (bidirectional) secure, default channels.
$\topoInst$ specifies reliable channels between $H$ and the platform $P$ and between $P$ and the authority $\server$. Also, $P$ and $\server$ are partially trusted to forward messages and reply to messages, respectively. 
Thus, by Theorem~\ref{theorem:completeChar}, it is possible to achieve \fullTime in the topology $\topoInst$.

\subsubsection{Protocol}\label{subsubsec:mixnet:protocol}
We present the protocol as a \emph{message sequence chart}, where each role is depicted by a vertical \emph{life line} and where the box on top names the role.
A role's life line denotes the role's events, ordered sequentially. 
A role's sent and received messages are depicted on top of arrows that start at the sender and end at the recipient.
Also, we denote explicit signals by solid squares and the roles' internal computations by dashed squares. 

\protocolnameInst's simplified specification is depicted in Figure~\ref{fig:protAletheaShort}.
The protocol's setup specifies that at each point in time, only one election takes place (i.e., there are no parallel sessions) and each voter possesses a unique trusted device $D$ to which he has exclusive access.
All devices are equipped with a unique signing key $\skD$ and the authority with a unique secret key $\skS$.
The corresponding verification keys from the devices $[\pkD]$ are known to $\server$ and $\server$'s public key $\pkS$ is known to all devices. Moreover, all these public keys are 
published on $\BB$ (denoted by the signals $\BBpkD$ and $\BBpkS$, respectively).
Additionally, at the protocol's start $\BB$ knows and publishes the list of eligible voters $\Hs$ (denoted by the signal $\BBH$) and which verification key corresponds to which voter. 
The latter is denoted by the signal $\corresponds([H,\pkD])$, where each pair $\langle H,\pkD \rangle$ in the list denotes that the signing key corresponding to $\pkD$ is installed on $H$'s device.

To vote, a voter $H$ uses his device $D$ to compute the ballot as follows: the vote is encrypted under $\server$'s public key and signed by the device.
Then, the voter casts his ballot by entering it on any platform $P$, which forwards it over the network to $\server$. 
For each received ballot $b$, $\server$ checks $b$'s validity, namely whether $b$ contains a signature corresponding to an eligible voter who has not previously voted. 
If this is the case, $\server$ adds $b$ to the list of recorded ballots $\bs$.
Moreover, as in other protocols~\cite{sElect}, to achieve dispute resolution, $\server$ sends back a confirmation to the voter $H$.
The confirmation consists of $H$'s ballot $b$ signed by $\server$ and serves as evidence that $b$ was indeed received by the authority.
The voter keeps this confirmation as evidence for later disputes (indicated by the signal $\Evidence$).

After the voting phase, $\server$ computes the tally from the recorded ballots $\bs$.
For this, a standard mixnet is used to decrypt the ballots. This procedure has the properties that no one can learn the correspondence between the encrypted ballots and the decrypted votes. Nevertheless, the mixnet produces evidence, which is published by $\server$ on the bulletin board, that allows everyone to verify that the tally was computed correctly.
We describe the detailed functions and equations modeling the $\Tally$ function in 
Appendix~\ref{subsubsec:newFunctionsAndEquations}.
Also, we describe there the detailed information that is produced by the mixnet and published on $\BB$ and how an auditor inspects this information to verify the tally.

Among other information, $\server$ publishes on $\BB$ the recorded ballots $\bs$ and the votes in the final tally $\vs$, as shown in Figure~\ref{fig:protAletheaShort}.
This allows a voter to read the recorded ballots on $\BB$ and verify that his ballot is included in this list.

A voter who abstains does not send any messages. After the results are published, he reads the list of recorded ballots $\bs$ on $\BB$ and believes at that step that no ballot should be recorded for him, which is denoted by the signal $\verifyNoVote(H,\emptyset)$.

We complete the protocol's specification with the definitions of the function $\CastBy$ and the dispute resolution procedure $\verdict$.
$\CastBy$ specifies that a ballot $b$ is considered to be cast by the voter $H$ if the ballot's signature can be verified with the verification key that is associated with $H$. 
$\verdict$ specifies that $\server$ is considered dishonest in all traces where (a) some agent possesses evidence consisting of a ballot $b$ signed by $\server$ that is not included in the recorded ballots $\bs$ on the bulletin board or (b) not all published recorded ballots $\bs$ contain a signature of a unique eligible voter. 
$\CastBy$ is defined in Figure~\ref{fig:protAletheaShort}'s caption and the description of $\verdict$ is given in Figure~\ref{fig:protAletheaShort}, although we omit here the details of how we model (b).

\subsubsection{Dispute resolution}\label{subsubsec:mixvote:DR}
Intuitively, by the channel and trust assumptions, each voter who casts a ballot $b$ receives, before the election's end, a confirmation.
As this confirmation serves as evidence that $b$ must be on $\BB$, $\provoterCast(\server)$ and $\timeliness(\server)$ hold.
Furthermore, since no one can forge $\server$'s signature, for a ballot $b$ that was not actually received by $\server$ no one can produce (false) evidence 
that $b$ should be on $\BB$.
Thus, $\server$ cannot be falsely convicted and $\proauth(\server)$ holds too.
Moreover, $\uniqueness(\server)$ holds because, when $\verdict$ does not hold in an execution, all recorded ballots are signed (and thus were sent) by a unique eligible voter.
In particular, $\uniqueness(\server)$ implies $\provoterAbstain(\server)$ in \protocolnameInst.

To understand the difference between $\provoterCast(\server)$ and $\timeliness(\server)$,
take a topology $\topoInst'$ equal to $\topoInst$ except that $\server$ is untrusted.
Assume for simplicity that a voter can interpret whether $\server$'s signature on the confirmation is valid. In reality, this would require an additional protocol step where the voter uses a device.
When the protocol is run in $\topoInst'$, it satisfies $\provoterCast(\server)$, as a voter only proceeds with his verifiability check when he has previously received a valid confirmation that convinces everyone that his ballot must be recorded.
However, $\timeliness(\server)$ is violated as $\server$ may never reply with a valid confirmation and thus block a voter.
Consequently, there is an unresolved dispute where an outside observer cannot tell whether a voter did not cast a ballot or the authority did not send a confirmation.
In contrast, when the protocol is run in topology $\topoInst$, $\server$ always sends a timely response and such disputes do not occur.

\subsubsection{Standard voting properties}
In addition to the dispute resolution properties, we prove in
Appendix~\ref{appendix:AnalyzingAmixnetbasedVotingProtocolWithDisputeResolution}
that \protocolnameInst satisfies
end-to-end verifiability, consisting of individual verifiability and \emph{tallied-as-recorded}, as well as \emph{eligibility verifiability}~\cite{kremerVerif}.
Tallied-as-recorded and eligibility verifiability are two universal verifiability properties that respectively denote that an auditor can verify that the recorded ballots are correctly counted in the final tally and that each vote in the final tally was cast by a unique eligible voter.
We also prove that \protocolnameInst satisfies \emph{receipt-freeness}~\cite{delaune}, which denotes that a voter cannot prove to the adversary how he voted, even when he provides the adversary with all secrets that he knows.
Intuitively, receipt-freeness holds 
because
the adversary cannot access the voter's device $D$. 
Moreover, the evidence used for disputes only contains the ballot and does not reveal the underlying (encrypted) vote. 

\subsubsection{Proofs}
\label{subsubsec:mixvote:proofs}
We prove in 
Appendix~\ref{appendix:mixvote:analysis}
and by the Tamarin files in~\cite{tamarinfiles}
that \protocolnameInst satisfies all above mentioned properties when run in the topology $\topoInst$.
In particular, we establish most of the properties by automatically proving them for one voter who casts a ballot in Tamarin and by proving them for an arbitrary number of voters by pen-and-paper proofs.
The only exceptions are: receipt-freeness, which we prove by Tamarin's built in support for observational equivalence~\cite{tamarinDiff};
$\provoterAbstain(\server)$, which we deduce (by hand) 
from $\uniqueness(\server)$ using Theorem~\ref{theorem:uniquenessImpliesDR};
and end-to-end verifiability which we deduce (by hand) from individual verifiability and tallied-as-recorded.
 
\section{Related work}\label{sec:relatedWork}

\subsection{Dispute resolution in poll-site voting protocols}

The idea of dispute resolution has been informally considered for \emph{poll-site} voting protocols.
In~\cite{vvote}, the property considered is called \emph{non-repudiation} and requires that failures \emph{``can not only be detected, but (in most cases) demonstrated''} and that no false convictions can be made.
\cite{BMDsStark} informally considers the properties \emph{contestability} and \emph{defensibility}, which are similar to our dispute resolution properties in that they also protect the honest voters and the honest authority.
Contestability requires that some guarantees hold for a voter when he starts the voting process at a polling station. In contrast, our properties $\timeliness(\server)$ and $\provoterCast(\server)$ are also suitable for remote settings and respectively make guarantees once a voter casts his ballot and believes that it should be recorded.
Moreover,~\cite{BMDsStark}'s definitions are informal and they do not consider timeliness.

In most poll-site voting protocols that consider dispute resolution,
voters receive a confirmation as evidence that their ballot was accepted by the authority~\cite{Forensics,vvote2,scantegrity2b,scantegrity2a,vvote,audiotegrity,scantegrity3}.
In some protocols~\cite{vvote2,vvote}, this confirmation contains the authority's digital signature. 
In the protocols based on Scantegrity~\cite{scantegrity2b,scantegrity2a,audiotegrity,scantegrity3} the confirmation consists of a code that is (physically) hidden on the ballot by invisible ink and revealed when a voter marks his choice.
A voter's knowledge of a valid code serves as evidence that he voted for a candidate. 
Thus, when a wrong ballot is recorded, a voter can prove the authority's dishonesty by revealing the code.

Compared to remote voting settings, poll-site protocols profit from the fact that on-site witnesses can observe certain actions.
For example, if voters are repeatedly prevented from casting their ballots, this is visible to other voters and auditors in the polling station.
Some protocols~\cite{audiotegrity} even explicitly 
state that voters should publicly declare some decisions before entering them on the voting machine to avoid disputes regarding whether the voting machine correctly followed their instructions.
Our notion of undeniable channels allows one to formally consider such assumptions during protocol analysis.

\subsection{Dispute resolution in remote voting protocols}
\emph{Remotegrity}~\cite{remotegrity} is a remote voting protocol based on Scantegrity, where paper sheets are sent to the voters by postal mail and ballots are cast over the Internet. 
As with Scantegrity~II and III~\cite{scantegrity2b,scantegrity2a,scantegrity3}, to achieve dispute resolution some codes on these sheets are obscured by a scratch-off surface. 
If a voter detects a (valid) ballot that is incorrectly recorded for him, he can show to anyone that he has not yet scratched off the relevant codes on his sheets and thus the authority must have falsely recorded this ballot.

\cite{remotegrity} discusses several dispute scenarios with respect to whether a ballot is recorded correctly.
However, it is stated that \emph{``The [authority] can always force a denial-of-service [..] What Remotegrity does not allow is the [authority] to fully accept (i.e., accept and lock) any ballot the voter did not cast without the voter being able to dispute it.''} 
Thus, the focus is on disputes \stwo in Figure~\ref{fig:disputeSituations}, while timeliness in disputes \sone is not further explored.
Moreover, the considered properties as well as the assumed setting are not specified precisely and thus the properties cannot be proven. 
In contrast, our model enables specifying detailed adversary and system assumptions and provides definitions of dispute resolution properties that can be formally analyzed.

\subsection{Accountability}
\label{subsec:thedisputersolutionproperty}

Our dispute resolution properties are closely related to different notions of \emph{accountability}~\cite{accountabilityBruni,accountabilityCSF19,accountability}.
Both accountability and our properties formalize how misbehaving protocol participants are identified.
While the accountability definitions are generic and allow one to blame different agents in different situations, we focus on understanding what disputes and properties are relevant for voting.

Two accountability definitions have been instantiated for voting protocols.
First, accountability due to \accountabilityCite was instantiated for \emph{Bingo Voting}~\cite{bingoVoting} in~\cite{accountability}, for \emph{Helios}~\cite{helios} in~\cite{clash}, and for \emph{sElect}~\cite{sElect}.
These instantiated notions of accountability state that when a defined goal is violated, then some (dishonest) agents can be blamed by a \emph{judge}.
A judge may blame multiple parties. 
As a result, in~\cite{clash} accountability does not guarantee an unambiguous verdict when a voter claims that his ballot is incorrectly recorded. 
That is, the property does not guarantee the resolution of such disputes even when the voter is honest.
The same holds in~\cite{accountability} and~\cite{sElect} for disputes where a voter claims that he did not receive a required confirmation.
To avoid ambiguous verdicts, \cite{accountability} proposes an alternative accountability property where voters' claims that they did not receive a required confirmation are just ignored.
However, this property does not guarantee that the authority is blamed in all situations where an (honest) voter's ballot is not recorded correctly and dispute resolution does not hold.

Second, accountability due to \emph{Bruni et al.}~\cite{accountabilityBruni} has been instantiated for Bingo Voting in~\cite{accountabilityBruni}.
In this work, \emph{accountability tests} decide whether a given agent
should be blamed. However, the accountability test takes as input a ballot and a confirmation that the voter received 
when casting his ballot. 
Thus disputes where a voter claims that he cannot receive a confirmation are not considered at all.

In contrast to these two accountability notions, we also consider and resolve disputes where a voter claims that he did not receive a required response from $\server$ after casting the ballot by the property $\timeliness(\server)$.
Moreover, our topology characterization allows us to quickly assess when given assumptions are insufficient to satisfy $\timeliness(\server)$.

\subsection{Other related properties}
\emph{Collection accountability}~\cite{pubEvidence} states that when a vote is incorrectly collected, the voter should be provided with evidence to convince an \emph{``independent party''} that this is the case, but it has neither been formally defined nor analyzed.
\emph{Dispute freeness}~\cite{framework} states that there is never a dispute.
This property is considered in voting protocols where voters are modeled as machines that conduct an election by engaging in a multi-party protocol~\cite{Schoenmakers99,selfTallying} and is thus inappropriate for large scale elections where voters must be assumed to have limited computational capabilities.
Finally, the FOO protocol~\cite{foo} allows voters to claim that something went wrong.
However, without additional assumptions, FOO does not satisfy our dispute resolution properties. 
In particular, the signed ballot a voter receives does not prove that the \emph{counter}, who is responsible for tallying, has received the ballot.

\section{Conclusion}\label{sec:conclusion}

Dispute resolution is an essential ingredient for trustworthy elections and worthy of a careful, formal treatment.
Based on a systematic analysis of disputes, we proposed new dispute resolution properties and introduced timeliness as an important aspect thereof. 
We fully characterized all topologies that achieve timeliness.
This provides a formal account for the intuition that timeliness requires strong assumptions.
For example, it is not achievable in standard remote voting settings where a network adversary can simply drop messages.

While we have focused on necessary assumptions for dispute resolution, in real elections there are other properties, notably privacy, which may require other assumptions.
As future work, we would like to investigate how our topology hierarchy must be adapted for these properties and to characterize the required assumptions for them.
The combination of such results with our characterization could lead to new insights about the possibility of achieving different properties simultaneously.
Furthermore, such combined results could be a starting point to identify the topologies enabling all properties required in voting; this would help in election design to quickly assess the minimal required setups.

 \bibliographystyle{plain} 
\bibliography{referencesDispRes}

\appendix

\subsection{Proofs from Section~\ref{sec:Timeliness}}
\label{appendix:additionalProofsTimeliness}
We present additional proofs and lemmas for proving the claims from Section~\ref{sec:Timeliness}.
\subsubsection{Topology hierarchy}
\label{appendix:proofTopologyHierarchy}

To prove Lemma~\ref{lemma:possAndImposs}, we will argue that if two traces or two sets of traces are ``similar enough'', then they either both satisfy our considered dispute resolution properties or both violate them. 
To help with this reasoning, we first define a notion of similarity and show two auxiliary lemmas.
\begin{mydef}\label{def:drSimilar}
Two traces $\tr_1$ and $\tr_2$ are \emph{dispute resolution equal}, denoted by $\drequal(\tr_1,\tr_2)$, iff for any voter $H$ and ballot $\ballot$ and for the authority $\server$ it holds that
\begin{alignat*}{1}
& (\exists \tr_1', \tr_1''\sd \tr_1 = \tr_1' \cdot \tr_1'' \wedge \Ballot(H,\ballot) \in\tr_1' \wedge \End\in \tr_1''
 \\\multicolumn{2}{r}{$ \Leftrightarrow$}
 \\\multicolumn{2}{r}{$ 
\exists \tr_2', \tr_2''\sd \tr_2= \tr_2' \cdot \tr_2'' \wedge \Ballot(H,\ballot) \in\tr_2' \wedge
 \End\in \tr_2'' )
 $}
 \\\multicolumn{2}{r}{$
 \wedge 
(\honest(\server) \in \tr_1 \Leftrightarrow \honest(\server) \in \tr_2)
 $}
 \\\multicolumn{2}{r}{$
 \wedge
(\pubTrace(\tr_1) = \pubTrace(\tr_2)).
 $}
\end{alignat*}
The set of traces $\TR(\prot_1,\topo_1)$ is \emph{dispute resolution similar} to the set of traces $\TR(\prot_2,\topo_2)$, denoted by $\drsimilar(\TR(\prot_1,\topo_1),\TR(\prot_2,\topo_2))$, iff
\begin{alignat*}{1}
 & \forall \adv \in \{\hShH,\hSmH,\mShH\},\tr_1 \in \TR(\prot_1,\advModel{\topo_1}{\adv}) \sd 
 \\\multicolumn{2}{r}{$
 \exists \tr_2 \in \TR(\prot_2,\advModel{\topo_2}{\adv}) \sd
 \drequal(\tr_1,\tr_2).
 $}
 \end{alignat*}
\end{mydef}
The following auxiliary lemma states that when two traces are dispute resolution equal, then either both satisfy $\timeliness(\server)$, respectively $\proauth(\server)$, or both do not satisfy it. 
\begin{lemma}\label{lemma:drequal}
 For two traces $\tr_1$ and $\tr_2$, where $\drequal(\tr_1,\tr_2)$, it holds that 
 \begin{alignat*}{1}
 & (\tr_1 \in \timeliness(\server)\Leftrightarrow \tr_2 \in \timeliness(\server) )
 \\\multicolumn{2}{r}{$
 \wedge (\tr_1 \in \proauth(\server) \Leftrightarrow \tr_2 \in \proauth(\server)).
 $}
 \end{alignat*}
\end{lemma}
\begin{proof}
We prove each of the conjuncts separately.
Consider two traces $\tr_1$ and $\tr_2$ such that $\drequal(\tr_1,\tr_2)$.
By Definition~\ref{def:proTime}, $\timeliness(\server)$ holds in a trace $\tr$ iff
\begin{alignat*}{1}
 &\forall H,\ballot \sd \exists \tr',\tr'' \sd
 \tr = \tr' \cdot \tr'' \wedge
 \Ballot(H,\ballot) \in \tr'
 \\&\multicolumn{2}{r}{$ 
 \wedge \End \in \tr''
 \implies (\exists \bs \sd \BBrecorded(\bs) \in \tr' \wedge \ballot \in \bs )
 $}
 \\&\multicolumn{2}{r}{$ 
 \vee (\tr \in \verdict(\server,\ballot)).
 $} 
 \end{alignat*}
This formula is of the form
\begin{alignat*}{1}
 & \forall H, b \sd \exists \tr',\tr''\sd 
 \\\multicolumn{2}{r}{$ 
 A(H,b,\tr',\tr'')\implies B(H,b,\tr',\tr'') \vee C(H,b).
 $}
\end{alignat*}
In the following we refer to these predicates simply by $A$, $B$, and $C$.
As only the truth values of the predicates $A$ and $B$ depend on the traces $\tr'$ and $\tr''$, we consider the formula $(\exists \tr',\tr''\sd A \implies B) \vee C$, which we will call $F$, and show that it holds for exactly the same $H$ and $\ballot$ in $\tr_1$ and $\tr_2$.
 Concretely, we show that if $F$ holds for a given $H$ and $\ballot$ in $\tr_1$, then it also holds for $H$ and $\ballot$ in $\tr_2$. 
 As $\drequal(\tr_1,\tr_2)$ is symmetric, the same arguments can be applied to show that if $F$ holds for a given $H$ and $\ballot$ in $\tr_2$, then it also holds in $\tr_1$. 
 As this holds for all $H$ and $b$, it follows that $\tr_1 \in \timeliness(\server)$ iff $\tr_2 \in \timeliness(\server)$.
 
 Let $H$ and $b$ be such that $F$ holds in $\tr_1$.
 We make a case distinction for the different truth values of $A$, $B$, and $C$.
 
\textit{Case 1):} 
 Let $H$ and $b$ be such that $(\exists \tr_1',\tr_1''\sd A \implies B)$ does not hold in~$\tr_1$.
 As by assumption $F$ holds, it must be the case that $C$ holds, that is $\tr_1 \in \verdict(\server,\ballot)$. 
 By Definition~\ref{def:drSimilar}, $\pubTrace(\tr_1) = \pubTrace(\tr_2)$ and by Requirement~\ref{requirement:verdict} $\pubTrace(\tr_1) = \pubTrace(\tr_2)$ implies that $\tr_1 \in \verdict(\server,\ballot)$ iff $\tr_2 \in \verdict(\server,\ballot)$. Thus, $\tr_2 \in \verdict(\server,\ballot)$ and $C$ also holds in $\tr_2$ for $H$ and $\ballot$.
 It follows that $F$ holds in $\tr_2$ for $H$ and $b$.
 
 \textit{Case 2):}
 Let $H$ and $b$ be such that $(\exists \tr_1',\tr_1''\sd A \implies B)$ holds in $\tr_1$, and $A$ holds and $B$ holds.
 That is, $H$ and $b$ are such that there exist two traces $\tr_1'$ and $\tr_1'' $, where 
 $$\tr_1 = \tr_1' \cdot \tr_1''\wedge \Ballot(H,\ballot) \in \tr_1' \wedge\End \in \tr_1''$$
 and such that there exists a list of ballots $\bs$ for which 
 $$\BBrecorded(\bs) \in \tr_1' \wedge \ballot \in \bs.$$
 By Definition~\ref{def:drSimilar}, there also exist two
 traces $\tr_2^1$ and $\tr_2^2 $, such that 
 $$\tr_2 = \tr_2^1 \cdot \tr_2^2\wedge\Ballot(H,\ballot) \in \tr_2^1\wedge\End \in \tr_2^2.$$
 Moreover, by Definition~\ref{def:drSimilar}, $\pubTrace(\tr_1) = \pubTrace(\tr_2)$.
 As $\BBrecorded$ is a signal in the publicly observable trace and as $\BBrecorded(\bs)$ with $\ballot \in \bs$ is recorded before $\End$ in~$\tr_1$, it holds that for two traces $\tr_2^3$ and $\tr_2^4$
 $$\tr=\tr_2^3\cdot \tr_2^4\wedge \BBrecorded(\bs)\in \tr_2^3\wedge \ballot\in \bs \wedge \End \in \tr_2^4.$$
 Thus, as both the signals $\Ballot(H,\ballot)$ and $\BBrecorded(\bs)$ are recorded in $\tr_2$ before $\End$ there exist two traces $\tr_2'$ and $\tr_2''$ (where $\tr_2'$ can be chosen to be the larger trace from $\tr_2^1$ and $\tr_2^3$) such that 
 \begin{alignat*}{1}
 & \tr_2 = \tr_2' \cdot \tr_2''\wedge\Ballot(H,\ballot) \in \tr_2'\wedge\End \in \tr_2''
 \\\multicolumn{2}{r}{$
 \wedge \BBrecorded(\bs)\in \tr_2'\wedge \ballot\in \bs.
 $}
 \end{alignat*}
 Hence, $(\exists \tr_2',\tr_2''\sd A \implies B)$ also holds in $\tr_2$ and $F$ is satisfied for $H$ and $b$.
 
 \textit{Case 3):}
 Let $H$ and $b$ be such that $(\exists \tr_1',\tr_1''\sd A \implies B)$ holds in $\tr_1$ and $A$ does not hold.
 That is, $$\exists \tr_1',\tr_1''\sd \neg(\tr_1 = \tr_1' \cdot \tr_1'' \wedge
 \Ballot(H,\ballot) \in \tr_1' \wedge \End \in \tr_1'').$$
 We choose $\tr_2'=\tr_2''=\emptyset$ and it holds that 
 $$ \neg(\tr_2 = \tr_2' \cdot \tr_2'' \wedge \Ballot(H,\ballot) \in \tr_2' \wedge \End \in \tr_2'').$$
In particular, empty traces cannot contain a signal, thus the second and third conjunct are always false. 
Thus, it follows that 
$$\exists \tr_2',\tr_2''\sd \neg(\tr_2 = \tr_2' \cdot \tr_2'' \wedge \Ballot(H,\ballot) \in \tr_2' \wedge \End \in \tr_2'')$$
and hence $A$ is false in $\tr_2$. 
It follows that $(\exists \tr_2',\tr_2''\sd A \implies B)$ and therefore $F$ hold in $\tr_2$ for $H$ and $b$.

Hence, we showed that in all cases where $F$ holds for a $H$ and $b$ in $\tr_1$, $F$ also holds for $H$ and $b$ in $\tr_2$.
 \\\\
 We next prove the lemma's second conjunct.
 Assume two traces $\tr_1$ and $\tr_2$ such that $\drequal(\tr_1,\tr_2)$ and assume $\tr_1 \in \proauth(\server)$.
 We show that this implies $\tr_2 \in\proauth(\server)$.
 As $\drequal$ is symmetric, the same arguments can be applied to show that $\tr_1 \in\proauth(\server)$ follows from $\tr_2 \in\proauth(\server)$.
 
 We distinguish two cases for which $\tr_1 \in \proauth(\server)$.
 First, let $\honest(\server)\notin \tr_1$.
 Definition~\ref{def:drSimilar} implies $\honest(\server)\notin \tr_2$, as $\honest(\server)\in \tr_2$ would require $\honest(\server)\in \tr_1$, which is a contradiction.
 From $\honest(\server)\notin \tr_2$, it follows that $\tr_2 \in \proauth(\server)$ by Definition~\ref{def:proauth}.
 Second, let $\honest(\server)\in \tr_1$ and assume that it holds for all ballots $\ballot$ that $\tr_1 \notin \verdict(\server,\ballot)$.
 Definition~\ref{def:drSimilar} implies $\honest(\server)\in \tr_2$. Moreover, by Definition~\ref{def:drSimilar} it holds that $\pubTrace(\tr_1) = \pubTrace(\tr_2)$ and, by Requirement~\ref{requirement:verdict}, $\pubTrace(\tr_1) = \pubTrace(\tr_2)$ implies that $\tr_1 \in \verdict(\server,\ballot)$ iff $\tr_2 \in \verdict(\server,\ballot)$.
 Thus, it cannot hold that there exists a ballot $\ballot$ for which $\tr_2 \in \verdict(\server,\ballot)$, as this would require $\tr_1 \in \verdict(\server,\ballot)$, which is a contradiction.
 We thus conclude that for all ballots $\ballot$, $\tr_2 \notin \verdict(\server,\ballot)$, and thus by Definition~\ref{def:proauth} $\tr_2 \in\proauth(\server)$.
\end{proof}

 Using the above lemma, we show the following lemma.

\begin{lemma}\label{lemma:drsimilar}
For two sets of traces $\TR(\prot_1,\topo_1)$ and $\TR(\prot_2,\topo_2)$, where $\drsimilar(\TR(\prot_1,\topo_1),\TR(\prot_2,\topo_2))$, it holds that
 \begin{alignat*}{1}
& \TR(\prot_2, \advModel{\topo_2}{\hShH}) \subseteq \timeliness(\server) \, \cap\,\proauth(\server)
\\\multicolumn{2}{r}{$
\wedge \TR(\prot_2, \advModel{\topo_2}{\mShH}) \subseteq \,\timeliness(\server)
$}
\\\multicolumn{2}{r}{$
\wedge \TR(\prot_2, \advModel{\topo_2}{\hSmH}) \subseteq\,\proauth(\server)
\implies
$}
\\\multicolumn{2}{r}{$
 \TR(\prot_1, \advModel{\topo_1}{\hShH}) \subseteq \timeliness(\server) \, \cap\,\proauth(\server)
$}
\\\multicolumn{2}{r}{$
\wedge \TR(\prot_1, \advModel{\topo_1}{\mShH}) \subseteq \,\timeliness(\server)
$}
\\\multicolumn{2}{r}{$
\wedge \TR(\prot_1, \advModel{\topo_1}{\hSmH}) \subseteq\,\proauth(\server).
$}
 \end{alignat*} 
\end{lemma}
\begin{proof}
 Assume that 
 \begin{alignat*}{1}
 &\TR(\prot_2, \advModel{\topo_2}{\hShH}) \subseteq \timeliness(\server) \cap\proauth(\server) 
 \\&\multicolumn{2}{r}{$ 
 \wedge
 \TR(\prot_2, \advModel{\topo_2}{\mShH}) \subseteq \timeliness(\server) 
 $} 
 \\&\multicolumn{2}{r}{$ 
 \wedge \TR(\prot_2, \advModel{\topo_2}{\hSmH}) \subseteq \proauth(\server)
 $} 
 \end{alignat*}
 and that $\drsimilar(\TR(\prot_1,\topo_1),\TR(\prot_2,\topo_2))$.
 
Let $\tr_1$ be a trace in $\TR(\prot_1,\advModel{\topo_1}{\hShH})$.
By Definition~\ref{def:drSimilar}, there exists a trace $\tr_2$ in $\TR(\prot_2,\advModel{\topo_2}{\hShH})$ such that $\drequal(\tr_1,\tr_2)$.
By Lemma~\ref{lemma:drequal}, it holds that 
$\tr_1 \in\timeliness(\server) \Leftrightarrow \tr_2 \in \timeliness(\server) $ and $\tr_1 \in \proauth(\server) \Leftrightarrow \tr_2 \in \proauth(\server)$.
Since, by assumption, $\TR(\prot_2, \advModel{\topo_2}{\hShH}) \subseteq \timeliness(\server) \cap\proauth(\server) $ and thus $\tr_2 \in \timeliness(\server) \cap\proauth(\server) $, it follows that
$\tr_1 \in \timeliness(\server) \cap\proauth(\server) $.
As $\tr_1$ is an arbitrary trace in $\TR(\prot_1,\advModel{\topo_1}{\hShH})$, it follows that $\TR(\prot_1,\advModel{\topo_1}{\hShH})\subseteq \timeliness(\server) \cap\proauth(\server) $.

Let $\tr_1$ be a trace in $\TR(\prot_1,\advModel{\topo_1}{\mShH})$.
By Definition~\ref{def:drSimilar}, there exists a trace $\tr_2$ in $\TR(\prot_2,\advModel{\topo_2}{\mShH})$ such that $\drequal(\tr_1,\tr_2)$.
By Lemma~\ref{lemma:drequal}, it holds that 
$\tr_1 \in\timeliness(\server) \Leftrightarrow \tr_2 \in \timeliness(\server) $.
Since, by assumption, $\TR(\prot_2, \advModel{\topo_2}{\mShH}) \subseteq \timeliness(\server)$ and thus $\tr_2 \in \timeliness(\server)$, it follows that $\tr_1 \in \timeliness(\server)$.
As $\tr_1$ is an arbitrary trace in $\TR(\prot_1,\advModel{\topo_1}{\mShH})$, it follows that $\TR(\prot_1,\advModel{\topo_1}{\mShH})\subseteq \timeliness(\server) $.

Finally, let $\tr_1$ be a trace in $\TR(\prot_1,\advModel{\topo_1}{\hSmH})$.
By Definition~\ref{def:drSimilar}, there exists a trace $\tr_2$ in $\TR(\prot_2,\advModel{\topo_2}{\hSmH})$ such that $\drequal(\tr_1,\tr_2)$.
By Lemma~\ref{lemma:drequal}, it holds that 
$\tr_1 \in \proauth(\server) \Leftrightarrow \tr_2 \in \proauth(\server)$.
Since, by assumption, $\TR(\prot_2, \advModel{\topo_2}{\hSmH}) \subseteq \proauth(\server) $ and thus $\tr_2 \in \proauth(\server) $, it follows that $\tr_1 \in \proauth(\server) $.
As $\tr_1$ is an arbitrary trace in $\TR(\prot_1,\advModel{\topo_1}{\hSmH})$, it follows that $\TR(\prot_1,\advModel{\topo_1}{\hSmH})\subseteq\proauth(\server)$.
\end{proof}

Using the above lemmas, we next prove Lemma~\ref{lemma:possAndImposs}.
\begin{proof}[Proof of Lemma~\ref{lemma:possAndImposs}]
We define $\topo_1 \smallertopoNotEqual \topo_2:= \topo_1 \smallertopo \topo_2 \wedge \topo_1 \neq \topo_2$.
Let $\topos=(V_S,E_S,\trust_S,\channel_S)$, $\topoi=(V_I,E_I,\trust_I,\channel_I)$, and $\topoi \smallertopo \topos$ and let $\prot$ be an arbitrary protocol such that $\fullTime(\prot,\topoi)$.
By Definition~\ref{def:useful},
$\fullTime(\prot,\topoi)= 
\TR(\prot, \advModel{\topoi}{\hShH}) \subseteq\, \protime(\server) \, \cap\,\proauth(\server)$
$ \wedge \,\TR(\prot, \advModel{\topoi}{\mShH}) \subseteq \,\protime(\server) $
$\wedge \,\TR(\prot, \advModel{\topoi}{\hSmH}) \subseteq\, \proauth(\server) $
$\wedge\, \TR(\prot, \advModel{\topoi}{\hShH})\, \cap \,\functional\, \neq \emptyset$.
For two channel or trust types $x_I$ and $x_S$, we say that $x_S$ is \emph{minimally stronger} than $x_I$ iff $x_I \smallertopoNotEqual x_S \wedge \neg \exists x_M \sd x_I \smallertopoNotEqual x_M \smallertopoNotEqual x_S$.
We say that two topologies \emph{differ by one attribute} in the following cases:
one topology contains exactly one channel or role that does not exist in the other,
a role that occurs in both topologies has a minimally stronger trust type in one of the topologies, 
or a channel that occurs in both topologies has a minimally stronger channel type in one of the topologies.
We distinguish three cases how $\topoi$ and $\topos$, where $\topoi \smallertopo \topos$, relate: (1) the topologies are equal, (2) they differ only by one attribute, and (3) they differ by several attributes.
Case (1), where $\topos = \topoi$, is trivial and we look at the other cases. 

 \textit{Case (2): $\topoi \smallertopoNotEqual \topos 
 \wedge \neg \exists \topom \sd 
 \topoi \smallertopoNotEqual \topom \smallertopoNotEqual \topos$.}
 
 $\topos$ can differ from $\topoi$ (i) because there is one more (default and insecure) channel or (untrusted) role, (ii) because one of the channel types is stronger, or (iii) because one of the trust types is stronger.
 We further distinguish between these three cases in (2i)--(2iii).

 \textit{Case (2i):} 
 $(\exists (v_a,v_b). (v_a,v_b)\in E_S \wedge (v_a,v_b) \notin E_I \wedge \channel_S(v_a,v_b) = \, \insecChanDefault)
 \vee (\exists v. v\in E_S \wedge v \notin E_I \wedge \trust_S(v) = \textit{untrusted})$.

By assumption, $\fullTime(\prot, \topoi)$.
Assume that $\prot$ only makes use of the vertices and edges in $\topoi$. 
This is without loss of generality because otherwise we can define a protocol $\prot'$, which is as $\prot$ but does not make use of the vertices and edges not contained in $\topoi$ and for which $\fullTime(\prot', \topoi)$. 

We argue that the traces resulting from $\prot$ run in $\topos$ are dispute resolution similar to the traces resulting from $\prot$ run in $\topoi$, i.e., $\drsimilar(\TR(\prot, \topos),\TR(\prot, \topoi))$.
First, when $\prot$ is run in $\topos$, the honest agents follow the protocol and never send or receive anything on $(v_a,v_b)$.
Second, when $v_a$ is under the adversary's control, the adversary can send messages on $(v_a,v_b)$. 
However, the adversary sending messages on an insecure channel does not change the fact that the 
resulting set of traces is dispute resolution similar to the original set $\TR(\prot, \topoi)$ (it does not change the publicly observable trace, nor the signals $\Ballot$, $\End$, and $\honest$).
Moreover, when $v_b$ is honest, it ignores all incoming messages and when $v_b$ is controlled by the adversary too, her knowledge does not differ compared to when the channel $(v_a,v_b)$ does not exist.
Thus, the adversary cannot learn, construct, and send messages that were not possible in $\TR(\prot, \topoi)$.
Finally, as the protocol does not specify any behavior for role $v$, no honest agent will instantiate this role and, as there is no initial knowledge specified for $v$, the adversary cannot learn any new messages even if she simulates such a role. 
We conclude that $\drsimilar(\TR(\prot, \topos),\TR(\prot, \topoi))$ and thus, by Lemma~\ref{lemma:drsimilar} and $\fullTime(\prot, \topoi)$, it follows that $\TR(\prot, \topos)$ satisfies $\timeliness(\server)$ and $\proauth(\server)$ in the required adversary models.
\\
Moreover, the trace that satisfies \functional in $\TR(\prot, \topoi)$ (by assumption such a trace exists) also satisfies \functional in $\TR(\prot, \topos)$, as the same trace is valid even when there is an additional (unused) channel or role.

 \textit{Case (2ii): $\exists (v_a,v_b) \in E_I. \channel_I (v_a,v_b)\smallertopo\channel_S(v_a,v_b) $.}
 
 We consider the following cases how a channel in $\topos$ is minimally stronger than the same channel in $\topoi$.
 \begin{description}
 \item [(a)]$\exists x \sd x \in \{\default,\reliable,\observable\}$
 \\ $\wedge\, \channel_S(v_a,v_b) =\, \secChan{x} \wedge \,\channel_I(v_a,v_b)=\,\authChan{x}$
 \item [(b)]$ \exists x \sd x \in\{\default,\reliable,\observable\}$
 \\ $\wedge\, \channel_S(v_a,v_b) =\, \secChan{x} \wedge\, \channel_I(v_a,v_b)=\,\confChan{x}$
 \item [(c)]$ \exists x \sd x \in \{\default,\reliable,\observable\}$
 \\ $\wedge\, \channel_S(v_a,v_b) =\, \authChan{x} \wedge \,\channel_I(v_a,v_b)=\,\insecChan{x}$
 \item [(d)]$ \exists x \sd x\in \{\default,\reliable,\observable\}$
 \\ $\wedge\, \channel_S(v_a,v_b) =\, \confChan{x} \wedge \,\channel_I(v_a,v_b)=\,\insecChan{x}$
 \item [(e)]$ \exists \chan{} \sd \chan{} \in \{ \secChan{}, \authChan{}, \confChan{}, \insecChan{}\}$
 \\$ \wedge \, \channel_S(v_a,v_b) =\, \chan{\reliable} \wedge \,\channel_I(v_a,v_b)=\,\chan{\default}$
 \item [(f)] $\exists \chan{} \sd \chan{} \,\in \{ \secChan{}, \authChan{}, \confChan{}, \insecChan{}\}$
 \\$ \wedge \, \channel_S(v_a,v_b) =\, \chan{\observable} \wedge \, \channel_I(v_a,v_b) = \,\chan{\reliable}$
 \end{description}
\noindent
We first discuss Cases (a)--(e) and then separately consider the Case {(f)}.

We first argue that for all traces $\tr$ and topologies $\topoi$ and $ \topos$ as described by one of the Cases {(a)}--{(e)}, it holds that $\drsimilar(\TR(\prot, \topos),$ $\TR(\prot,\topoi))$.
This holds as even if the adversary has full control over a channel, she can always behave according to the protocol and not send additional messages, change messages, or reuse messages. Also, she can always deliver messages correctly even on channels that are not reliable. Therefore, the adversary on a weaker channel can perform at least everything that she can on a stronger channel and moreover, the same behavior on a weaker channel does not change any of the signals that are relevant for dispute resolution similarity.
We thus conclude by Lemma~\ref{lemma:drsimilar} and by the assumption that $\fullTime(\prot, \topoi)$ that $\TR(\prot, \topos)$ satisfies $\timeliness(\server)$ and $\proauth(\server)$ in the required adversary models.

It remains to show, that $\prot$ also satisfies the functional property, i.e., $\exists \tr\sd \tr \in \TR(\prot,\topoi) \,\cap\, \functional \implies \exists \tr'\sd \tr' \in \TR(\prot,\topos) \,\cap\, \functional $.
By assumption, the same agents are honest in $\topoi$ and $\topos$. Thus, the same behavior as in $\tr$ can be simulated in $\tr'$.
Also, on all channels, the adversary only forwards the messages in~$\tr$.
As $\forall (v_a,v_b)\sd (v_a,v_b) \in E_I \implies (v_a,v_b) \in E_S$, the same is possible in $\topos$, thus the same messages can be sent and the same signals are produced.
Therefore, $\TR(\prot,\topos) \,\cap\, \functional\, \neq \,\emptyset $ as required.
 
 Finally, we consider the Case {(f)} separately.
 $\topos$ and $\topoi$ are equal except that $(v_a,v_b)$ is additionally undeniable in $\topos$. Thus, the same messages can be sent and received with both topologies and the only difference is that if the protocol $\prot$ uses the channel $(v_a,v_b)$, then, for some $m$, $\pubChan(v_a,v_b,m)$ is additionally recorded in the traces in $\TR(\prot,\topos)$ in contrast to the traces in $\TR(\prot,\topoi)$.
 Let $\prot$ be the protocol such that $\fullTime(\prot,\topoi)$ (which exists by assumption) and $\verdict_I$ be the verdict defined as part of $\prot$.
 First, consider a ballot $\ballot$ such that the verdict whether a trace $\tr_P$ is in $\verdict_I(\server,\ballot)$ does not depend on whether or not $\pubChan(v_a,v_b,m)\in \tr_P$.
 \begin{alignat*}{3}
 &\forall adv \in \{\hShH,\mShH,\hSmH\}, \tr_S \in \TR(\prot,\advModel{\topos}{adv})\sd 
 \\\multicolumn{2}{r}{$
 \exists \tr_I \in \TR(\prot,\advModel{\topoi}{adv}) \sd$}
 \\\multicolumn{2}{r}{$
 \tr_I \in \verdict_I(\server,\ballot) \Leftrightarrow
 \tr_S \in \verdict_I(\server,\ballot)$}
 \\
 \xRightarrow[]{(1)}&\forall adv \in \{\hShH,\mShH,\hSmH\}, \tr_S \in \TR(\prot,\advModel{\topos}{adv})\sd 
 \\\multicolumn{2}{r}{$
 \exists \tr_I \in \TR(\prot,\advModel{\topoi}{adv}) \sd
 $}
 \\\multicolumn{2}{r}{$
 (\tr_I \in\, \protime(\server) \Leftrightarrow \tr_S \in\, \protime(\server) )
 $}
 \\\multicolumn{2}{r}{$ 
 \wedge \,(\tr_I \in\, \proauth(\server) \Leftrightarrow \tr_S \in \,\proauth(\server) )
 $}
 \\
 \xRightarrow[]{(2)}&
 \forall \tr_S \in \TR(\prot,\advModel{\topos}{\hShH})\sd 
 \\\multicolumn{2}{r}{$
 \tr_S \in \protime(\server)\cap \proauth(\server)
 $}
 \\\multicolumn{2}{r}{$
 \wedge \, \forall \tr_S \in \TR(\prot,\advModel{\topos}{\mShH})\sd 
 \tr_S \in \protime(\server)
 $}
 \\\multicolumn{2}{r}{$
\wedge \, \forall \tr_S \in \TR(\prot,\advModel{\topos}{\hSmH})\sd 
 \tr_S \in \proauth(\server).
 $}
\end{alignat*} 
The traces $\tr_S$ and $\tr_I$
only differ in the signals $\pubChan(v_a,v_b,m)$ which are irrelevant for $\verdict_I(\server,\ballot)$ by assumption.
Next, the only way the signals $\pubChan(v_a,v_b,m)$ influence whether a trace is in $\protime(\server)$ or $\proauth(\server)$ is if they change the decision whether the trace is in $\verdict_I(\server,\ballot)$, which is not the case by assumption (Step (1)).
Finally, Step~(2) holds by the assumption that 
$\fullTime(\prot,\topoi)$ and thus
$\TR(\prot,\advModel{\topoi}{\hShH}) \subseteq \, \protime(\server) \,\cap \,\proauth(\server) 
\wedge \,\TR(\prot,\advModel{\topoi}{\mShH}) \subseteq \, \protime(\server) $
$\wedge \,\TR(\prot,\advModel{\topoi}{\hSmH}) \subseteq \,\proauth(\server)$.

Now consider a ballot $\ballot$ for which the verdict whether $\tr_P$ is in $\verdict_I(\server,\ballot)$ depends on whether or not $\pubChan(v_a,v_b,m)\in \tr_P$.
We construct a protocol $\prot'$ that is as $\prot$ except that the definition of the verdict for such ballots (that depend on whether or not $\pubChan(v_a,v_b,m)\in \tr_P$) is changed. Namely, 
$\verdict_S(\server,\ballot)$ (in $\prot'$) is defined as $\verdict_I(\server,\ballot)$ except that all occurrences of $\pubChan(v_a,v_b,m)\in \tr_P$ in $\verdict_I(\server,\ballot)$'s definition are substituted by $\false$ and all occurrences of $\pubChan(v_a,v_b,m)\notin \tr_P$ by $\true$.
Then
 \begin{alignat*}{3}
 &\forall adv \in \{\hShH,\mShH,\hSmH\}, \tr_S \in \TR(\prot',\advModel{\topos}{adv})\sd 
 \\\multicolumn{2}{r}{$
 \exists \tr_I \in \TR(\prot,\advModel{\topoi}{adv}) \sd
 $}
 \\\multicolumn{2}{r}{$
 \tr_I\in \verdict_I(\server,\ballot) \Leftrightarrow \tr_S \in \verdict_S(\server,\ballot)
 $}
\end{alignat*} 
as all signals $\pubChan(v_a,v_b,m)$ that could not influence $\verdict_I(\server,\ballot)$ (as they did not occur in $\tr_I$), are just ignored by $\verdict_S(\server,\ballot)$ to simulate the same behavior.
We can apply the same Steps as (1) and (2) above, to conclude that $\prot'$ satisfies $\timeliness(\server)$ and $\proauth(\server)$ in the required adversary models.

It remains to show that the functional property is preserved.
Let $\tr\in \TR(\prot,\topoi)\,\cap\,\functional$, which exists by assumption.
Then, there is a trace $\tr'$ in $\TR(\prot',\topos)$ where all agents behave exactly as in $\tr$, but which might contain additional signals $\pubChan(v_a,v_b,m)$ compared to $\tr$.
We can conclude that $\tr' \in \,\functional$ as $\functional$ only requires that certain signals appear in the trace and adding more signals does not invalidate the property.

 \textit{Case (2iii): 
 $\exists v \in V_I. \trust_I(v) \smallertopo\trust_S(v)$.}
 We consider the following cases:
 
 \begin{description}
 \item [ {(a)}]$\trust_S(v)=\textit{trusted} \wedge \trust_I(v) \in \{\trustedforward,\trustedreply\}$
 \item [ {(b)}]$\trust_S(v)\in \{\trustedforward,\trustedreply\} \wedge \trust_I(v)=\textit{untrusted}.$
 \end{description}
 Let $\prot$ be a protocol such that $\fullTime(\prot, \topo_I)$, which exists by assumption.
 We argue that, in both Cases (a) and (b), $\TR(\prot, \topos)$ is dispute resolution similar to $\TR(\prot, \topoi)$, i.e., 
 $\drsimilar(\TR(\prot,\topos),\TR(\prot,\topoi))$.
 For this, we take an arbitrary trace $\tr_S$ and adversary assumption $\adv$ such that $\tr_S \in\TR(\prot, \advModel{\topos}{\adv})$ and argue that there exists a trace $\tr_I$ in $\TR(\prot,\advModel{\topoi}{\adv})$, such that $\drequal(\tr_1,\tr_2)$.

 This is the case as a partially trusted (untrusted) agent in $\tr_I$ can always also behave according to the protocol as a trusted (partially trusted) agent in $\tr_S$.
 Thus the traces $\tr_S$ and $\tr_I$ denote the same behavior, but may differ in some signals that are only recorded for honest agents, such as $\verifyCast$.
 However, all signals $\BB$ in $\tr_S$ are equal to those in $\tr_I$ ($\BB$ is always honest and receives the same messages in $\tr_I$ and $\tr_S$),
 all signals $\Evidence$ are equal in $\tr_S$ and $\tr_I$ (they can be recorded by untrusted agents), and all signals $\pubChan$ are equal in $\tr_S$ and $\tr_I$ (as the same channels are used). Thus, $\pubTrace(\tr_I)=\pubTrace(\tr_S)$.
 Moreover, we never consider topologies where $t_S(H)=\textit{trusted}$ or $t_S(\server)=\textit{trusted}$ for the voter $H$ or the authority $\server$ thus the signals $\Ballot$ and $\honest(\server)$ are also equal in $\tr_I$ and $\tr_S$. 
 We conclude by Definition~\ref{def:drSimilar} that $\drequal(\tr_1,\tr_2)$ and by Lemma~\ref{lemma:drsimilar} that $\TR(\prot, \topos)$ satisfies $\timeliness(\server)$ and $\proauth(\server)$ with the required adversary models.

 Furthermore, it holds that $\exists \tr\sd \tr \in \TR(\prot,\topoi) \, \cap \,\functional \implies \exists \tr'\sd \tr' \in \TR(\prot,\topos)\, \cap\, \functional$, as by definition in $\tr$ all agents behave according to the protocol which is a valid behavior in $\topos$ that may however result in additional signals being recorded in $\tr'$. However, as additional signals cannot break the property $\functional$, $\tr'$ satisfies the functional property.
 
 \textit{Case (3): $\topoi \smallertopoNotEqual \topos 
 \wedge \exists \topom \sd 
 \topoi \smallertopoNotEqual \topom \smallertopoNotEqual \topos$.}
 
 Let us, step by step and in an arbitrary order, remove one edge, remove one vertex, weaken one channel assumption, and weaken one trust assumption at a time in $\topos$ until we arrive at the topology $\topoi$.
 As our topologies are finite, i.e., we only consider finitely many roles (vertices), channels (edges), trust types, and channel types, we get a finite sequence of topologies $\topo_1,\dots,\topo_n$ for which
 $\topoi\smallertopo \topo_1 \smallertopo \dots \smallertopo \topo_n \smallertopo \topos$ and where each pair of topologies
 $(\topoi,\topo_1)$, $(\topo_i,\topo_{i+1})_{i \in \{1,\dots,n-1\}}$, and $(\topo_n,\topos)$ only differs in one attribute.
 By the results of Case 2, it follows from $\topoi \smallertopo \topo_1 \wedge \fullTime(\prot,\topoi)$ that $\exists \prot'. \fullTime(\prot',\topo_1)$.
 We can consecutively apply the result of Case 2 to conclude from $\exists \prot. \fullTime(\prot,\topo_{i})$ that $\exists \prot'. \fullTime(\prot',\topo_{i+1})$, for all 
 $ i \in \{1,\dots, n-1\}$.
 Finally, we can conclude from $\topo_n\smallertopo \topos \wedge \exists \prot. \fullTime(\prot,\topo_n)$ that 
 $\exists \prot'. \fullTime(\prot',\topos)$, which concludes the proof. 
\end{proof}

\subsubsection{Proof of Theorem~\ref{theorem:completeChar}}
\label{appendix:proofoftheorem:completeChar}
Next, we state necessary conditions for a topology in our voting protocol class to satisfy $\fullTime$ with some protocol. Then, we show that these conditions are also sufficient.
We use this, in combination with the topology hierarchy, to prove Theorem~\ref{theorem:completeChar} at the end of this section.

Recall from Figure~\ref{fig:GenericSystemSetup} (p.~\pageref{fig:GenericSystemSetup}) that we distinguish two possible setups and types of protocols: the setup $G_S$ for protocols where the ballots are cast by the voters and the setup $G_U$, where the voters and the platforms are not distinguished, for protocols that use trusted platforms to cast the ballots.
As the topologies modeling these two setups are incomparable, we first consider the necessary conditions for topologies $\topo$ such that $G(\topo)\subseteq_G G_S$ and afterwards for the topologies $\topo$ such that $G(\topo)\subseteq_G G_U$.
As explained in the proof of Lemma~\ref{lemma:possAndImposs}, we use $\topo_1 \smallertopoNotEqual \topo_2:= \topo_1 \smallertopo \topo_2 \wedge \topo_1 \neq \topo_2$.
Also, we sometimes write $\topo_2\not\smallertopo \topo_1$ for $\neg(\topo_2\smallertopo \topo_1)$.

\paragraph{Necessary conditions for $\topo$ with $G(\topo)\subseteq_G G_S$}

In order to satisfy the functional property, a necessary requirement is the existence of the roles $\human$, $\platform$, and $\server$ and of the channels $(H,P)$ and $(P,\server)$. 
This is the case as $\human$ must cast a ballot, $\platform$ must forward it, and $\server$ must publish it on the bulletin board.

Given this, we first show that to achieve $\fullTime$ in any topology, there must be a reliable path from $\human$ to $\server$.
That is, both the channels from $\human$ to $\platform$ as well as from $\platform$ to $\server$ must be reliable and the platform must be trusted to forward messages correctly.
This corresponds to the observation that, in practice, when ballots can be dropped between $H$ and $\server$, for example when they are cast by mail and not delivered by the post office, then $\fullTime$ cannot hold even when both
the voter $H$ and the authority $\server$ are honest.
\begin{lemma}\label{lemma:TopoReliableChannel}
 Let $\prot$ be a protocol and $\topo=(V,E,\trust,\channel)$ be a topology in our voting protocol class such that $G(\topo)\subseteq_G G_S$, where $G_S$ is the topology graph in Figure~\ref{fig:GenericSystemSetup}.
 \begin{alignat*}{3}
 &
 \fullTime(\prot,\topo)
 \implies 
 \\ &\multicolumn{2}{r}{$ 
 \insecChan{\reliable}\, \smallertopo \channel(\human,\platform)\,
 \wedge\; \insecChan{\reliable} \,\smallertopo \channel(\platform,\server)
 \wedge\; \trustedforward \smallertopo \trust(\platform). 
 $}
 \end{alignat*}
 \end{lemma}
 \begin{proof}
 We assume that there exists a protocol $\prot$ that satisfies $\fullTime(\prot,\topo)$ for a topology $\topo$ where 
 {(i)} $\insecChan{r} \, \not\smallertopo \channel(\human,\platform)$, 
 {(ii)}$\insecChan{r} \,\not\smallertopo \channel(\platform,\server)$, 
 or {(iii)} $\trustedforward \not\smallertopo \trust(\platform)$ 
 and show that such a protocol cannot exist by arriving at a contradiction.
 For simplicity, we show the proof for Case {(i)}, but it is analogous for the Cases {(ii)} and {(iii)}.

 By the definition of a dispute resolution property (Definition~\ref{def:useful}) and by the functional property $\functional$ (Definition~\ref{def:functional}), there exists a trace $\tr_1 \in \TR(\prot,\advModel{\topo}{\hShH})$ such that 
 $\exists \tr_1',\tr_1'', \human,\ballot, \bs \sd 
 \tr_1 = \tr_1' \cdotp \tr_1'' \wedge
 \Ballot(\human,\ballot) \in \tr_1' \wedge
 \BBrecorded(\bs) \in \tr_1' \wedge \ballot \in \bs \wedge \End \in \tr_1'' \wedge \tr_1 \in \honestNetwork$. Let $\tr_1$ be the smallest trace which satisfies this property.
 As $\server$ is honest in this trace, it further holds that $\honest(\server) \in \tr_1$ and since $\proauth(\server)$ holds, it follows that $\forall \ballot \sd \tr_1 \notin \verdict(\server,\ballot)$.
 
 Suppose that {(i)} holds.
 We show that there exists a trace $\tr_2$ in $\TR(\prot,\advModel{\topo}{\hShH})$, which contradicts $\DRprop{\prot}{\advModel{\topo}{}}{\timeliness(\server)}{\,\proauth(\server)}{\,\functional}$.
 In $\tr_2$, the adversary drops all ballots that are sent on the channel between $\human$ and $\platform$ but other than that behaves as in $\tr_1$.
 This is possible because the channel $(\human, \platform)$ is not reliable in (i).
 Assume that in $\tr_2$ all honest agents, including $\human$ and $\server$, behave according to the same role as in $\tr_1$ and as specified by $\prot$.
 
 It holds that for some traces $\tr_2'$ and $\tr_2''$ and for some list $[\ballot']$, $\tr_2 =\tr_2' \cdotp \tr_2'' \wedge
 \Ballot(\human,\ballot)\in \tr_2' \wedge 
 \BBrecorded([\ballot'])\in \tr_2' \wedge \ballot \notin [\ballot']\wedge \End \in \tr_2'' $.
 This trace exists since $\human$ casts $\ballot$ as in $\tr_1$ and $\server$ publishes only received ballots which are by construction the same as in $\tr_1$ except that $\server$ cannot have received $H$'s ballot $\ballot$.
 First, $\server$ cannot receive $\ballot$ from $\human$, as the adversary drops all ballots on the channel between $\human$ and $\platform$. 
 Moreover, as the adversary does not inject any messages and all honest agents behave as in $\tr_1$, $\ballot$ can also not be received by $\server$ through other channels.
 Nevertheless, $\server$ follows the same role as in $\tr_1$ and, by the assumptions of our protocol class, publishes the result even if it has not received ballots from all voters.
 As $\ballot \notin [\ballot']$, for $\tr_2\in\, \protime(\server)$ it must hold that $ \tr_2 \in \verdict(\server,\ballot)$. However, for $\tr_2\in \,\proauth(\server)$ it must hold that $ \tr_2 \notin \verdict(\server,\ballot)$ as $\server$ is honest. Thus we have a contradiction and such a protocol $\prot$, where $\fullTime(\prot, \topo)$ in Case (i), cannot exist.
 
 The other cases can be shown analogously, as the adversary can again drop the ballots that the voter $H$ sends, either on the channel from $\platform$ to $\server$ (Case {(ii)}) or on $\platform$ when it is dishonest (Case {(iii)}).
 Note that we allow for protocols with multiple instantiations of the roles.
 However, by the topology, the channels from and to any instance of $P$ have the same channel type and any instance of $P$ the same trust type.
 Therefore, the adversary can drop the messages respectively on \emph{all} channels to and from any instance of $P$ or on all instances of $P$ and the above contradiction can be derived independently of the number of instances.
 \end{proof}

In addition to the above lemma, a second necessary condition states that a topology must ensure the existence of evidence, which is required in dispute resolution to unequivocally determine whether a ballot under dispute has been received by the authority.
This can be achieved under five conditions (the five disjuncts in the next lemma).
Evidence can be ensured when one of the channels $(\human,P)$ or $(\platform,\server)$ is undeniable and generates public evidence that a ballot must have been received by $\server$.
Alternatively, when $\platform$ is trusted, it can be used to keep a trustworthy record of the cast ballots as evidence. 
Finally, evidence can be collected in the form of some confirmation that is sent back from $\server$. This requires $\server$ to be trusted to provide a timely reply and the confirmation must either be sent on an undeniable channel $(\server,\platform)$ or it must be ensured that $\human$ receives the confirmation by a reliable path from $\server$ to $\human$.
\begin{lemma}\label{lemma:TopoEvidence}
 Let $\prot$ be a protocol and $\topo=(V,E,\trust,\channel)$ be a topology in our voting protocol class such that $G(\topo)\subseteq_G G_S$, where $G_S$ is the topology graph in Figure~\ref{fig:GenericSystemSetup}. 
 \begin{alignat*}{3}
 &
\fullTime(\prot,\topo)
 \implies
 \\&\multicolumn{2}{r}{$
 \insecChan{\observable} \,\smallertopo \channel(\human,\platform)
 \vee \insecChan{\observable}\,\smallertopo \channel(\platform,\server)
 \vee \textit{trusted} \smallertopo \trust(\platform)
 $} 
 \\&\multicolumn{2}{r}{$
\vee ( \trustedreply \smallertopo \trust(\server)
 \,\wedge \insecChan{\observable}\,\smallertopo \channel(\server,\platform))
$} 
\\&\multicolumn{2}{r}{$
\vee ( \trustedreply\smallertopo \trust(\server)
 \wedge 
 \insecChan{\reliable} \,\smallertopo \channel(\server,\platform)
 \wedge \insecChan{\reliable}\, \smallertopo \channel(\platform,\human)).
$}
 \end{alignat*}
\end{lemma}
\begin{figure*} 
 \centering
\begin{subfigure}[b]{0.4\textwidth}
 \begin{center} \scalebox{0.7}{
\begin{tikzpicture}[->,>=stealth',shorten >=1pt,auto,node distance=3cm and 1.5cm,semithick]

 \node[state] (H) {$\human$};
 \node[state, accepting] (D)[left=1cm of H] {$\device$}; 
 \node[state, accepting,style=dashed](P)[right of=H] {$\platform$};
 \node[state] (S) [right of=P] {$\server$};
 
 \node[state, accepting
 ] (BB)[below=0.4cm of P] {$\BB$};
 \node[state, accepting
 ] (A) [below=0.4cm of H] {$\auditor$};

 \path 
 (H) edge[\edgeSec,bend left]node[above left] {\reliable} (P) 
 (P) edge[\edgeSec]node[below] {\observable} (H) 
 (H) edge[\edgeSec,bend left] node[below] {\observable} (D) 
 (D) edge[\edgeSec,bend left] node[above] {\observable} (H) 
 
 (P) edge[\edgeSec,bend left]node[above right] {\reliable} (S) 
 (S) edge[\edgeSec]node[below] {\observable} (P) 
 
 (S) edge[\edgeAuth
 ] node[below] {\default} (BB)
 (BB)edge[\edgeAuth
 ] node[below] {\default} (H)
 edge[\edgeAuth
 ] node[below left] {\default} (A)
 ;
\end{tikzpicture}}
\end{center}
\caption{The maximal topology $\topo_{I1}$ in Case (i).}
\label{fig:evidenceImpossibility1}
\end{subfigure}
\begin{subfigure}[b]{0.4\textwidth}
 \begin{center} \scalebox{0.7}{
\begin{tikzpicture}[->,>=stealth',shorten >=1pt,auto,node distance=3cm and 1.5cm,semithick]

 \node[state] (H) {$\human$};
 \node[state, accepting] (D)[left=1cm of H] {$\device$};
 \node[state, accepting,style=dashed](P)[right of=H] {$\platform$};
 \node[state, accepting,style=dashed] (S) [right of=P] {$\server$};
 
 \node[state, accepting
 ] (BB)[below=0.4cm of P] {$\BB$};
 \node[state, accepting
 ] (A) [below=0.4cm of H] {$\auditor$};

 \path 
 (H) edge[\edgeSec,bend left]node[above left] {\reliable} (P) 
 (P) edge[\edgeSec]node[below] {\observable} (H) 
 (H) edge[\edgeSec,bend left] node[below] {\observable} (D) 
 (D) edge[\edgeSec,bend left] node[above] {\observable} (H) 
 
 (P) edge[\edgeSec,bend left]node[above right] {\reliable} (S) 
 (S) edge[\edgeSec]node[below] {\default} (P) 
 
 (S) edge[\edgeAuth
 ] node[below] {\default} (BB)
 (BB)edge[\edgeAuth
 ] node[below] {\default} (H)
 edge[\edgeAuth
 ] node[below left] {\default} (A)
 ;
\end{tikzpicture}}
\end{center}
\caption{The maximal topology $\topo_{I2}$ in Case~(ii).}
\label{fig:evidenceImpossibility2}
\end{subfigure}
\begin{subfigure}[b]{0.4\textwidth}
 \begin{center} \scalebox{0.7}{
\begin{tikzpicture}[->,>=stealth',shorten >=1pt,auto,node distance=3cm and 1.5cm,semithick]

 \node[state] (H) {$\human$};
 \node[state, accepting] (D)[left=1cm of H] {$\device$};
 \node[state, accepting,style=dashed](P)[right of=H] {$\platform$};
 \node[state, accepting,style=dashed] (S) [right of=P] {$\server$};
 
 \node[state, accepting
 ] (BB)[below=0.4cm of P] {$\BB$};
 \node[state, accepting
 ] (A) [below=0.4cm of H] {$\auditor$};

 \path 
 (H) edge[\edgeSec,bend left]node[above left] {\reliable} (P) 
 (P) edge[\edgeSec]node[below] {\default} (H) 
 (H) edge[\edgeSec,bend left] node[below] {\observable} (D) 
 (D) edge[\edgeSec,bend left] node[above] {\observable} (H) 
 
 (P) edge[\edgeSec,bend left]node[above right] {\reliable} (S) 
 (S) edge[\edgeSec]node[below] {\reliable} (P) 
 
 (S) edge[\edgeAuth
 ] node[below] {\default} (BB)
 (BB)edge[\edgeAuth
 ] node[below] {\default} (H)
 edge[\edgeAuth
 ] node[below left] {\default} (A)
 ;
\end{tikzpicture}}
\end{center}
\caption{The maximal topology $\topo_{I3}$ in Case (iii).}
\label{fig:evidenceImpossibility3}
\end{subfigure}
\caption{Topologies for which it is impossible to achieve $\fullTime$ by Lemma~\ref{lemma:TopoEvidence}. 
}\label{fig:impossResult}
\end{figure*}
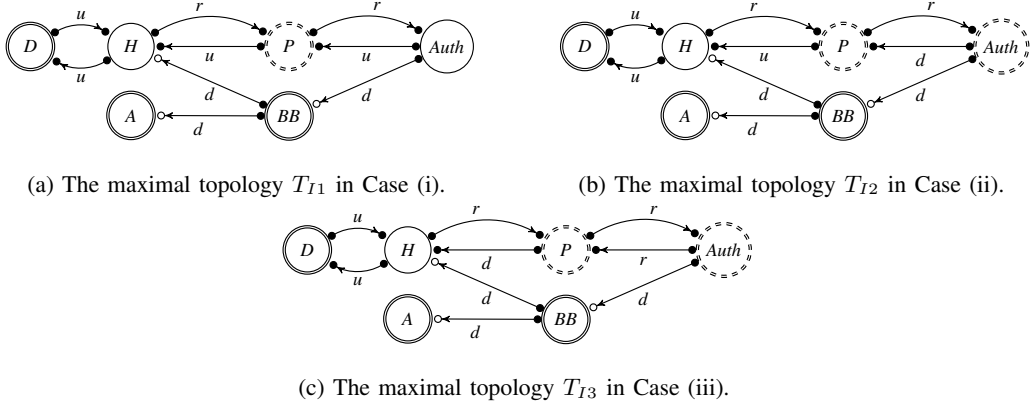

\begin{proof}
 We show the statement by proving its contrapositive, i.e., we show that if all of the disjuncts are false, then so is the property. 
 We distinguish the three cases where all the disjuncts are false.
 In all three cases, it holds that the first three disjuncts are false, i.e., 
 $ \insecChan{\observable}\,\not\smallertopo \channel(\human,\platform)$, 
 $\insecChan{\observable}\,\not\smallertopo \channel(\platform,\server)$,
 and $ \textit{trusted} \,\not\smallertopo\trust(\platform)$.
 Additionally, it holds in Case (i) that 
 $\trustedreply \not\smallertopo \trust(\server) $,
 in Case (ii) that 
 $ \insecChan{r}\,\not\smallertopo \channel(\server,\platform)$,
 and in Case (iii) that 
 $ \insecChan{\observable}\,\not\smallertopo\channel(\server,\platform) 
 \,\wedge\, \insecChan{r}\,\not\smallertopo \channel(\platform,\human)$.
 For each case, we take the topology $\topo$ with the strongest assumptions and where all the conditions of this case hold, i.e. $\topo$ satisfies the conditions and $\forall \topo'\sd \topo' \textit{satisfy conditions} \implies \topo' \smallertopo \topo$, and show that there cannot exist any protocol $\prot$ such that $\fullTime(\prot,\topo)$. 
 It follows by Lemma~\ref{lemma:possAndImposs} that 
 $\forall \topo'\sd \topo' \textit{ satisfies conditions} \implies 
 \neg \exists \prot' \sd \fullTime(\prot',\topo')$ (Lemma~\ref{lemma:possAndImposs} implies that if there were a possibility result for $T'$, there would also be one for $T$, since $T'\smallertopo T$). 
 
 \textit{Case (i):} 
 Figure~\ref{fig:evidenceImpossibility1} depicts the maximal topology $\topo_{I1}$ which satisfies the conditions of this case.
 We assume that there exists a protocol $\prot$ such that 
 $\fullTime(\prot,\topo_{I1})$
 and show that this leads to a contradiction.
 By the functional property and the definition of \functional:
 \begin{alignat*}{3}
 &\exists \tr\in\TR(\prot, \advModel{\topo_{I1}}{\hShH}),
 \tr',\tr'', \human,\ballot,\bs
 \sd \tr = \tr' \cdot \tr''
 \\ \multicolumn{2}{r}{$ 
 \wedge \Ballot(\human,\ballot) \in \tr' 
 \wedge \BBrecorded(\bs) \in \tr '
 \wedge \ballot \in \bs 
 \wedge \End \in \tr'' 
 $}
 \\\multicolumn{2}{r}{$
 \wedge \tr \in \honestNetwork 
 .$}
 \end{alignat*}
 \noindent
 Let $\tr$ be the minimal trace satisfying this.
 We construct two traces 
 $\tr_1\in\TR(\prot, \advModel{\topo_{I1}}{\mShH})$ and 
 $\tr_2 \in \TR(\prot, \advModel{\topo_{I1}}{\hSmH})$.
 
In $\tr_1$, all roles behave exactly as in $\tr$, except 
for $\server$ who ignores $\human$'s ballot $\ballot$ when it is received. In particular, $\server$ does not send any messages on $(\server, \platform)$ which depend on the receipt of $\human$'s ballot $\ballot$ and does not include $\ballot$ in the published recorded ballots. This is possible as the adversary has full control over $\server$ in this trace and can simulate all send and receive events from $\tr$ but leave out selected ones.
In $\tr_2$, all roles behave as in $\tr$, except that $\human$, who is dishonest, never casts the ballot $\ballot$.
This is possible as the adversary can simulate all behaviors of $\human$ in $\tr$ up to the point where $\human$'s ballot $\ballot$ is cast.
$\server$ behaves in $\tr_2$ as in $\tr$, however, as it does not receive $\human$'s ballot $\ballot$, this is not included in the tallying process. Recall that this complies with any role specification of the authority as, by assumption of our protocol class, the honest $\server$ does not wait for all voters' ballots.

We next reason that the traces can be constructed such that $\pubTrace(\tr_1)=\pubTrace(\tr_2)$, i.e., that all signals whose leading function symbol is one of $\Evidence$, $\pubChan$, and $\BB$ are equal in both traces.
First, note that all signals $\Evidence(\ballot, \evidence)$ where $\evidence$ only consists of terms initially known by $\human$, can be forged by the adversary in $\tr_2$ as she compromises $\human$ and learns all his secrets.
Second, all signals $\Evidence(\ballot, \evidence)$ recorded in $\tr_1$ where $\evidence$ does not depend on terms that are only sent by $\server$ if it has received $\ballot$, are also recorded in $\tr_2$.
Finally, messages that are sent by $\server$ only when $\server$ has received the ballot $\ballot$, are never sent in $\tr_1$ by assumption and never sent in $\tr_2$ as $\server$ does not receive $\ballot$.
Thus, if $\evidence$ depends on such messages, it can neither occur in a recorded signal in $\tr_1$ nor in $\tr_2$.
 
Similarly, any recorded signal whose leading function symbol is one of $\BB$ or $\pubChan$ in one trace can be simulated in the other.
In particular, all terms $t$ appearing in signals $\BB(t)$ or $\pubChan(\server,\platform, t)$ must have been sent by $\server$. 
As the dishonest $\server$ in $\tr_1$ has the same knowledge as the honest $\server$ in $\tr_2$, the same terms can be sent to the bulletin board and over the undeniable channel from $\server$ to $\platform$ to generate the same signals.
Finally, all signals $\pubChan(A,B,t)$, for $A \neq \server$, are equal: either sends on such channels do not depend on $\server$ receiving $\human$'s ballot and are done in both traces, or they do depend on it and, in both traces, cannot possibly be done as $\server$ ignores $\human$'s ballot in $\tr_1$ and does not receive it in $\tr_2$.

We have shown that $\tr_1\in\TR(\prot, \advModel{\topo_{I1}}{\mShH})$ and $\tr_2\in\TR(\prot, \advModel{\topo_{I1}}{\hSmH})$ such that $\trP=\pubTrace(\tr_1)=\pubTrace(\tr_2)$.
By $\timeliness(\server)$, it must hold that $\trP\in\verdict(\server,\ballot)$, as the honest voter's ballot $\ballot$ is not recorded in $\tr_1$. However, by $\proauth(\server)$, it must hold that $\trP\notin\verdict(\server,\ballot)$ as $\server$ is honest in $\tr_2$, which yields a contradiction.
Note that the contradiction can be established independently of the number of devices $D$ and $P$ that can communicate with a voter, as each instantiation of $D$ and $P$ as well as all their incoming and outgoing channels have the same trust and channel types, respectively.

\textit{Case (ii):} 
The maximal topology $\topo_{I2}$ that satisfies these conditions is depicted in Figure~\ref{fig:evidenceImpossibility2}.
Similarly to Case (i), we assume a protocol $\prot$ for which $\fullTime(\prot,\topo_{I2})$ and construct two traces $\tr_3\in\TR(\prot, \advModel{\topo_{I2}}{\mShH})$ and $\tr_4\in\TR(\prot, \advModel{\topo_{I2}}{\hSmH})$ such that $\trP=\pubTrace(\tr_3)=\pubTrace(\tr_4)$ and that yield a contradiction as they respectively require that 
$\trP \in \verdict(\server,\ballot)$ and $\trP \notin \verdict(\server,\ballot)$.
In particular, this holds for $\tr_4=\tr_2$, where $\tr_2$ is from Case (i), and for $\tr_3$ which is as $\tr_1$ from Case (i) except for the following differences.
In $\tr_3$, instead of ignoring $H$'s ballot $\ballot$ as in $\tr_1$, the partially trusted $\server$ answers with any response required by the protocol but does not further consider the ballot $\ballot$, e.g., when writing all recorded ballots on the bulletin board. 
The adversary then drops on the channel $(\server,\platform)$ any such responses that are sent by $\server$ only when the ballot $\ballot$ from $H$ has been received. This results in $\pubTrace(\tr_3)=\pubTrace(\tr_1)$ and thus the conclusions from Case~(i) apply.

\textit{Case (iii):} 
The maximal topology $\topo_{I3}$ that satisfies these conditions is depicted in Figure~\ref{fig:evidenceImpossibility3}.
Again we assume a protocol $\prot$ for which $\fullTime(\prot,\topo_{I3})$ and construct two traces $\tr_5\in\TR(\prot, \advModel{\topo_{I3}}{\mShH})$ and $\tr_6\in\TR(\prot, \advModel{\topo_{I3}}{\hSmH})$ such that $\trP=\pubTrace(\tr_5)=\pubTrace(\tr_6)$ and that yield a contradiction as they respectively require that 
$\trP \in \verdict(\server,\ballot)$ and $\trP \notin \verdict(\server,\ballot)$.
This holds for $\tr_6=\tr_2$, where $\tr_2$ is from Case (i), and for $\tr_5$ which is as $\tr_3$ from Case (ii), except that the adversary drops all relevant messages sent by $\server$ on the channel $(\platform,\human)$ instead of $(\server,\platform)$.
The same conclusions as in Cases (i) and (ii) follow.
\end{proof}

\paragraph{Necessary conditions for $\topo$ with $G(\topo)\subseteq_G G_U$}
Next, we consider the necessary conditions to satisfy $\fullTime$ for topologies where the roles of the voter and platform are unified.
We will see that these conditions are closely related to the ones established above.
First, note that we require at least the roles $H$ and $\server$ and the channel $(H,\server)$ as $H$ must cast the ballot and $\server$ must publish it on the bulletin board in order for the functional property to hold.

Next, similarly to Lemma~\ref{lemma:TopoReliableChannel}, we establish that to achieve $\fullTime$ in any topology, there must be a reliable channel from $\human$ to $\server$.
This corresponds to the observation that, in practice, in a remote e-voting setting when ballots are cast over the insecure Internet where they can be dropped, then $\fullTime$ cannot hold even when both the voter $H$ and the authority $\server$ are honest.
\begin{lemma}\label{lemma:TopoReliableChannelUnified}
Let $\prot$ be a protocol and $\topo=(V,E,\trust,\channel)$ be a topology in our voting protocol class such that $G(\topo)\subseteq_G G_U$, where $G_U$ is the topology graph in Figure~\ref{fig:GenericSystemSetup}. 
\begin{alignat*}{3}
 &
\fullTime(\prot,\topo)
 \implies 
\insecChan{\reliable} \,\smallertopo \channel(\human,\server) .
\end{alignat*}
\end{lemma}
\begin{proof}
We assume that there exists a protocol $\prot$ and a topology $\topo$ such that $\fullTime(\prot,\topo)$ but where nevertheless $\insecChan{r} \,\not\smallertopo \channel(\human,\server)$ and show that such a protocol cannot exist by arriving at a contradiction.
\\
By the definition of a dispute resolution property (Definition~\ref{def:useful}) and by the functional property $\functional$ (Definition~\ref{def:functional}), there exists a trace $\tr_1 \in \TR(\prot,\advModel{\topo}{\hShH})$ such that 
$\exists \tr_1',\tr_1'', \human,\ballot, \bs \sd 
\tr_1 = \tr_1' \cdotp \tr_1'' \wedge
\Ballot(\human,\ballot) \in \tr_1' \wedge
\BBrecorded(\bs) \in \tr_1' \wedge \ballot \in \bs \wedge \End \in \tr_1'' \wedge \tr_1 \in \honestNetwork$. Let $\tr_1$ be the smallest trace which satisfies this property.

As in the proof of Lemma~\ref{lemma:TopoReliableChannel}, we show that there exists a trace $\tr_2$ in $\TR(\prot,\advModel{\topo}{\hShH})$, which contradicts $\DRprop{\prot}{\advModel{\topo}{}}{\timeliness(\server)}{\,\proauth(\server)}{\,\functional}$.
In particular, this holds for a trace $\tr_2$ where all honest agents behave according to the same role as in $\tr_1$ and where the adversary drops all ballots that are sent on the channel between $\human$ and $\server$ but other than that behaves as in $\tr_1$.

We can apply the reasoning from the proof of Lemma~\ref{lemma:TopoReliableChannel} and conclude that for $\tr_2\in\, \protime(\server)$ it must hold that $\tr_2 \in \verdict(\server,\ballot)$, as the honest voter $H$'s ballot $\ballot$ is not in the list of recorded ballots on the bulletin board in $\tr_2$.
However, at the same time, for $\tr_2\in \,\proauth(\server) $ it must hold that $ \tr_2 \notin \verdict(\server,\ballot)$. Thus we have a contradiction and such a protocol $\prot$, where $\fullTime(\prot, \topo)$, cannot exist.
By the same reasoning as in the proof of Lemma~\ref{lemma:TopoReliableChannel}, we can also conclude that this holds independently of the number of agents instantiating the roles.

\end{proof}
As for the other setup, we show that a second necessary condition is that a topology ensures the existence of evidence.
In particular, evidence can be established by an undeniable channel from $H$ to $\server$.
Alternatively it can be established by a reliable channel $(\server,H)$ and a partially trusted $\server$, in which case a confirmation can be sent back from $\server$ to $H$ upon receiving the ballot, which can serve as evidence.
\begin{figure*} 
 \centering
\begin{subfigure}[b]{0.4\textwidth}
 \begin{center} \scalebox{0.7}{
\begin{tikzpicture}[->,>=stealth',shorten >=1pt,auto,node distance=3cm and 1.5cm,semithick]

 \node[state] (H) {$\human$};
 \node[state, accepting] (D)[left=1cm of H] {$\device$};
 \node[state] (S) [right of=H] {$\server$};
 
 \node[state, accepting
 ] (BB)[right=0.6 of A] {$\BB$};
 \node[state, accepting
 ] (A) [below=0.2cm of H] {$\auditor$};

 \path 
 
 (H) edge[\edgeSec,bend left]node[above right] {\reliable} (S) 
 (S) edge[\edgeSec]node[below] {\observable} (H) 
 (H) edge[\edgeSec,bend left] node[below] {\observable} (D) 
 (D) edge[\edgeSec,bend left] node[above] {\observable} (H)

 (S) edge[\edgeAuth
 ] node[below] {\default} (BB)
 (BB)edge[\edgeAuth
 ] node[below] {\default} (H)
 edge[\edgeAuth
 ] node[below left] {\default} (A)
 ;
\end{tikzpicture}}
\end{center}
\caption{The maximal topology $\topo_{I4}$ in Case (i).}
\label{fig:evidenceImpossibility4}
\end{subfigure}
\begin{subfigure}[b]{0.4\textwidth}
 \begin{center} \scalebox{0.7}{
\begin{tikzpicture}[->,>=stealth',shorten >=1pt,auto,node distance=3cm and 1.5cm,semithick]

 \node[state] (H) {$\human$};
 \node[state, accepting] (D)[left=1cm of H] {$\device$};
 \node[state, accepting,style=dashed] (S) [right of=H] {$\server$};
 
 \node[state, accepting
 ] (BB)[right=0.6 of A] {$\BB$};
 \node[state, accepting
 ] (A) [below=0.2cm of H] {$\auditor$};

 \path 
 
 (H) edge[\edgeSec,bend left]node[above right] {\reliable} (S) 
 (S) edge[\edgeSec]node[below] {\default} (H) 
 (H) edge[\edgeSec,bend left] node[below] {\observable} (D) 
 (D) edge[\edgeSec,bend left] node[above] {\observable} (H) 
 
 (S) edge[\edgeAuth
 ] node[below] {\default} (BB)
 (BB)edge[\edgeAuth
 ] node[below] {\default} (H)
 edge[\edgeAuth
 ] node[below left] {\default} (A)
 ;
\end{tikzpicture}}
\end{center}
\caption{The maximal topology $\topo_{I5}$ in Case~(ii).}
\label{fig:evidenceImpossibility5}
\end{subfigure}
\caption{Topologies for which it is impossible to achieve $\fullTime$ by Lemma~\ref{lemma:TopoEvidenceUnified}. 
}\label{fig:impossResultHP}
\end{figure*}
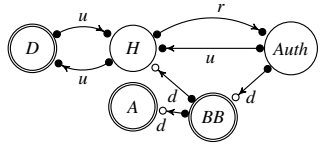
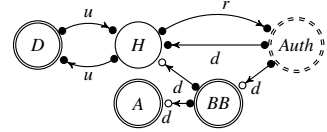
\begin{lemma}\label{lemma:TopoEvidenceUnified}
Let $\prot$ be a protocol and $\topo=(V,E,\trust,\channel)$ be a topology in our voting protocol class such that $G(\topo)\subseteq_G G_U$, where $G_U$ is the topology graph in Figure~\ref{fig:GenericSystemSetup}. 
\begin{alignat*}{3}
 &
\fullTime(\prot,\topo)
 \implies 
 \insecChan{\observable} \,\smallertopo \channel(\human,\server)
 \\&\multicolumn{2}{r}{$
\vee ( \trustedreply\smallertopo \trust(\server)
 \wedge 
 \insecChan{\reliable} \,\smallertopo \channel(\server,\human)).
$}
 \end{alignat*}
\end{lemma}
\begin{proof}
 As in the proof of Lemma~\ref{lemma:TopoEvidence}, we show the statement by proving its contrapositive, i.e., we show that if both disjuncts are false, then so is the property. 
 We distinguish the two cases where the disjuncts are false.
 In Case (i), $\insecChan{\observable} \,\not\smallertopo \channel(\human,\server)$ and $ \trustedreply\not\smallertopo \trust(\server)$ and in Case (ii), 
 $\insecChan{\observable} \,\not\smallertopo \channel(\human,\server)$ and
 $\insecChan{\reliable} \,\not\smallertopo \channel(\server,\human)$.
Then, for both cases, we take the topology $\topo$ with the strongest assumptions and where all the conditions of this case hold and show that there cannot exist any protocol $\prot$ such that $\fullTime(\prot,\topo)$. 
As argued in the proof of Lemma~\ref{lemma:TopoEvidence}, this implies an impossibility for all topologies satisfying the conditions (by Lemma~\ref{lemma:possAndImposs}).

\textit{Case (i):} 
The maximal topology $\topo_{I4}$ that satisfies these conditions is depicted in Figure~\ref{fig:evidenceImpossibility4}.
As in the proof of Lemma~\ref{lemma:TopoEvidence}, we assume a protocol $\prot$ for which $\fullTime(\prot,\topo_{I4})$ and construct two traces $\tr_1\in\TR(\prot, \advModel{\topo_{I4}}{\mShH})$ and $\tr_2\in\TR(\prot, \advModel{\topo_{I4}}{\hSmH})$ such that $\trP=\pubTrace(\tr_1)=\pubTrace(\tr_2)$ and that yield a contradiction as they respectively require that 
$\trP \in \verdict(\server,\ballot)$ and $\trP \notin \verdict(\server,\ballot)$.
In particular, let $\tr_1$ be as $\tr_1$ in Lemma~\ref{lemma:TopoEvidence}'s proof where $\server$ ignores $H$'s ballot, except that here the fact that $\server$ ignores $H$'s ballot $b$ means that it never sends any message dependent on the receipt of $b$ on the channel $(\server,H)$ (rather than the channel $(\server,P)$ in Lemma~\ref{lemma:TopoEvidence}'s proof).
Also, let $\tr_2$ be as $\tr_2$ in Lemma~\ref{lemma:TopoEvidence}'s proof, where $H$ never casts a ballot.
As argued in Lemma~\ref{lemma:TopoEvidence}'s proof, this results in $\pubTrace(\tr_1)=\pubTrace(\tr_2)$ and yields a contradiction.

\textit{Case (ii):}
The maximal topology $\topo_{I5}$ that satisfies these conditions is depicted in Figure~\ref{fig:evidenceImpossibility5}.
Again we assume a protocol $\prot$ for which $\fullTime(\prot,\topo_{I5})$ and construct two traces $\tr_3\in\TR(\prot, \advModel{\topo_{I5}}{\mShH})$ and $\tr_4\in\TR(\prot, \advModel{\topo_{I5}}{\hSmH})$ such that $\trP=\pubTrace(\tr_3)=\pubTrace(\tr_4)$ and that yield a contradiction as they respectively require that 
$\trP \in \verdict(\server,\ballot)$ and $\trP \notin \verdict(\server,\ballot)$.
This holds for $\tr_4$ that is as $\tr_2$ from Case (i), and for $\tr_3$ that is as $\tr_1$ from Case (i), 
except that in $\tr_3$, instead of ignoring $H$'s ballot $\ballot$, the partially trusted $\server$ answers with any response required by the protocol but does not further consider the ballot $\ballot$, e.g., when writing all recorded ballots on the bulletin board. 
The adversary then drops on the channel $(\server,\human)$ any such responses that are sent by $\server$ only when the ballot $\ballot$ from $H$ has been received. This results in $\pubTrace(\tr_3)=\pubTrace(\tr_1)$ and thus the same conclusions as in Case~(i) and in the proof of Lemma~\ref{lemma:TopoEvidence} apply.
\end{proof}

\paragraph{Sufficient conditions}
We next show that the above conditions are also sufficient by showing that the seven topologies $\topo_1,\dots,\topo_7$ in Figure~\ref{fig:possResults} satisfy these conditions and by establishing a possibility result for each of them. 
All topologies
have a reliable path from $\human$ to $\server$, as this is required by Lemmas~\ref{lemma:TopoReliableChannel} and~\ref{lemma:TopoReliableChannelUnified}, 
and additional trust assumptions, required by Lemmas~\ref{lemma:TopoEvidence} and~\ref{lemma:TopoEvidenceUnified}.

The next lemma states that for all these topologies, there exists a protocol that satisfies $\fullTime$ and that these topologies are minimal.
That is, there are no topologies that have weaker assumptions than $\topo_1,\dots,\topo_7$ but where it is nevertheless possible to achieve $\fullTime$ with some protocol.
We establish the lemma's first part by presenting for each topology $\topo_{i}$, $i\in \{1,\dots,7\}$ a protocol $\prot_i$ and proving $\fullTime(\prot_i,\topo_i)$ with the Tamarin tool~\cite{tamarin}.
All relevant Tamarin files can be found in~\cite{tamarinfiles}.
\begin{lemma}\label{lemma:possibility}\label{lemma:minimalTopo}
Let the $\topo_i$, for $i\in\{1,\dots,7\}$, be the topologies depicted in Figure~\ref{fig:possResults} and $T$ be a topology in our voting protocol class where $\topo \smallertopoNotEqual \topo_i$ for some $i$.
\begin{alignat*}{3}
\exists \prot \sd &\fullTime(\prot,\topo_i)
\wedge \neg \exists \prot' \sd \fullTime(\prot',\topo).
\end{alignat*}
\end{lemma}
\begin{proof}

\begin{figure}
\centering
 \scalebox{0.8}{\includegraphics{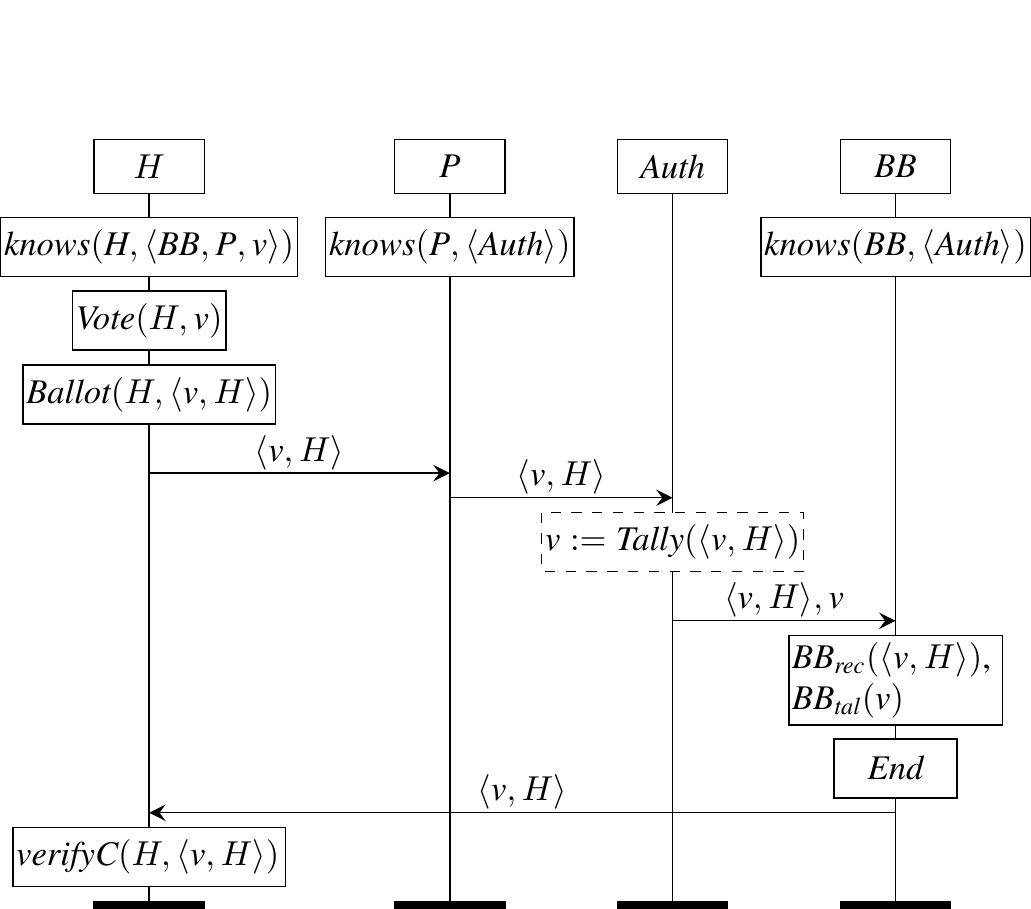} }
\scalebox{0.84}{\parbox[c]{\linewidth}{
\begin{alignat*}{1}
 &\verdict(\server,\ballot):= 
 \{\tr|\exists \platform, \bs, H, v\sd \pubChan(\platform,\server,\ballot) \in \tr
 \\ \multicolumn{2}{r}{$
 \wedge \BBrecorded(\bs) \in \tr
 \wedge \ballot \notin \bs
 \wedge \ballot = \langle v,H \rangle \}.
 $}
\end{alignat*}}}
 \caption{The protocol $\prot_1$.}\label{fig:possProtocol1}
\end{figure}

\begin{figure}
 \centering
\scalebox{0.8}{\includegraphics{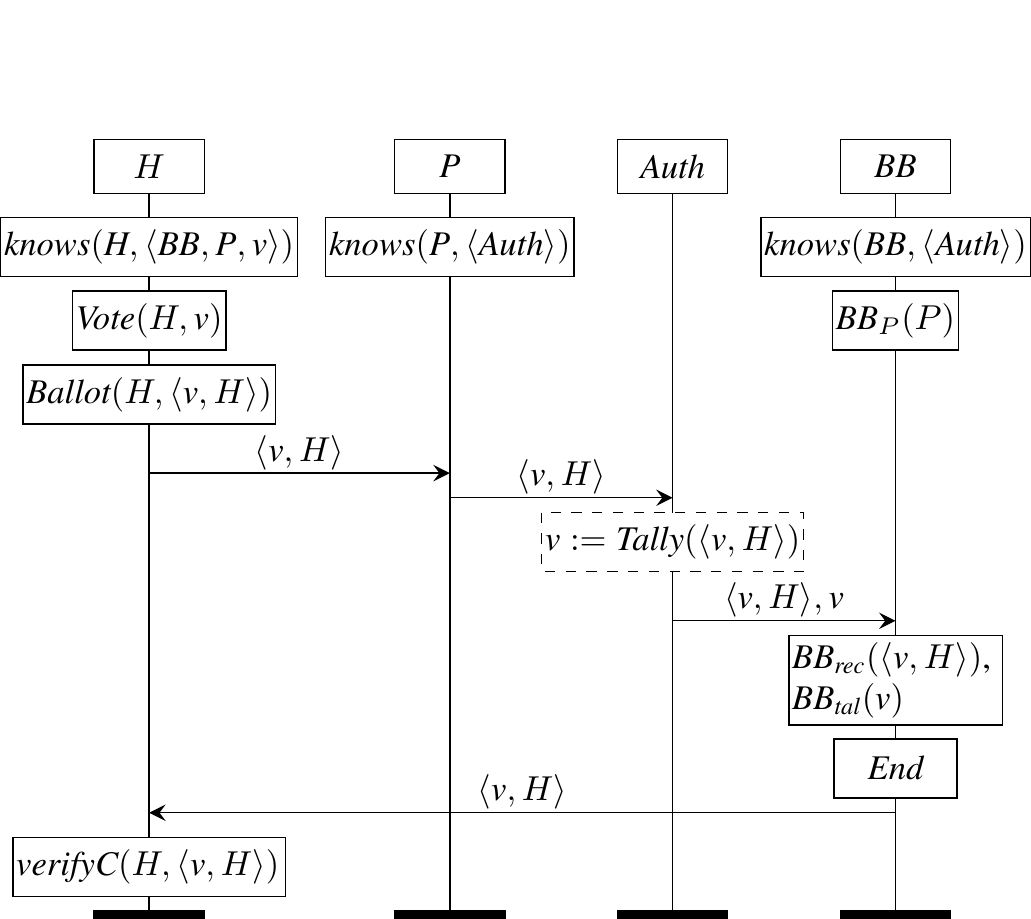}}
\scalebox{0.85}{\parbox[c]{\linewidth}{
\begin{alignat*}{1}
 \verdict(\server,\ballot):=
 \{\tr|
\exists \human,\platform, \bs, \vote \sd
\pubChan(\human,\platform,\ballot) \in \tr
\wedge \BB_\platform(\platform) \in \tr 
\\\multicolumn{2}{r}{$
\wedge \BBrecorded(\bs) \in \tr 
\wedge \ballot \notin \bs 
 \wedge \ballot = \langle v,H \rangle \}.$}
\end{alignat*}}}
\caption{The protocol $\prot_2$.}\label{fig:possProtocol2}
\end{figure}

\begin{figure}
 \centering
\scalebox{0.8}{\includegraphics{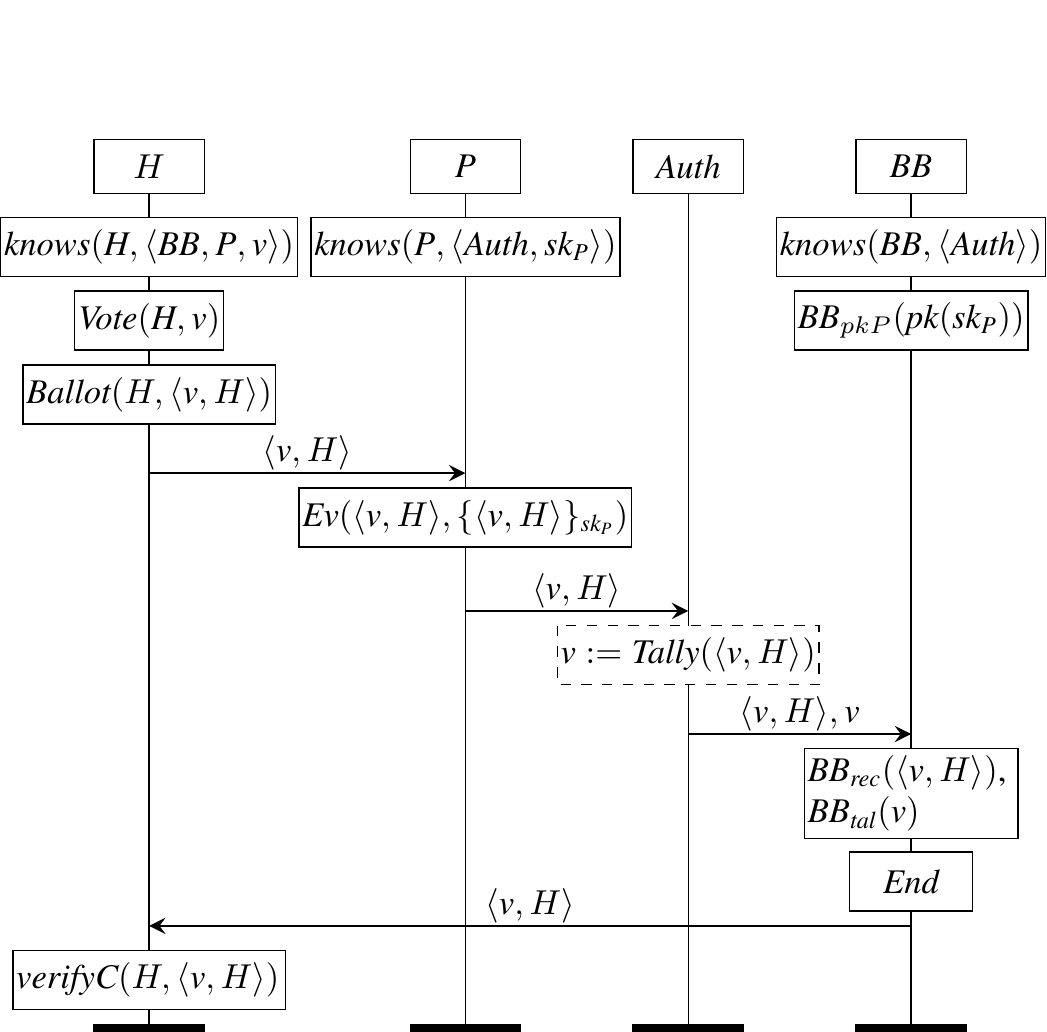}}
 \scalebox{0.85}{ \parbox[c]{\linewidth}{
\begin{alignat*}{1}
\textit{For } \ballot = \perp: &\; \verdict(\server, \ballot):=\{ \}.
 \\
\textit{For }\ballot \neq \perp: &\;
\verdict(\server,\ballot):= 
\{\tr|
\exists \confirmation,\bs, \pkP \sd 
\BBpkP(\pkP) \in \tr 
\\\multicolumn{2}{r}{$ 
\wedge \Evidence(\ballot,\confirmation) \in \tr 
\wedge \ver(\confirmation, \pkP)=\ballot 
\wedge \BBrecorded(\bs) \in \tr 
\wedge \ballot \notin \bs\}.
$}
\end{alignat*}}}
\caption{The protocol $\prot_3$.}\label{fig:possProtocol3}

\end{figure}
\begin{figure}
 \centering
\scalebox{0.8}{\includegraphics{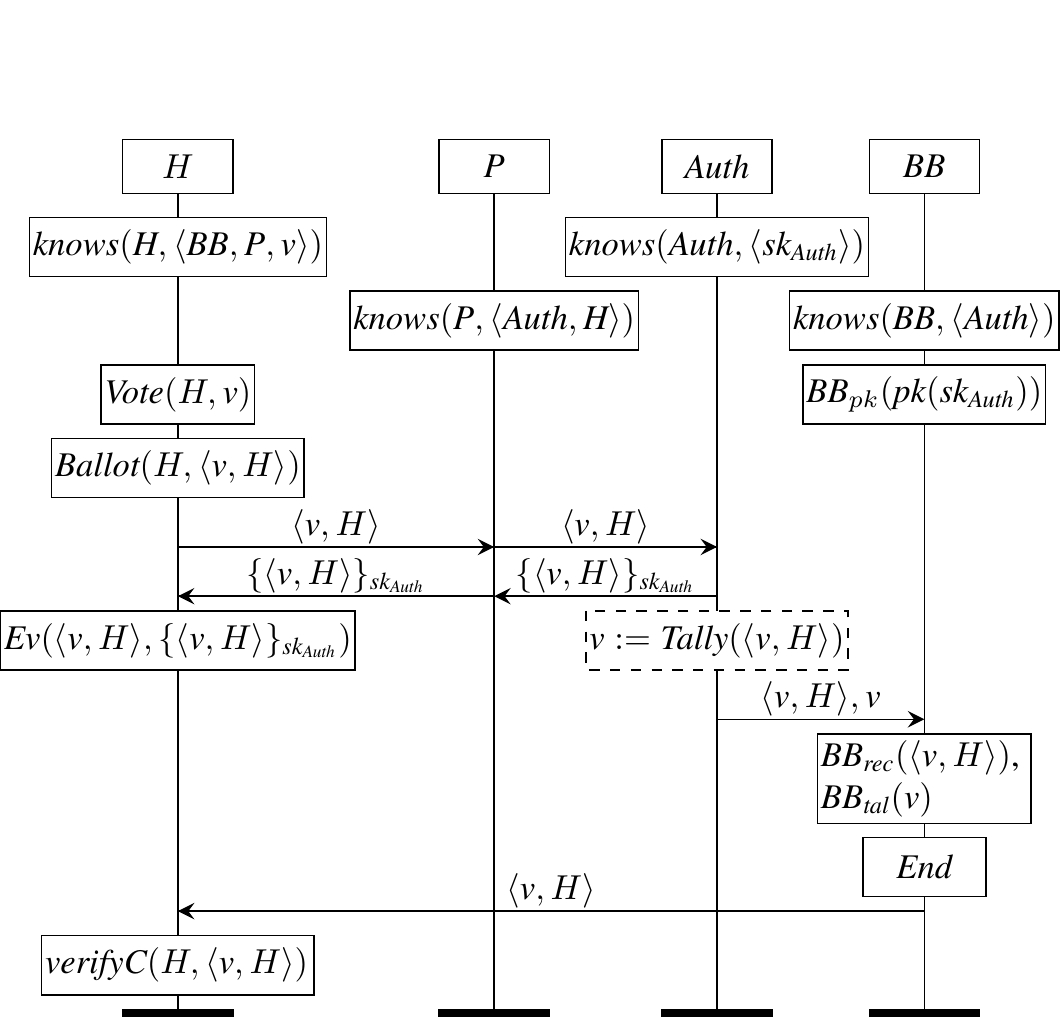}}
\scalebox{0.85}{ \parbox[c]{\linewidth}{
\begin{alignat*}{1}
\textit{For } \ballot = \perp: &\; \verdict(\server, \ballot):=\{ \}.
 \\
\textit{For }\ballot \neq \perp: &\;
 \verdict(\server,\ballot):= 
 \{\tr|
\exists \confirmation, \pkS, \bs \sd
\BBpkS(\pkS) \in \tr 
\\\multicolumn{2}{r}{$
\wedge \Evidence(\ballot,\confirmation) \in \tr 
\wedge \ver(\confirmation, \pkS)=\ballot
\wedge \BBrecorded(\bs) \in \tr \wedge \ballot \notin \bs\}.
$}
\end{alignat*}}}
 \caption{The protocol $\prot_4$.}\label{fig:possProtocol4}
\end{figure}

\begin{figure}
 \centering
\scalebox{0.8}{\includegraphics{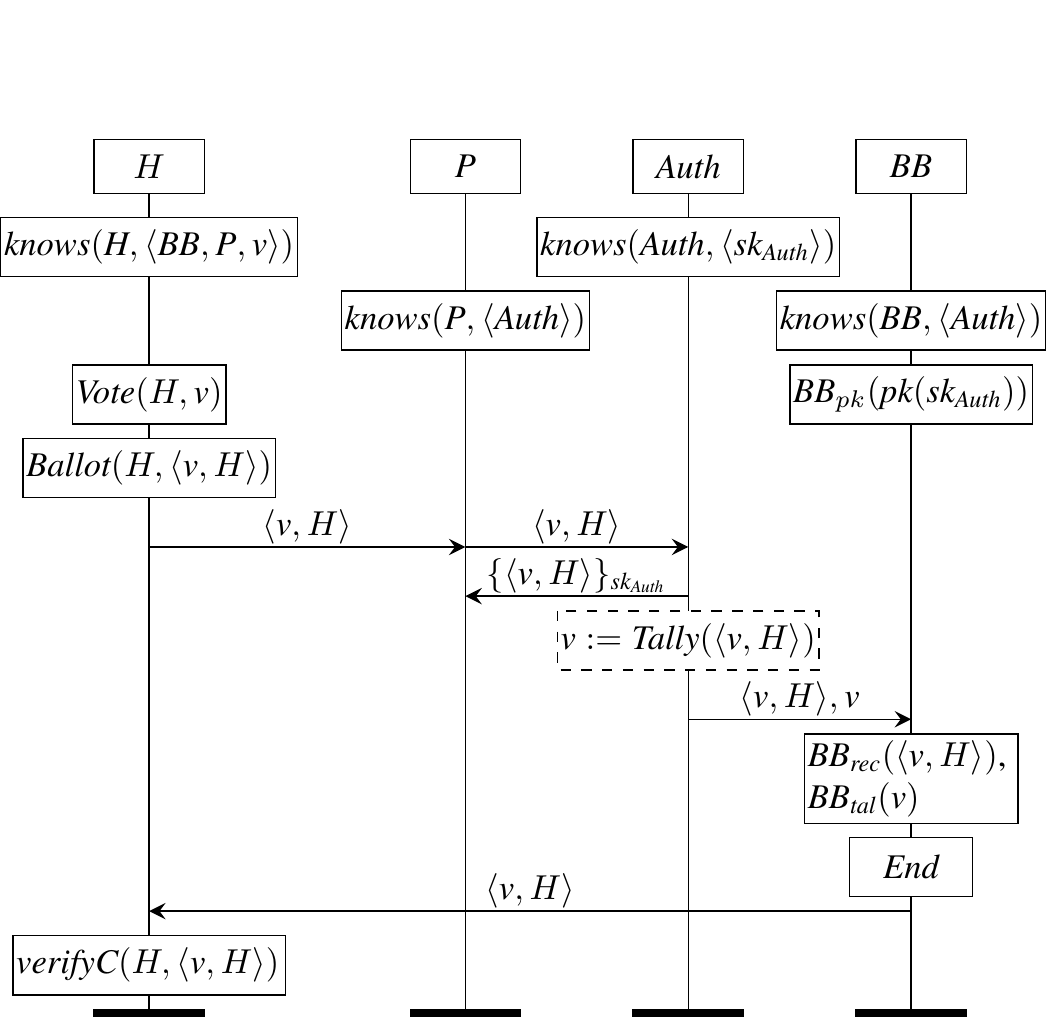}}
\scalebox{0.85}{ \parbox[c]{\linewidth}{
\begin{alignat*}{1}
\textit{For } \ballot = \perp: &\; \verdict(\server, \ballot):=\{ \}.
 \\
\textit{For }\ballot \neq \perp: &\;
 \verdict(\server,\ballot):= 
 \{\tr|
\exists \confirmation, \platform, \pkS, \bs \sd
\BBpkS(\pkS) \in \tr
\\\multicolumn{2}{r}{$
\wedge \pubChan(\server,\platform,\confirmation) \in \tr
\wedge \ver(\confirmation, \pkS)=\ballot
\wedge \BBrecorded(\bs) \in \tr \wedge \ballot \notin \bs\}.
$}
\end{alignat*}}}
\caption{The protocol $\prot_5$.}\label{fig:possProtocol5}
\end{figure}

\begin{figure}
 \centering
\scalebox{0.8}{\includegraphics{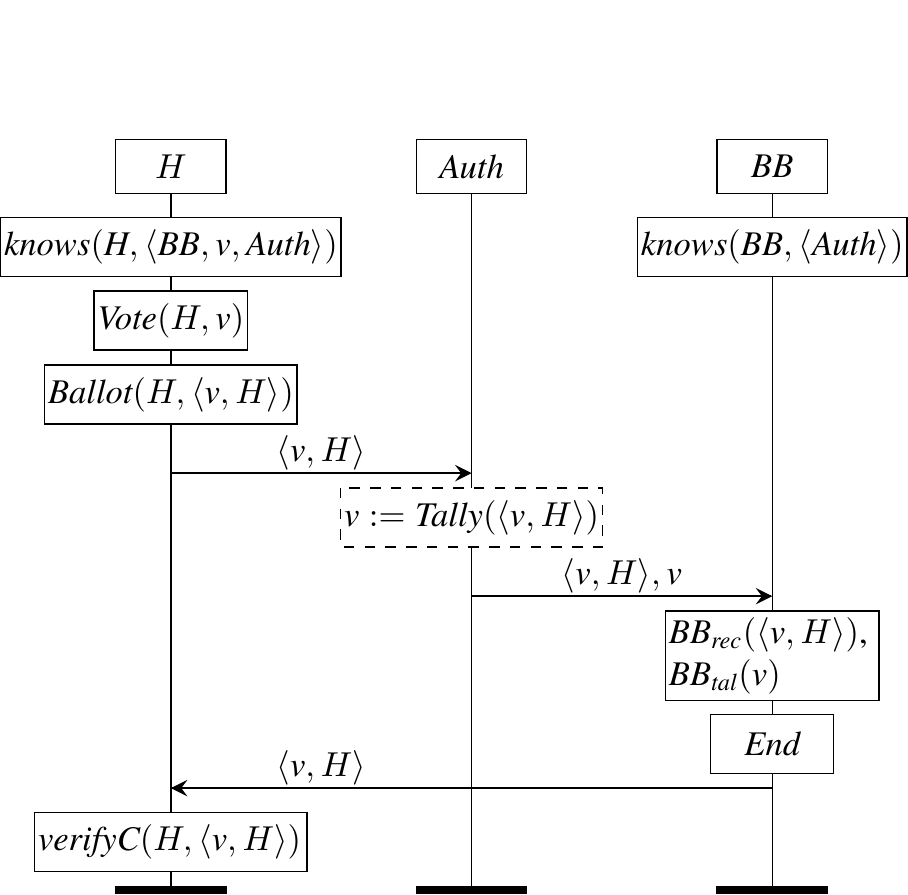}}
\scalebox{0.85}{ \parbox[c]{\linewidth}{
\begin{alignat*}{1}
 \verdict(\server,\ballot):= 
 \{\tr|\exists H, \bs,\vote \sd \pubChan(H,\server,\ballot) \in \tr \\ \multicolumn{2}{r}{$
 \wedge \BBrecorded(\bs) \in \tr 
 \wedge \ballot \notin \bs
 \wedge \ballot = \langle v,H\rangle \}.
 $}
\end{alignat*}}}
\caption{The protocol $\prot_6$.}\label{fig:possProtocol6}
\end{figure}

\begin{figure}
 \centering
\scalebox{0.8}{\includegraphics{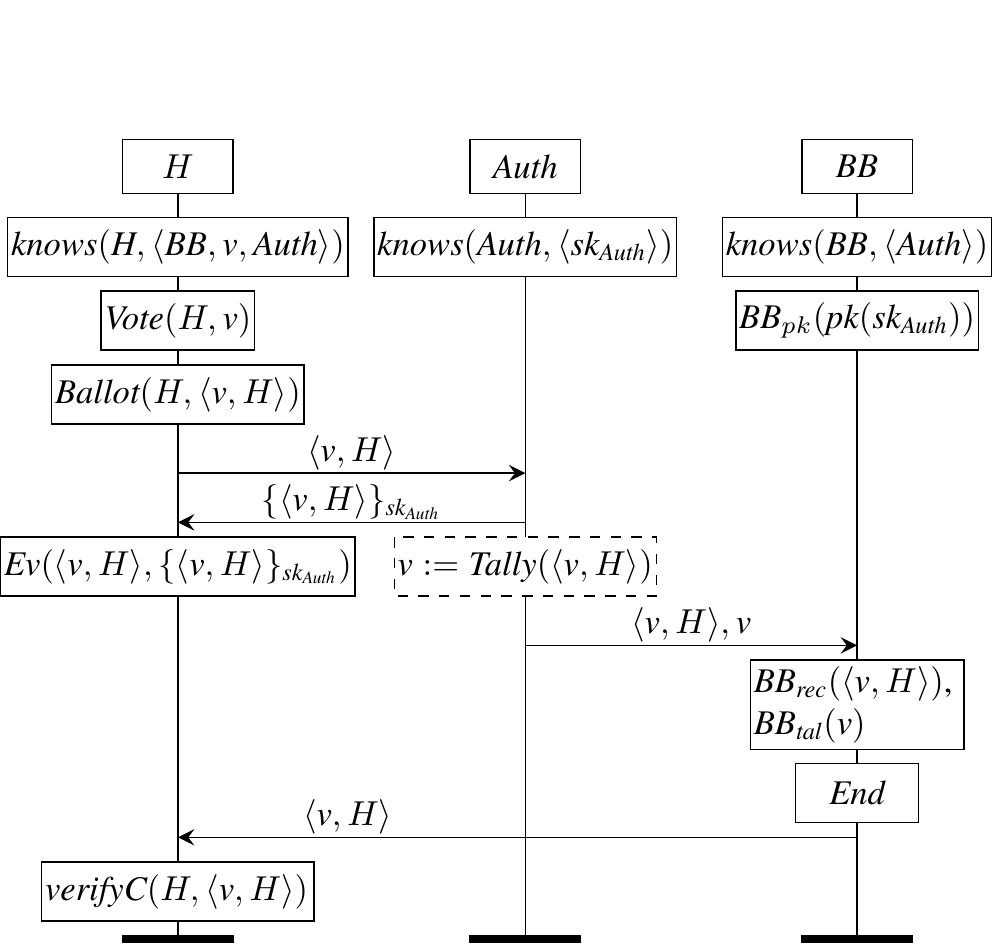}}
\scalebox{0.85}{ \parbox[c]{\linewidth}{
\begin{alignat*}{1}
\textit{For } \ballot = \perp: &\; \verdict(\server, \ballot):=\{ \}.
 \\
\textit{For }\ballot \neq \perp: &\;
 \verdict(\server,\ballot):= 
 \{\tr|
\exists \confirmation, \pkS, \bs \sd
\BBpkS(\pkS) \in \tr 
\\\multicolumn{2}{r}{$
\wedge \Evidence(\ballot,\confirmation) \in \tr 
\wedge \ver(\confirmation, \pkS)=\ballot
\wedge \BBrecorded(\bs) \in \tr \wedge \ballot \notin \bs\}
$}
\end{alignat*}}}
\caption{The protocol $\prot_7$.}\label{fig:possProtocol7}
\end{figure}

Recall the topologies $\topo_1,\dots,\topo_7$ in Figure~\ref{fig:possResults}. 
We first present for each topology $\topo_i, i \in \{1, \dots, 7\}$ a protocol $\prot_i$, as depicted in Figures~\ref{fig:possProtocol1}--\ref{fig:possProtocol7}. 
\label{msc}
We present the protocols using message sequence charts, as explained in Section~\ref{subsubsec:mixnet:protocol}.
As a simple protocol with one voter is sufficient to demonstrate possibility, all protocols $\prot_i$ specify that only one election is run at a time with one voter $H$ and one platform $P$.
In all protocols, the ballot is a pair consisting of the voter's vote $\vote$ and his identity $\human$.
We extend the term algebra from Section~\ref{sec:protocolModel} by defining for the tally function the equation $\Tally(\langle \vote, \human \rangle)=\vote$.
Moreover, for each protocol, the definition of $\verdict$ is given below the protocol's message sequence chart.

For each topology $\topo_i$ and protocol $\prot_i$, we prove that $\fullTime(\prot_i,\topo_i)$ holds.
To do this, we model in three separate Tamarin theories the traces $\TR(\prot_i,\topo_i^{\hShH})$, $\TR(\prot_i,\topo_i^{\mShH})$, and $\TR(\prot_i,\topo_i^{\hSmH})$.
We then automatically prove that 
$\TR(\prot_i, \advModel{\topo_i}{\mShH}) \subseteq \, \protime(\server)$, 
$\TR(\prot_i, \advModel{\topo_i}{\hSmH}) \subseteq\, \proauth(\server)$,
and $\TR(\prot_i, \advModel{\topo_i}{\hShH}) \,\subseteq\,\protime(\server) \,\cap\,\proauth(\server)$.
All relevant Tamarin files are in~\cite{tamarinfiles}.
Furthermore, it is obvious that all protocols satisfy
$\TR(\prot_i, \advModel{\topo_i}{\hShH})\, \cap\, \functional \,\neq \, \emptyset$.
\\\\
Next, we examine for each topology $\topo_i$ in Figure~\ref{fig:possResults} all minimal possibilities of making the topology weaker, i.e., to generate a topology 
 $\topo'$ such that $\topo' \smallertopoNotEqual \topo_i$ and 
 $\neg \exists \topo_M \sd \topo' \smallertopoNotEqual \topo_M \smallertopoNotEqual \topo_i$.
 We then argue that any such topology $\topo'$ cannot satisfy $\fullTime$ for any protocol, by Lemmas~\ref{lemma:TopoReliableChannel}--\ref{lemma:TopoEvidenceUnified}.
 It follows by Lemma~\ref{lemma:possAndImposs} that all weaker topologies $\topo'' \smallertopo \topo'$ also cannot satisfy the property (as this would imply that $\topo'$ satisfies the property, too).
 Note that we can weaken a topology in three ways: by weakening the channel assumptions, by weakening the trust assumptions, or by removing a channel or a role.
 
 First, recall that by assumption of our protocol class, the roles $\BB$ and $\auditor$ and their incoming and outgoing channels are fixed.
 Furthermore, for in the topologies $\topo_1,\dots,\topo_5$, we cannot remove the channels $(\human,\platform)$ or $(\platform,\server)$ or the roles $\human$, $\platform$, or $\server$, as otherwise the functional property cannot hold. 
 For the same reason, we cannot remove the channel $(\human,\server)$ or the roles $\human$ or $\server$ in the topologies $\topo_6$ and $\topo_7$.
 We further argue that we cannot weaken any other assumptions either.
 
 \textit{Topology $\topo_1$:}
 When generating $\topo'$ by making the channel $(\human,\platform)$ default or by making $\platform$ untrusted, $\topo'$ cannot satisfy $\fullTime$ by Lemma~\ref{lemma:TopoReliableChannel}.
 Also, by Lemma~\ref{lemma:TopoEvidence}, the same is true when we make the channel $(\platform,\server)$ reliable only.
 
 \textit{Topology $\topo_2$:} 
 When generating $\topo'$ by making the channel 
 $(\platform,\server)$ default or by making $\platform$ untrusted, 
 $\topo'$ cannot satisfy $\fullTime$ by Lemma~\ref{lemma:TopoReliableChannel}.
 Also, by Lemma~\ref{lemma:TopoEvidence}, the same is true when we make 
 the channel $(\human,\platform)$ reliable only.
 
 \textit{Topology $\topo_3$:} 
 When generating $\topo'$ by making any of
 the channels $(\human,\platform)$ or $(\platform,\server)$ default, 
 $\topo'$ cannot satisfy $\fullTime$ by Lemma~\ref{lemma:TopoReliableChannel}.
 Also, by Lemma~\ref{lemma:TopoEvidence}, the same is true when we make
 $\platform$ of type $\trustedforward$ only.
 
 \textit{Topology $\topo_4$:} 
 When generating $\topo'$ by making the
 channels $(\human,\platform)$ or $(\platform,\server)$ default or by making $\platform$ untrusted,
 $\topo'$ cannot satisfy $\fullTime$ by Lemma~\ref{lemma:TopoReliableChannel}.
 Also, by Lemma~\ref{lemma:TopoEvidence}, the same is true when we make $\server$ untrusted, or when we make either of the channels $(\server,\platform)$ or $(\platform, \human)$ default.
 
 \textit{Topology $\topo_5$:} 
 When generating $\topo'$ by making the
 channels $(\human,\platform)$ or $(\platform,\server)$ default or by making $\platform$ untrusted,
 $\topo'$ cannot satisfy $\fullTime$ by Lemma~\ref{lemma:TopoReliableChannel}.
 Also, by Lemma~\ref{lemma:TopoEvidence}, the same is true when we make $\server$ untrusted, or when we make the channel $(\server,\platform)$ reliable only.
 
 \textit{Topology $\topo_6$:} 
 When generating $\topo'$ by making the
 channel $(\human,\server)$ reliable,
 $\topo'$ cannot satisfy $\fullTime$ by Lemma~\ref{lemma:TopoEvidenceUnified}.
 
 \textit{Topology $\topo_7$:} 
 When generating $\topo'$ by making the
 channel $(\human,\server)$ default,
 $\topo'$ cannot satisfy $\fullTime$ by Lemma~\ref{lemma:TopoReliableChannelUnified}.
 Also, by Lemma~\ref{lemma:TopoEvidenceUnified}, the same is true when we make $\server$ untrusted, or when we make the channel $(\server,\human)$ default.

 In all cases, there are no further options to weaken the given topology in a minimal way. 
\end{proof}

Finally, we use the above results to prove Theorem~\ref{theorem:completeChar}.

\begin{proof}[Proof of Theorem~\ref{theorem:completeChar}]
 The theorem's statement is a direct consequence of 
 Lemmas~\ref{lemma:possAndImposs} and~\ref{lemma:minimalTopo} and because $\topo_1,\dots,\topo_7$ are the only minimal topologies that satisfy all requirements necessary by Lemmas~\ref{lemma:TopoReliableChannel}--\ref{lemma:TopoEvidenceUnified}.
 The latter is established by an exhaustive case distinction on the finitely many possible topologies:
 each possible way of exchanging one assumption in a topology $\topo_1,\dots,\topo_7$ by another one, results either in a topology $\topo'$ that does not satisfy the conditions required by Lemmas~\ref{lemma:TopoReliableChannel}--\ref{lemma:TopoEvidenceUnified} or in a topology $\topo'$ that makes stronger assumptions than one of the seven topologies and is thus not minimal, i.e., $\exists \topo_j \in \{\topo_1,\dots,\topo_7\} \sd \topo_j \smallertopoNotEqual \topo'$.
\end{proof}

\subsection{Proof of Theorem~\ref{theorem:uniquenessImpliesDR}}
\label{appendix:securityproperties1}
We next prove Theorem~\ref{theorem:uniquenessImpliesDR} from Section~\ref{sec:uniqueness}, which states that $\uniqueness(\server)$ can be used to infer $\provoterAbstain(\server)$ in disputes~\stwo in protocols that satisfy certain conditions.

\begin{proof}[Proof of Theorem~\ref{theorem:uniquenessImpliesDR}]
Let $\prot$ be a protocol as assumed by the theorem.
Let $\tr$ be a trace in $\TR(\prot,\topo)$ such that $\tr \in \uniqueness(\server)$ and such that there is a voter $H$ for which $\verifyNoVote(H,\emptyset)\in \tr$.
All traces where no such signal $\verifyNoVote$ is recorded trivially satisfy $\provoterAbstain(\server)$ and as $\prot$ is a protocol without re-voting, $\verifyNoVote$ is only recorded for honest voters who abstain (that is, there are no signals $\verifyNoVote(H,x)$ for $x\neq \emptyset$).

We make a case distinction.
First, assume $$\neg (\exists \bs,\ballot \sd \BBrecorded(\bs) \in \tr \wedge \ballot \in \bs \wedge \CastBy(\ballot)=H).$$
Then, $\provoterAbstain(\server)$ holds by Definition~\ref{def:provoterAbstain}.

Second, assume $\exists \bs,\ballot \sd \BBrecorded(\bs) \in \tr \wedge \ballot \in \bs \wedge \CastBy(\ballot)=H$.
Thus, we consider a trace where
\begin{alignat*}{1}
&\tr \in \uniqueness(\server)
 \wedge \verifyNoVote(H,\emptyset)\in \tr
\\\multicolumn{2}{r}{$ 
 \wedge \BBrecorded(\bs) \in \tr \wedge \ballot \in \bs 
 \wedge \CastBy(\ballot)=H
$}
\\
\xRightarrow[]{(1)}&
\tr \in \uniqueness(\server)
\wedge \BBrecorded(\bs)\in \tr
 \\\multicolumn{2}{r}{$
\wedge \ballot \in \bs 
\wedge \CastBy(\ballot)=H
\wedge \neg \exists m,A \sd \send(H,A,m)\in \tr 
$}
\\\xRightarrow[]{(2)}&
\neg \exists \ballot'. {\ballot'\neq \perp \wedge} \tr \notin \verdict(\server,\ballot') 
\\\xRightarrow[]{(3)} &
\exists \ballot'.\tr \in \verdict(\server,\ballot').
\end{alignat*}
\noindent
In (1), we use that $\verifyNoVote$ is only recorded for an honest voter who abstains and thus, by assumption, does not send any message.
In (2) we use that $\uniqueness(\server)$ holds and therefore there cannot be any $\ballot'\neq \perp$ such that $\tr \notin \verdict(\server,\ballot')$ as this would require that $\send(H,A',m')$ is in the trace for some $A'$ and $m'$ such that $ m'\subterm \ballot$, which is false. 
Thus, for all ballots $\ballot'\neq \perp$ the trace $\tr$ must be in $\verdict(\server,\ballot')$ and Step (3) follows (for an arbitrary choice of a ballot $\ballot'\neq \perp$).
Thus, $\tr$ satisfies $\provoterAbstain(\server)$ by Definition~\ref{def:provoterAbstain}.
\end{proof}

\subsection{Analyzing a mixnet-based voting protocol with dispute resolution}
\label{appendix:AnalyzingAmixnetbasedVotingProtocolWithDisputeResolution}
We present the details of the protocol \protocolnameInst and show that it satisfies dispute resolution as well as other standard voting properties. 
To formalize the protocol, we first extend the protocol model from Section~\ref{sec:protocolModel} with new functions and equations.
Then, we present the communication topology and the detailed protocol specification.
Finally, we present the (instantiated) dispute resolution properties, introduce new properties, and analyze the protocol.
Our model and the definitions of the properties other than dispute resolution are closely related to~\cite{alethea}, where a voting protocol has been analyzed in a formalism supported by the Tamarin tool.

\subsubsection{New functions and equations}
\label{subsubsec:newFunctionsAndEquations}
Let $[l_1]$ and $[l_2]$ be two lists.
We denote by $[l_1]\subseteq_l [l_2]$ that the elements of $[l_1]$ are a sub-multiset of the elements of $[l_2]$.
Next, we introduce a generalization of the signature verification function $\ver$ that is applied to a list of signed messages and a list of verification keys, where the lists are of the same length. 
The function is denoted by $\verlist([s],[k])$ and if each signed message $[s]_i$ is verified with the verification key $[k]_i$, then the function returns the list of messages that were included in $[s]$ but with the signatures removed.
Otherwise, the function returns the default value $\perp$. 
 \scalebox{0.98}{ \parbox[]{\linewidth}{
 \begin{equation*}
 \verlist([s],[k]) \; := \;
 \begin{cases}
 [m], & \forall i \sd \ver([s]_i,[k]_i) = [m]_i \wedge [m]_i \neq \perp \\
 \perp, &\text{otherwise} .
 \end{cases}\label{eq:AL:verlist}
 \end{equation*} }}

\protocolnameInst uses non-interactive zero knowledge proofs to shuffle, i.e., randomly permute and re-encrypt, and decrypt ballots in a publicly verifiable manner.
First, we introduce a probabilistic asymmetric encryption scheme.
The encryption of the message $m$ under the public key $\pk$ and using the randomness $r$ is denoted by $\encp{m}{\pk}{r}$ and the decryption of the ciphertext $c$ with the private key $k$ by $\deccp(c,k)$.
The functions obey the equation $\deccp(\encp{m}{pk(k)}{r},k)=m$.
Then, we introduce $\zkp([\textit{i}],[\textit{o}],k)$ to denote the non-interactive zero knowledge proof that the list $[\textit{i}]$ was correctly shuffled and its elements correctly decrypted to the elements in $[\textit{o}]$.
Thereby, $k$ denotes the private key needed to decrypt the messages in $[\textit{i}]$ and to generate a correct proof.
A zero knowledge proof can be verified using the function $\verifyzkp(\zkproof,[\textit{i'}],[\textit{o'}],pk')$, which takes as input a proof $\zkproof$, two lists $[\textit{i'}]$ and $[\textit{o'}]$, where the latter presumably contains the decryptions of the former up to permutation, and a public key $pk'$.
Such a verification is successful if the following three conditions hold.
\begin{enumerate}
 \item $\zkproof$ is a non-interactive zero knowledge proof that was constructed with respect to permutations of the two lists that were input to the verification function ($[\textit{i'}]$ and $[\textit{o'}]$ above).
 \item The elements of the first list $[\textit{i'}]$ correspond to encryptions of the elements of the second list $[\textit{o'}]$, but can be permuted.
 \item The proof $\zkproof$ was constructed with the private key $k$ corresponding to the public key $pk'=\pk(k)$ used in the verification and the elements in the list $[\textit{i'}]$ are encrypted with the same key $pk'$.
\end{enumerate}
The last condition means that only an agent who possesses the private key to decrypt the messages in $[\textit{i'}]$ can construct a valid proof.
The following equation models these conditions, where $\pi_1$, $\pi_2$, and $\pi_3$ denote arbitrary permutations and $\pi[x]$ denotes that permutation $\pi$ is a applied to the list $[x]$.
\begin{alignat*}{1}\label{eq:AL:verifyzk}
 & \verifyzkp(\zkp(\pi_1[\encp{m}{\pk(k)}{r}],\pi_2 [m],k),
 \pi_3 [\encp{m}{\pk(k)}{r}],[m],
 \\\multicolumn{2}{r}{$ \pk(k))
 =\true.$}
\end{alignat*} 
These functions could, for example, be realized by the scheme of \emph{Groth}~\cite{grothShuffle}, who proposes a zero knowledge proof for the combined shuffle-and-decrypt operation. 
In~\cite{grothShuffle}, several decryption servers each possess a share of the private key corresponding to the public key used for the encryptions. Each server, in turn, shuffles the ballots and decrypts them with respect to its key share.
As we only consider one agent~$\server$, we can use \emph{Groth}'s scheme where $\server$ performs all servers' computations.

In addition to the new functions, we extend our protocol model with new signals, which we explain next while presenting the protocol.

\subsubsection{Protocol and communication topology}
\begin{figure}
 \begin{center}
 \scalebox{0.7}{
\begin{tikzpicture}[->,>=stealth',shorten >=1pt,auto,node distance=3cm and 1.5cm,semithick]

 \node[state] (H) {$\human$};
 \node[state, accepting] (D)[left=1cm of H] {$D$};
 \node[state, accepting,style=dashed](P)[right of=H] {$\platform$};
 \node[state, accepting,style=dashed] (S) [right of=P] {$\server$};
 
 \node[state, accepting
 ] (BB)[below=0.4cm of P] {$\BB$};
 \node[state, accepting
 ] (A) [below=0.4cm of H] {$\auditor$};

 \path 
 (H) edge[\edgeInsec,bend left]node[above left] {\reliable} (P) 
 (H) edge[\edgeSec,bend left]node[below] {\default} (D) 
 (D) edge[\edgeSec,bend left]node[above] {\default} (H)
 
 (P) edge[\edgeInsec,bend left]node[above right] {\reliable} (S) 
 (S) edge[\edgeInsec]node[below] {\reliable} (P)
 (P) edge[\edgeInsec]node[below] {\reliable} (H)
 
 (S) edge[\edgeAuth
 ] node[below] {\default} (BB)
 (BB)edge[\edgeAuth
 ] node[below] {\default} (H)
 edge[\edgeAuth
 ] node[below left] {\default} (A)
 ;
\end{tikzpicture}}
 \caption{The topology $\topoInst$ for the protocol \protocolnameInst.}\label{fig:topoAlethea}
 \end{center}
\end{figure}

\begin{figure*}
\begin{center}
 \scalebox{0.8}{\includegraphics{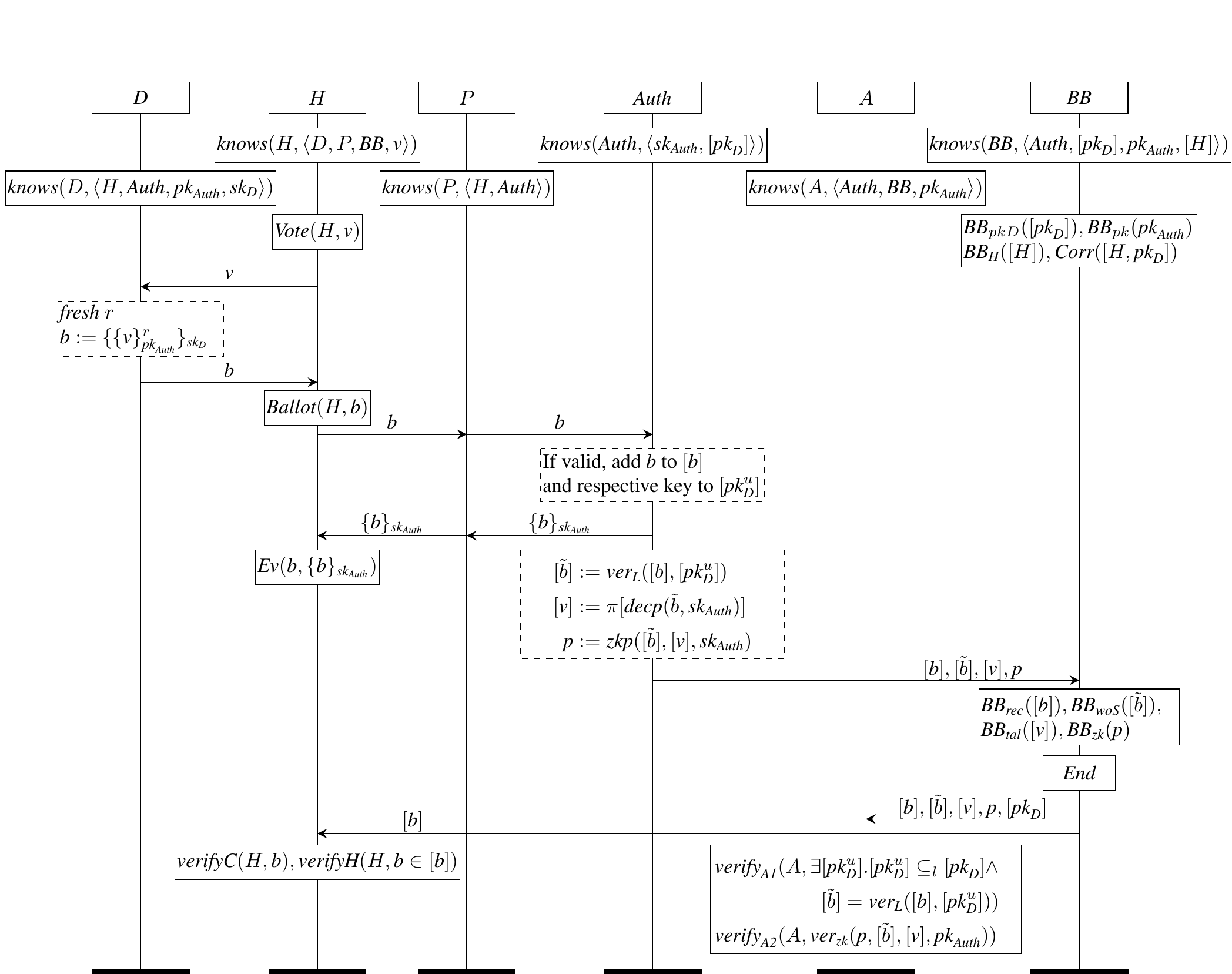}}
 \scalebox{0.9}{ \parbox[c]{\textwidth}{
 \begin{alignat*}{1}
 \textit{For } \ballot = \perp: &\; \verdict(\server, \ballot):=\{ \}.
 \\
\textit{For }\ballot \neq \perp: &\;
\verdict(\server, \ballot):=\{\tr |
(\exists [\ballot], \pkS, \confirmation \sd \BBpkS(\pkS) \in \tr \wedge \Evidence(\ballot,\confirmation) \in \tr 
\wedge \ver(\confirmation, \pkS) = \ballot
\wedge \BBrecorded([\ballot]) \in \tr
\wedge \ballot \notin [\ballot])
\\&\multicolumn{2}{r}{$
\vee 
(\exists [\ballot],[\pkD],[\ballotPrime] \sd
\BBrecorded(\bs)\in \tr 
\wedge \BBpkD([\pkD])\in \tr
\wedge \BBballotWoSignature([\ballotPrime])\in \tr 
 \wedge (\neg \exists [\pkDused] \sd
[\pkDused] \subseteq_l [\pkD] 
\wedge [\ballotPrime]=\verlist(\bs, [\pkDused]) ))\}.
$}
\end{alignat*}
}}
\end{center}
 \caption{The protocol \protocolnameInst, or $\protInst$, where $\pkD =\pk(\skD)$, $\pkS=\pk(\skS)$, and $\CastBy(\ballot)= H$ holds iff $ \exists \pk \sd\ver(\ballot,\pk)\neq \perp \wedge \langle H,\pk\rangle \in [H,\pk] \wedge \corresponds([H,\pk])\in \tr$.
The protocol's setup specifies a single agent $\server$, that each voter $H$ is associated with a unique trusted off-line device $D$, and that there is no restriction on the relation between voters $H$ and platforms $P$, i.e., several voters can be connected with the same $P$.
The role for a voter $H$ who abstains consists of receiving the list of recorded ballots from the bulletin board followed by the signal $\verifyNoVote(H,\emptyset)$.
 }\label{fig:protAlethea}
\end{figure*}

\protocolnameInst's communication topology $\topoInst$ and detailed protocol steps are respectively depicted in Figures~\ref{fig:topoAlethea} and~\ref{fig:protAlethea}.
As explained in Section~\ref{subsubsec:mixnet:protocol}, we present the protocol as a message sequence chart. 

The topology specifies roles for the trusted off-line device $D$, only connected to the voter $H$, and the platform $P$, through which $H$ can access $\server$. 
Apart from the fixed assumptions about the roles $\BB$, $\auditor$ and their incoming and outgoing channels, the topology specifies the following assumptions.
We assume that the channels $(H,D)$ and $(D,H)$ are secure and default. Moreover, the network between $H$ and $\server$ is insecure but the channels $(H,P)$, $(P,\server)$, $(\server,P)$, and $(P,H)$ reliably deliver messages. 
$P$ is partially trusted to forward messages and $\server$ is partially trusted to reply with a confirmation when it receives a valid ballot. Except for the device $D$, the topology is as the topology $\topo_4$ (Figure~\ref{fig:topo4}) and can be interpreted as explained in Section~\ref{subsec:topologiesprovidingfulltime}.

As \server naturally learns how each voter voted when collecting and tallying the ballots, 
there are no privacy guarantees if \server is not trusted.
For simplicity, and as other authors~\cite{helios,alethea,Chaum}, we thus model a trusted \server when examining privacy properties. In reality, this trust could be distributed as is done in other work, e.g.,~\cite{civitas,eperio}.

We assume that, at each point in time, only one election takes place, i.e., there are no parallel sessions.
Moreover, the protocol's setup specifies that there is a single agent $\server$ and that each voter $H$ has a unique trusted device $D$ to which he has exclusive access. In contrast, there is no fixed relation between platforms and voters, i.e., several voters may use the same platform $P$ to cast their ballot.
For readability,
we explain the protocol for one voter, device, and platform in Figure~\ref{fig:protAlethea}.
At the protocol's start, the device and authority each know a unique private key ($\skD$ and $\skS$, respectively) and know each others corresponding public key, which is also published on the bulletin board $\BB$.
Moreover, $\BB$ contains the list of all eligible voters $[H]$ and information on which signing key corresponds to which eligible voter, i.e., which signing key is installed on that voter's device.
We denote the later by the signal $\corresponds([H,\pkD])$, where each pair $\langle H,\pkD \rangle$ in the list denotes that the signing key corresponding to $\pkD$ is installed on $H$'s device.

To compute a ballot, a voter enters his vote on his device $D$.
$D$ generates a fresh random number $r$ (denoted by $\fresh$), uses $r$ to probabilistically encrypt the vote under $\server$'s public key, and signs the resulting encryption with its signing key.
$D$ displays the computed ballot to $H$ and $H$ casts the ballot (denoted by the signal $\Ballot$) by entering it on his platform $P$ from where it is forwarded to $\server$.
This communication could, for example, be realized by $D$ displaying the ballot as a QR code and $H$ scanning this code with $P$.
Upon receiving a ballot, $\server$ checks whether it is well-formed and signed with an eligible voter's key.
If this is the case, $\server$ adds the ballot to the list of recorded ballots $[\ballot]$ and the respective verification key to $[\pkDused]$ to keep track of the used public keys, as each voter can only vote once.
Also, $\server$ sends as confirmation the signed ballot back to $H$ via $P$, where the voter keeps it as evidence in case of subsequent disputes.

After the voting phase ends, $\server$ computes the tally.
Instead of the function $\Tally$, we describe this procedure in terms of several, more detailed steps:
$\server$ removes the signatures from the recorded ballots, shuffles and decrypts the resulting encryptions with its private key, and produces a zero knowledge proof that these operations were done correctly. $\server$ then publishes all the lists and the proof on the bulletin board, where an auditor can read them.

The auditor first checks that all recorded ballots contain a signature corresponding to a unique public key in $[\pkD]$.
Then, the auditor verifies the zero knowledge proof.
We denote this respectively by the explicit signals $\verifyValidSig(A,p_1)$ and $\verifyProof(A,p_2)$, which record that the agent $A$ checks whether the predicates $p_1$ and $p_2$ hold.
In the protocol's traces, these signals are recorded as $\verifyValidSig(A,p_1,t_1)$ and $\verifyProof(A,p_2,t_2)$, with $t_1,t_2\in\{\true,\false\}$ respectively indicating whether the predicate $p_1$ and $p_2$ is satisfied in this particular trace. 
This allows us to refer both to the terms that are evaluated in the predicate and the predicate's truth value during execution. 

Moreover, the voter reads the list of recorded ballots from $\BB$ and checks whether his ballot is contained in this list.
In addition to $\verifyCast(H,b)$, this is recorded by the signal $\verifyIV(H,\ballot \in [\ballot])$, which we will need to express individual verifiability.
As with $\verifyValidSig$ and $\verifyProof$, $\verifyIV$ denotes the agent and the checked predicate and contains in the protocol's trace a third argument stating whether the predicate is satisfied.

Finally, a voter who abstains does not send any messages. After the results are published, he reads from the bulletin board the list of recorded ballots and believes at that step that no ballot is recorded for him, which is recorded by the signal $\verifyNoVote(H,\emptyset)$.

We complete the protocol specification as follows.
A ballot is considered to be cast by a voter when the ballot's signature can be verified with the public key associated with this voter. 
That is $\CastBy(\ballot)= H$ iff there exists a verification key $\pk$ such that $\ver(\ballot,\pk)\neq \perp \wedge \langle H,\pk\rangle \in [H,\pk] \wedge \corresponds([H,\pk])\in \tr$.
Moreover, we define the verdict that $\server$ behaved dishonestly as the set of traces where a) there exists evidence consisting of a ballot signed by $\server$ but this ballot is not included in the recorded ballots on the bulletin board or b) there are recorded ballots that are not signed by a unique eligible voter. 
The corresponding set $\verdict$ is defined in Figure~\ref{fig:protAlethea}.

\subsubsection{Security properties}\label{caseStudy:SecurityProperties}
With the above instantiation of the function $\verdict$, all introduced dispute resolution properties are now defined and can be analyzed.
As the definition of $\uniqueness(\server)$ additionally contains $\CastBy$, we explicitly instantiate this part for simplicity.
\begin{mydef}\label{def:uniquenessMV}
Let the length of the list $\bs$ be $n$.
 \begin{alignat*}{1}
&\uniqueness(\server) := \; \{\tr \mid 
 \ballot \neq \perp \wedge
 \tr \notin \verdict(\server,\ballot)
 \\\multicolumn{2}{r}{$ 
 \wedge \BBrecorded(\bs) \in \tr
 \wedge i \in \{1,\dots,n\}
 \wedge j \in \{1,\dots,n\}
 \implies
 $}
 \\\multicolumn{2}{r}{$ 
 \exists \Hs,i',j',[\pk],{\pk}_1,{\pk}_2, A_1, A_2,m_1,m_2 \sd
 \BBH(\Hs) \in \tr 
 $}
 \\\multicolumn{2}{r}{$ 
 \wedge \corresponds([H,\pk]) \in \tr
 \wedge\; \ver(\bs_i,{\pk}_1)\neq \perp 
 \wedge\; \ver(\bs_j,{\pk}_2)\neq \perp 
 $}
 \\\multicolumn{2}{r}{$ 
 \wedge \langle \Hs_{i'},{\pk}_1 \rangle \in[H,\pk]
 \wedge \langle \Hs_{j'},{\pk}_2 \rangle \in[H,\pk]
 $}
 \\\multicolumn{2}{r}{$ 
 \wedge \send(\Hs_{i'},A_1,m_1)\in \tr
 \wedge \send(\Hs_{j'},A_2,m_2)\in \tr 
 $}
 \\\multicolumn{2}{r}{$ 
 \wedge m_1\subterm \bs_i
 \wedge m_2 \subterm \bs_j 
 \wedge (i \neq j \implies \Hs_{i'} \neq \Hs_{j'})
 \}. 
 $}
\end{alignat*}
\end{mydef}
Compared to the more general definition of $\uniqueness(\server)$ from Definition~\ref{def:uniqueness}, the function $\CastBy$ is instantiated in the fourth and fifth line.

In addition to the dispute resolution properties, we analyze standard verifiability and privacy properties that we introduce next.
We start with individual verifiability, or \indivVerif for short, that states that whenever a voter verifies that his ballot is in the list of recorded ballots, then indeed one of the recorded ballots corresponds to his vote.
The property's definition is based on individual verifiability due to~\cite{alethea} and \cite{kremerVerif}.
\begin{mydef}\label{def:AL:indivVerifvote}
\begin{alignat*}{1}
&\textit{\indivVerif}:=
\{\tr \mid
 \verifyIV(H,\ballot\in[\ballot],\true) \in \tr \wedge 
 \\\multicolumn{2}{r}{$
 \Vote(H,\vote) \in \tr \implies 
\exists [\ballot'], \pkS,r,\skD \sd
 $}
 \\\multicolumn{2}{r}{$
\BBrecorded([\ballot'])\in \tr \wedge
 \ballot \in [\ballot'] \wedge 
 \ballot=\{\encp{v}{\pkS}{r}\}_{\skD} \}.
 $}
\end{alignat*}
\end{mydef}

Next, we introduce the universal verifiability property \emph{Tallied-as-recorded} that states that any auditor can verify that all recorded ballots are counted correctly in the final tally.
In particular, if an auditor performs its specified checks on some lists, then the bulletin board contains the same lists of ballots, votes, and verification keys, the votes in the final tally correspond to the votes encrypted in the recorded ballots, and each ballot contains a signature associated with a different verification key in $[\pkD]$. 
The order of the votes and ballots can be permuted, as indicated by the permutation $\pi$.
\begin{mydef}
\begin{alignat*}{1}
&\textit{\TalliedAsRecorded}:=
\{\tr \mid
 \verifyValidSig(A,\exists [\pkDused]. 
\\
\multicolumn{2}{r}{$ 
[\pkDused]\subseteq_l [\pkD] \wedge
[\ballotPrime]=\verlist(\bs, [\pkDused]),\true) \in \tr 
$}
\\
\multicolumn{2}{r}{$ 
\wedge \verifyProof(A, \verifyzkp(\zkproof,[\ballotPrime],[\vote],\pkS),\true) \in \tr \implies 
$}
\\
\multicolumn{2}{r}{$ 
\exists [r],[\skD],\pi \sd 
 \BBrecorded([\ballot])\in \tr \wedge
 \BBtallied([\vote])\in \tr \wedge
$}
\\
\multicolumn{2}{r}{$ 
 \BBpkD([\pkD])\in \tr
 \wedge 
[\ballot]=\pi [\{\encp{\vote}{\pkS}{r}\}_{\skD}]
$}
\\
\multicolumn{2}{r}{$ 
\wedge [\pk(\skD)] = [\pkDused]
\wedge [\pkDused] \subseteq_l [\pkD]\}.
$}
\end{alignat*}
\end{mydef}
\noindent
We then define \emph{end-to-end verifiability}, or \verifEnd for short, as the conjunction of individual verifiability and tallied-as-recorded.
\begin{mydef} 
 $$\textit{\verifEnd} := \textit{\indivVerif} \cap \textit{\TalliedAsRecorded}.$$
\end{mydef}

\emph{Eligibility verifiability} is another universal verifiability property that is defined by \emph{Kremer et al.}~\cite{kremerVerif} as the property that \emph{``anyone can check that each vote in the election outcome was cast by a registered voter and there is at most one vote per voter.''}
We denote the property by \eligibilityVerif and define it conditional on an auditor $A$ performing all of its specified verifiability checks.
If these checks are verified, then each vote in the published final tally must have been sent by an eligible voter $H$. 
Furthermore, no two votes in the final tally correspond to the same voter.
\begin{mydef}
Let $n$ be the length of the list $\vs$.
\begin{alignat*}{1}
 &\textit{\eligibilityVerif}:= 
\{\tr \mid
 \verifyValidSig(A,\exists [\pkDused]. 
[\pkDused]\subseteq_l [\pkD] 
\\
\multicolumn{2}{r}{$ 
\wedge [\ballotPrime]=\verlist(\bs,[\pkDused]),\true) \in \tr 
$}
\\\multicolumn{2}{r}{$ 
\wedge \verifyProof(A, \verifyzkp(\zkproof,[\ballotPrime],[\vote],\pkS),\true) \in \tr $}
\\\multicolumn{2}{r}{$ 
\wedge i \in \{1,\dots,n\} \wedge j \in \{1,\dots,n\}
$}
\\
\multicolumn{2}{r}{$ 
\implies 
\exists [H], i', j',A,A' \sd \BBH([H]) \in \tr \wedge \BBtallied([\vote])\in \tr 
 $}\\
\multicolumn{2}{r}{$ 
 \wedge \send([H]_{i'},A,[\vote]_{i}) \in \tr
\wedge \send([H]_{j'},A',[\vote]_{j}) \in \tr
 $}\\
\multicolumn{2}{r}{$ 
\wedge (i \neq j \implies [H]_{i'} \neq [H]_{j'}) 
\}.
$}
\end{alignat*}
\end{mydef}
\noindent
Note that \uniqueness and \eligibilityVerif make similar guarantees in that both state that each element of a list is associated with a \emph{unique eligible voter}.
\eligibilityVerif states this guarantee with respect to all votes in the final tally.
In contrast, \uniqueness states the guarantee for the \emph{recorded} ballots as these are the relevant terms for the disputes.

Next, we introduce the privacy property \emph{receipt-freeness}.
Whereas \emph{vote privacy} denotes that an adversary cannot link voters to their votes, receipt-freeness is strictly stronger and additionally requires that this holds even when a voter reveals his secrets to the adversary, i.e., when he tries to provide a receipt of his vote.
This property is important, as a protocol should ensure privacy even when honest voters are forced by an adversary to reveal private information.

As~\cite{alethea} and similarly to~\cite{delaune}, we define receipt-freeness as an observational equivalence property, which states that an adversary cannot distinguish between two systems.
Concretely, we define a \emph{left system} where the voter $A$ votes $v_1$ and the voter $B$ votes $v_2$ and a \emph{right system} where $A$ votes $v_2$ and $B$ votes $v_1$.
Moreover, we change the original protocol $\prot$ to $\prot'$ where $A$ sends all his secrets to the adversary except that $A$ always claims that his vote is $\vote_1$ (which is only true in the left system). 
We then define a set $\mathcal{S}$ of trace pairs $(tr_L,tr_R)$ where $tr_L$ is from the left system and $tr_R$ from the right system.
A protocol $\prot$ satisfies receipt-freeness if for all traces of $\prot'$ in one system, there exists a trace in the other system such that the pair of traces is contained in $\mathcal{S}$.

We use the notation $\prot_{m_1\leftarrow m_1',m_2\leftarrow m_2'}$ to denote the specification of the protocol $\prot$ where each occurrence of the terms $m_1$ and $m_2$ is replaced by $m_1'$ and $m_2'$, respectively. 
Furthermore, we write $t_1\approx t_2$ to denote that two traces are indistinguishable for an adversary. We use this definition informally here and refer to~\cite{tamarinDiff} for a formal definition of $\approx$.
\begin{mydef}
Let $\prot'$ be the protocol obtained from $\prot$ as described above, let $\vote_A$ and $\vote_B$ be the term that denotes $A$'s and $B$'s vote, respectively, and let $\vote_1$ and $\vote_2$ be message terms. 
Receipt-freeness of the protocol $\prot$ run in the topology $\topo$ is defined as follows.
\begin{alignat*}{3}
 \textit{\receiptfree}:=&
 \{(\tr_L,\tr_R)\in 
 \\\multicolumn{2}{r}{$ 
 \TR(\prot'_{\vote_A\leftarrow \vote_1,\vote_B\leftarrow \vote_2},\topo)\times 
 \TR(\prot'_{\vote_A\leftarrow \vote_2,\vote_B\leftarrow \vote_1},\topo)
 $}
 \\\multicolumn{2}{r}{$ 
 \mid \tr_L \approx \tr_R \}.
 $}
 \end{alignat*}
\end{mydef}\noindent
This set defines all indistinguishable trace pairs such that the traces are from two systems where $A$ and $B$ vote the opposite way. $A$ reveals all secrets, except he claims in both systems that he votes $v_1$, which is only true in the left system. 

\subsubsection{Analysis}\label{appendix:mixvote:analysis}
As mentioned in Section~\ref{subsec:CaseStudy}, by Theorem~\ref{theorem:completeChar} it is possible to achieve \fullTime in the topology $\topoInst$ 
as $\topo_4 \smallertopo \topoInst$.
We next analyze whether \protocolnameInst in Figure~\ref{fig:protAlethea}, $\protInst$ for short, indeed satisfies all above properties when run in the 
topology $\topoInst$.
Recall that $\provoterCast(\server)$, $\timeliness(\server)$, and $\provoterAbstain(\server)$ are guarantees for the voter and should hold for a distinguished voter $H$, even when the authority $\server$ and all other voters are dishonest.
Similarly, $\uniqueness(\server)$ and the verifiability properties must hold even when the authority is dishonest.
In contrast, $\proauth(\server)$ is a guarantee for $\server$ and should hold even when all voters are dishonest. 
Finally, recall that we assume for receipt-freeness that $\server$ is honest.
\begin{theorem}
 \protocolnameInst satisfies 
 $\provoterCast(\server)$, $\timeliness(\server)$, $\provoterAbstain(\server)$, $\proauth(\server)$, 
 $\uniqueness(\server)$, \indivVerif, \TalliedAsRecorded, \verifEnd, \eligibilityVerif, and \receiptfree when run in the topology $\advModel{\topoInst}{\hShH}$.
 When run in the topology $\advModel{\topoInst}{\mShH}$, the protocol satisfies the same properties, except for $\proauth(\server)$ and \receiptfree.
 When run in the topology $\advModel{\topoInst}{\hSmH}$ the protocol satisfies $\proauth(\server)$.
\end{theorem}
\begin{proof}
To prove the theorem, we combine proofs by Tamarin with pen-and-paper proofs.
 We establish all properties except for $\provoterAbstain(\server)$, \verifEnd, and \receiptfree for one voter who casts a vote in Tamarin and model each of the topologies $\advModel{\topoInst}{\hShH}$, $\advModel{\topoInst}{\mShH}$, and $\advModel{\topoInst}{\hSmH}$ in a separate Tamarin theory.
 Moreover, to prove \receiptfree we model two voters as explained above and automatically prove the property in the topology $\advModel{\topoInst}{\hShH}$ using Tamarin's built in support for observational equivalence~\cite{tamarinDiff}.
 All Tamarin files can be found in~\cite{tamarinfiles}.
 
 As Tamarin requires specifying a fixed number of voters, we prove by hand in the Lemmas~\ref{lemma:protocolInst_timeliness}--\ref{lemma:protocolInst_EligibilityVerif} below that the properties 
 $\timeliness(\server)$, $\provoterCast(\server)$, $\proauth(\server)$, \indivVerif, \TalliedAsRecorded, $\uniqueness(\server)$, and \eligibilityVerif also hold for an arbitrary number of voters in the topologies claimed by the theorem.
 We then show in Lemma~\ref{lemma:protocolInst_provoterAbstain} that $\provoterAbstain(\server)$ is implied by $\uniqueness(\server)$ in the required topologies, as Theorem~\ref{theorem:uniquenessImpliesDR} can be applied.
 Finally, as \verifEnd is the intersection of \indivVerif and \TalliedAsRecorded, it directly follows that \verifEnd also holds in the required adversary models.
\end{proof}

\begin{lemma}\label{lemma:protocolInst_timeliness}
 \protocolnameInst satisfies $\timeliness(\server)$ when run in the topologies $\advModel{\topoInst}{\hShH}$ and $\advModel{\topoInst}{\mShH}$.
\end{lemma}
 
\begin{proof}
 Let $H$ be the distinguished voter for which the property should hold, independently of whether or not there are other (dishonest) voters in the system.
 Let $\tr$ be a trace in $\TR(\protInst,\advModel{\topoInst}{\hShH})\cup\TR(\protInst,\advModel{\topoInst}{\mShH})$ such that $\tr=\tr'\cdot\tr''$ and $ \Ballot(H,\ballot) \in \tr'$ and $\End \in \tr''$.
 \begin{alignat*}{1}
& \tr=\tr'\cdot\tr'' \wedge \Ballot(H,\ballot) \in \tr' \wedge\End \in \tr''
\\
\xRightarrow{(1)} &
\exists P,\bs \sd \send(H,P,\ballot) \in \tr \wedge \BBrecorded([\ballot]) \in \tr'
\\
\xRightarrow{(2)} &
\send(H,P,\ballot) \in \tr \wedge 
\receive(H,P,\ballot) \in \tr 
\\\multicolumn{2}{r}{$
\wedge \send(P,\server,\ballot) \in \tr
$}
\\
\xRightarrow{(3)} &
\send(H,P,\ballot) \in \tr \wedge 
\receive(P,\server,\ballot) \in \tr 
\\
\xRightarrow{(4)} &
\exists \skS \sd 
\send(H,P,\ballot) \in \tr \wedge 
\send(\server,P,\{\ballot\}_{\skS}) \in \tr 
\\\multicolumn{2}{r}{$
\wedge \BBpkS(\pk(\skS)) \in \tr
$}
\\
\xRightarrow{(5)} &
\send(H,P,\ballot) \in \tr \wedge 
\receive(\server,P,\{\ballot\}_{\skS}) \in \tr
\\\multicolumn{2}{r}{$
\wedge \send(P,H,\{\ballot\}_{\skS}) \in \tr
\wedge \receive(P,H,\{\ballot\}_{\skS}) \in \tr
$}
\\
\xRightarrow{(6)} &
\Evidence(\ballot, \{\ballot\}_{\skS}) \in \tr
\end{alignat*} 
The first conjunct of Step (1) holds as $\Ballot(H,\ballot)$ is recorded when $H$ sends the ballot.
Note that the honest voter $H$ only sends this message on the channel $(H,P)$.
The second part of Step (1) holds as $\End$ is be definition recorded after the last message is published on the bulletin board. 
Thus, for some list $[\ballot]$, $\BBrecorded([\ballot])$ is recorded in the trace before $\End$.
Step (2) follows by the fact that in both topologies $\advModel{\topoInst}{\hShH}$ and $\advModel{\topoInst}{\mShH}$, it holds that $\channel(H,P)=\;\insecChanReliable$, $\trust(H)=\textit{trusted}$, and $\trust(P)=\trustedforward$. Thus, by the restrictions from Section~\ref{subsubsec:channelTypes}, each message that is sent from $H$ to $P$ will be received by $P$.

Moreover, by the fact that $\trust(P)=\trustedforward$ $P$ always forwards messages correctly, i.e., $P$ forwards the ballot to the (unique) authority $\server$.
Similarly to Step (2), Steps (3) and (4) follow by the assumption that the channel from $P$ to $\server$ is reliable and that $\trust(\server)=\trustedreply$. 
In particular, as $\server$ is partially trusted to answer with the (correct) confirmation for a valid ballot, it sends back the ballot signed with its (correct) signing key (as $H$ and his device $D$ are honest, the ballot $\ballot$ is valid).
By the shared initial knowledge, the setup assumptions, and since the bulletin board is honest by the topology $\topoInst$, the public key corresponding to $\server$'s signing key is published on the bulletin board (second part in Step (4)).
In Step (5), we use the same reasoning as in Steps (2) and (3) for the channels $(\server,P)$ and $(P,H)$ which are also reliably by $\topoInst$.
By the protocol specification, when the honest $H$ receives the confirmation, the signal $\Evidence$ is recorded, where the first argument is the ballot that $H$ sent and the second the received confirmation. Thus (6) holds.

Finally, we make a case distinction.
First, let $\ballot \in \bs$. 
Then, it immediately follows that $\timeliness(\server)$ holds.
Second, let $\ballot \notin \bs$. 
Then, it holds that $\BBpkS(\pk(\skS)) \in \tr \wedge \Evidence(\ballot, \{\ballot\}_{\skS}) \in \tr \wedge \ver( \{\ballot\}_{\skS}, \pk(\skS)) = \ballot \wedge \BBrecorded([\ballot]) \in \tr \wedge \ballot \notin [\ballot]$ and thus, by the definition of $\verdict$ it holds that $\tr \in \verdict(\server,\ballot)$. Therefore, $\timeliness(\server)$ also holds in this case.
\end{proof}

\begin{lemma}\label{lemma:protocolInst_provoterCast}
 \protocolnameInst satisfies $\provoterCast(\server)$ when run in the topologies $\advModel{\topoInst}{\hShH}$ and $\advModel{\topoInst}{\mShH}$.
\end{lemma}
\begin{proof}
Let $H$ be the distinguished voter for which the property should hold, independently of whether there are other (dishonest) voters in the system.
Let $\tr$ be a trace in $\TR(\protInst,\advModel{\topoInst}{\hShH})\cup\TR(\protInst,\advModel{\topoInst}{\mShH})$ such that $\verifyCast(H,\ballot) \in \tr$.
 \begin{alignat*}{2}
 &\verifyCast(H,\ballot) \in \tr
 \\
\xRightarrow{(1)} &
\exists P , [\ballot],\BB \sd \send(H,P,\ballot) \in \tr
\wedge \receive(\BB,H,[\ballot]) \in \tr
 \\ 
\xRightarrow{(2)} &
 \send(H,P,\ballot) \in \tr
\wedge \send(\BB,H,[\ballot]) \in \tr
\\ 
\xRightarrow{(3)} &
\send(H,P,\ballot) \in \tr
\wedge \BBrecorded([\ballot]) \in \tr
\end{alignat*} 
As $H$ is honest in $\advModel{\topoInst}{\hShH}$ and $\advModel{\topoInst}{\mShH}$, he follows his role specification. 
Thus, when $\verifyCast(H,\ballot)$ is in the trace, so are the signals $\send(H,P,\ballot)$ for some $P$ and $\receive(\BB,H,[\ballot])$ for some $\BB$ and $\bs$ as these signals are preceding the verifiability check in $H$'s role (Step (1)).
By the topology $\topoInst$, the channel $(\BB,H)$ is authentic. Therefore, $\BB$ must have sent the message $\bs$ for $H$ to receive it on this channel (Step (2)).
Moreover, $\BB$ is honest by the topology $\topoInst$. Thus, $\BB$ only sends one message containing a single list, which was previously recorded in the signal $\BBrecorded$. In particular, as there are no parallel sessions, this list cannot be confused with another list sent by $\BB$ and it follows in Step (3) that $\BBrecorded(\bs)\in \tr$.

We now have a trace with the same signals as in the proof of Lemma~\ref{lemma:protocolInst_timeliness} after Step (1). 
We can thus apply the same steps as in the Proof of Lemma~\ref{lemma:protocolInst_timeliness} to conclude that $\provoterCast(\server)$ holds in all the required topologies.
(Note that in the proof of Lemma~\ref{lemma:protocolInst_timeliness}, $\BBrecorded$ is in the subtrace $\tr'$ rather than in $\tr$, which is however not needed in the further proof steps.)
\end{proof}

\begin{lemma}\label{lemma:protocolInst_proauth}
 \protocolnameInst satisfies $\proauth(\server)$ when run in the topologies $\advModel{\topoInst}{\hShH}$ and $\advModel{\topoInst}{\hSmH}$.
 \end{lemma}
 
\begin{proof}
Let $\tr$ be a trace in $\TR(\protInst,\advModel{\topoInst}{\hShH})\cup\TR(\protInst,\advModel{\topoInst}{\hSmH})$, such that $\honest(\server)\in \tr$.
 We show that, for all ballots $\ballot$, this trace is not in $\verdict(\server,\ballot)$.
First, consider $\ballot = \perp$. As $\verdict(\server,\perp)=\{\}$ and as we just defined that $\honest(\server)\in \tr$, it holds that $\tr \notin \verdict(\server,\perp)$.

We thus consider only ballots $\ballot$ different from $\perp$ in the following and show that $\tr$ is not in $\verdict(\server,\ballot)$ by separately showing that $\tr$ satisfies neither of the two disjuncts of the definition of $\verdict(\server,\ballot)$ for ballots $\ballot \neq \perp$.
We start with the second disjunct and assume that the corresponding lists are published on the bulletin board. If this is not given, the second disjunct is false and we are done.
\begin{alignat*}{1}
 & \BBrecorded(\bs)\in \tr \wedge \BBpkD([\pkD])\in \tr\wedge \BBballotWoSignature([\ballotPrime])\in \tr 
 \\
 \xRightarrow{(1)} &
 \exists \skS \BB, [\vote], \zkproof\sd 
 \knows(\server,\langle \skS, [\pkD]\rangle) \in \tr
 \\
 \multicolumn{2}{r}{$
 \wedge
 \send(\server,\BB,\langle [\ballot], [\ballotPrime], [\vote], \zkproof \rangle) \in \tr
 $}
\\
 \xRightarrow{(2)} &
 \exists [\pkDused]\sd
 [\ballotPrime]=\verlist(\bs, [\pkDused])
 \wedge [\pkDused] \subseteq_l [\pkD]
 \end{alignat*} 
 \noindent
 The first part of Step (1) holds by the shared initial knowledge specified in the protocol and as $\BB$ is honest by the topology $\topoInst$. 
 Thus, $\BB$ and $\server$ share the same list $[\pkD]$ in their initial knowledge that $\BB$ has published in the signal $\BBpkD$.
 Also, there are no parallel sessions so no message confusion is possible.
 The second part of Step (1) holds as the honest $\BB$ only records the lists $\bs$ and $[\ballotPrime]$ after it has received them from $\server$ as the first two of four messages.
 As the channel $(\server,\BB)$ is authentic by topology $\topoInst$, $\server$ must thus have sent these messages for $\BB$ to receive them.
 
 Next, as $\server$ only sends one message to $\BB$ and as there are no parallel sessions, the lists sent by $\server$ respectively correspond to the correctly computed ballots, ballots without signatures, tallied votes, and the zero knowledge proof.
 Thus, in particular, each recorded ballot is signed by a unique eligible voter's signing key, hence Step (2) holds. (Note that the honest $\server$ uses the list $[\pkD]$ from its initial knowledge for its check whether a ballot is valid.)
 This is a contradiction to the second disjunct in $\verdict$'s definition, which thus cannot be satisfied by the trace~$\tr$.
 
 Next, we show that the first disjunct in the definition of $\verdict(\server,\ballot)$ for ballots $\ballot \neq \perp$ is also not satisfied by $\tr$.
Let $S$ be an explicit signal, i.e., a top-level function in our protocol model, and let $x$ and $y$ be arbitrary terms. 
Assume that there is a signal in the trace $\tr$, which contains a subterm 
signed by \server's signing key. 
 Recall that $m_2\subterm m_1$ denotes that $m_1$ is a subterm of the composed message $m_2$.
\begin{alignat*}{1}
 & \honest(\server) \in \tr 
 \wedge
 S(x) \in \tr \wedge x \subterm \{y\}_{\skS}
\\
\multicolumn{2}{r}{$
 \wedge \BBpkS(\pk(\skS))\in \tr$}
 \\
 \xRightarrow{} &
 \exists A,x' \sd \send(\server,A,x') \in \tr \wedge x' \subterm \{y\}_{\skS}
 \end{alignat*} 
 The implication holds since, by the initial knowledge specified in the protocol, only $\server$ knows the signing key $\skS$ corresponding to the public key published in $\BBpkS$ by $\BB$ (which is honest by topology $\topoInst$). 
 Furthermore, no one can learn $\skS$ during the protocol execution and no explicit signal is recorded in $\server$'s role. 
 In particular, the former holds as $\server$ only sends out messages containing $\skS$ where $\skS$ is used as a signing key, which cannot be extracted by our equational theory. 
 It follows that $\server$ must have sent a term containing $\{y\}_{\skS}$ for it to be recorded in an explicit signal in the trace and thus the implication holds.
 
 Next, we make a case distinction.
 First, suppose that there are no recorded ballots on the bulletin board, that is $\neg \exists [\ballot] \sd \BBrecorded([\ballot])\in \tr $.
 Then it follows that $\forall \ballot \sd \tr \notin \verdict(\server,\ballot)$.
 
 Second, suppose that there are some recorded ballots on the bulletin board.
 \begin{alignat*}{1}
 &
 \exists [\ballot] \sd 
 \send(\server,A,x') \in \tr
 \wedge x' \subterm \{y\}_{\skS} \wedge 
 \\\multicolumn{2}{r}{$ 
 \BBrecorded([\ballot])\in \tr 
 $}
\\ \xRightarrow{(1)} &
 \exists \BB, [\ballotPrime], [\vote], \zkproof \sd
 \send(\server,A,x') \in \tr
 \wedge x' \subterm \{y\}_{\skS}
 \\\multicolumn{2}{r}{$ 
 \wedge \send(\server,\BB,\langle [\ballot], [\ballotPrime], [\vote], \zkproof \rangle) \in \tr
 $}
 \\
 \xRightarrow{(2)} &
 y \in [\ballot]
 \end{alignat*}
 \noindent 
 As $\BB$ is honest by topology $\topoInst$, it only records $\BBrecorded([\ballot])$ if it has previously received this term from $\server$ and as the first of four terms.
 As by $\topoInst$ the channel $(\server,\BB)$ is authentic, $\server$ must have sent these terms for $\BB$ to receive them (Step (1)).
 Next, as $\server$ is honest and follows its role specification, 
 the only messages sent by $\server$ that are signed by $\skS$ are the ballots, which are also all contained in the recorded ballots $[\ballot]$, i.e., in the first list of the only message that $\server$ sends to $\BB$ (Step (2)).
 Again no message confusions are possible as $\server$ only sends one message to $\BB$ and as there are no parallel sessions.
 Finally, as the above observations hold for any explicit signal $S$ containing $\server$'s signature, they hold in particular for the signal $S=\Evidence$.
 Moreover, as $\BB$ is honest there is only one signal $\BBpkS(\pkS)$ in the trace.
 Thus, for any $\ballot, \pkS, \confirmation$ such that $\BBpkS(\pkS)\in \tr \wedge \Evidence(\ballot, \confirmation)\in \tr \wedge \ver(\confirmation,\pkS)=\ballot $ 
 it follows that, for some $\bs$, $ \BBrecorded([\ballot])\in \tr \wedge \ballot \in [\ballot]$.
 Therefore, the first disjunct of the definition of $\verdict(\server,\ballot)$ for ballots $\ballot \neq \perp$ is also not satisfied by $\tr$ and we conclude that for all ballots $\ballot$ it holds that $\tr \notin \verdict(\server, \ballot)$.
\end{proof}

\begin{lemma}\label{lemma:protocolInst_indivverif}
 \protocolnameInst satisfies \indivVerif when run in the topologies $\advModel{\topoInst}{\hShH}$ and $\advModel{\topoInst}{\mShH}$.
\end{lemma}
\begin{proof}
 Let $H$ be the distinguished voter for which the property should hold, independently of whether or not there are other (dishonest) voters in the system.
Let $\tr$ be a trace in $\TR(\protInst,\advModel{\topoInst}{\hShH})\cup\TR(\protInst,\advModel{\topoInst}{\mShH})$ such that $\verifyIV(H,\ballot\in[\ballot],\true) \in \tr$ and $\Vote(H,\vote) \in \tr $. 

 \begin{alignat*}{1}
 &\verifyIV(H,\ballot\in[\ballot],\true) \in \tr 
 \wedge \Vote(H,\vote) \in \tr
 \\
\xRightarrow{(1)}&
 \exists \BB, D,P \sd \knows(H,\pair{D,P,\BB,\vote}) \in \tr
 \\\multicolumn{2}{r}{$
 \wedge \receive(D,H,\ballot) \in \tr 
 \wedge \receive(\BB,H,\bs) \in \tr
 \wedge \ballot \in \bs 
 $}
 \\\multicolumn{2}{r}{$
 \wedge \Vote(H,\vote) \in \tr
 $}
 \\
 \xRightarrow{(2)}&
 \send(D,H,\ballot) \in \tr
 \wedge \send(\BB,H,\bs) \in \tr 
 \\\multicolumn{2}{r}{$
 \wedge \Vote(H,\vote) \in \tr
 $}
 \\
 \xRightarrow{(3)}& 
 \exists r,\pkS, \skD, \vote'\sd
 \receive(H,D,\vote') \in \tr 
 \wedge \ballot = \{\encp{\vote'}{\pkS}{r}\}_{\skD}
 \\\multicolumn{2}{r}{$
 \wedge \BBrecorded(\bs) \in \tr
 \wedge \Vote(H,\vote) \in \tr
 $}
\\
 \xRightarrow{(4)} & 
 \send(H,D,\vote') \in \tr 
 \wedge \ballot = \{\encp{\vote'}{\pkS}{r}\}_{\skD}
 \\\multicolumn{2}{r}{$
 \wedge \Vote(H,\vote) \in \tr
 $}
 \\
 \xRightarrow{(5)}& 
 \send(H,D,\vote) \in \tr 
 \wedge \ballot = \{\encp{\vote}{\pkS}{r}\}_{\skD}.
\end{alignat*}
By the assumption $\tr \in \TR(\protInst,\advModel{\topoInst}{\hShH})\cup\TR(\protInst,\advModel{\topoInst}{\mShH})$, $H$ is honest and thus only performs the check with the ballot $\ballot$ that he has previously received from $D$ and the list of recorded ballots $\bs$ that he has previously received from $\BB$ (Step (1)).
Moreover, $\ballot \in \bs$ holds since the last argument of the verifiability check is $\true$.
By the topology $\topoInst$, the channel $(D,H)$ is secure and the channel $(\BB,H)$ is authentic. Therefore, $D$ and $\BB$ must respectively have sent the messages $\ballot$ and $\bs$ for $H$ to receive them on these channels (Step (2)).

$\BB$ is honest by the topology $\topoInst$. Thus, $\BB$ only sends one message containing a single list which was previously recorded in the signal $\BBrecorded$.
Moreover, as there are no parallel sessions, this list cannot be confused with another list sent by $\BB$.
It follows in Step~(3) that $\BBrecorded(\bs) \in \tr$.
The other parts of Step (3) follow as $D$ is honest by the topology $\topoInst$. 
Thus, $D$ only sends a ballot $\ballot$ that contains a vote~$\vote'$, encrypted and signed, that it has previously received from $H$.
Also, no message confusion is possible as there are no parallel sessions and $D$ only sends one message in the protocol.
As by the topology $\topoInst$ the channel $(H,D)$ is secure, $H$ must have sent $\vote'$ for $D$ to receive it (Step (4)).
As $H$ is honest and there are no parallel sessions, he only sends on the channel $(H,D)$ the vote that is also recorded in the signal $\Vote(H,\vote)$. Thus, it holds that $v=v'$ and Step (5) follows.

\end{proof}

\begin{lemma}\label{lemma:protocolInst_TalliedAsRecorded}
 \protocolnameInst satisfies \TalliedAsRecorded when run in the topologies $\advModel{\topoInst}{\hShH}$ and $\advModel{\topoInst}{\mShH}$. 
 \end{lemma}
\begin{proof}
Let $\tr$ be a trace in $\TR(\protInst,\advModel{\topoInst}{\hShH})\cup \TR(\protInst,\advModel{\topoInst}{\mShH})$
such that $q_1=\verifyValidSig(A,\exists [\pkDused]. [\pkDused]\subseteq_l [\pkD] \wedge 
[\ballotPrime]= \verlist(\bs, [\pkDused]),\true) \in \tr$ and 
$q_2= \verifyProof(A, \verifyzkp(\zkproof,[\ballotPrime],[\vote],\pkS),\true) \in \tr$.
\begin{alignat*}{1}
 &q_1 \wedge q_2 
 \\
 \xRightarrow[]{(1)}& 
 q_1 \wedge q_2 \wedge
 \exists [r],\pi \sd 
 [\ballotPrime] =\pi [\encp{\vote}{\pkS}{r}]
 \\
 \xRightarrow[]{(2)}& 
 q_1 \wedge q_2 \wedge
 \exists [\skD] \sd
 [\ballotPrime] =\pi [\encp{\vote}{\pkS}{r}]
 \wedge [\ballot] = [\{\ballotPrime\}_{\skD}] 
 \\\multicolumn{2}{r}{$
 \wedge [\pk(\skD)]=[\pkDused]
 \wedge [\pkDused] \subseteq_l [\pkD]
 $}
 \\
 \xRightarrow[]{(3)}& 
 q_1 \wedge q_2 \wedge
 [\ballot] = \pi [\{\encp{\vote}{\pkS}{r}\}_{\skD}] 
 \wedge [\pk(\skD)]=[\pkDused]
 \\\multicolumn{2}{r}{$
 \wedge [\pkDused] \subseteq_l [\pkD]
 $}
 \\
 \xRightarrow[]{(4)}& 
 \exists \BB \sd \receive(\BB,A,\langle [\ballot], [\ballotPrime], [\vote], p, [\pkD]\rangle) \in \tr
 \\
 \xRightarrow[]{(5)}& 
 \send(\BB,A,\langle [\ballot], [\ballotPrime], [\vote], p, [\pkD]\rangle) \in \tr
 \\
 \xRightarrow[]{(6)}& 
 \BBrecorded([\ballot]) \in \tr 
 \wedge \BBtallied([\vote]) \in \tr
 \wedge \BBpkD([\pkD]) \in \tr .
\end{alignat*}
Step (1) holds as the verification $\verifyProof$ succeeds, indicated by its third argument $\true$, and by the definition of the verification function for zero knowledge proofs $\verifyzkp$ (see Appendix~\ref{subsubsec:newFunctionsAndEquations}).
As the third argument of $\verifyValidSig$ is also $\true$ and by the definition of $\subseteq_l$ (see Appendix~\ref{subsubsec:newFunctionsAndEquations}), Step (2) holds.
Next, Step (3) combines the results of the first two steps.
In (4), we use that the auditor $A$ is honest by the topology $\topoInst$ and follows its role specification. 
Thus, $q_1$ and $q_2$ are only in the trace if $A$ has previously received the corresponding lists from $\BB$.
As the channel $(\BB,A)$ is authentic by the topology $\topoInst$, $\BB$ must thus have sent these messages for $A$ to receive them on this channel (Step (5)).
Finally, (6) holds as by the topology $\topoInst$, $\BB$ is honest and thus only sends out these lists if it has previously published them, i.e., the corresponding signals were recorded in the trace.
Thereby, as $\BB$ only sends one message of this form and as there are no parallel sessions, no message confusion is possible.
Thus, we have shown that all required signals are in the trace and that the lists have the required relations.
\end{proof}

\begin{lemma}\label{lemma:protocolInst_uniqueness}
 \protocolnameInst satisfies $\uniqueness(\server)$ when run in the topologies $\advModel{\topoInst}{\hShH}$ and $\advModel{\topoInst}{\mShH}$.
\end{lemma}

\begin{proof}
 Let $\tr$ be a trace in $\TR(\protInst,\advModel{\topoInst}{\hShH})\cup \TR(\protInst,\advModel{\topoInst}{\mShH})$ such that $\tr \notin \verdict(\server,\ballot')$, $\ballot' \neq \perp $, and $\BBrecorded([\ballot])\in \tr$.
 \begin{alignat*}{1}
 & \ballot' \neq \perp 
 \wedge \tr \notin \verdict(\server,\ballot') 
 \wedge \BBrecorded([\ballot])\in \tr
 \\\xRightarrow[]{(1)}& 
 \exists [\ballotPrime],[\pkD],[H] \sd 
 \ballot' \neq \perp 
 \wedge \tr \notin \verdict(\server,\ballot') 
 \\\multicolumn{2}{r}{$
 \wedge \BBrecorded([\ballot])\in \tr
 \wedge \BBballotWoSignature([\ballotPrime]) \in \tr
 \wedge \BBpkD([\pkD]) \in \tr 
 $}
 \\\multicolumn{2}{r}{$
 \wedge \BBH([H]) \in \tr 
 $}
 \\\xRightarrow[]{(2)}& 
 \exists [\pkDused] \sd
 \BBrecorded([\ballot])\in \tr
 \wedge \BBballotWoSignature([\ballotPrime]) \in \tr
 \\\multicolumn{2}{r}{$
 \wedge \BBpkD([\pkD]) \in \tr 
 \wedge \BBH([H]) \in \tr 
 \wedge [\ballotPrime]= \verlist(\bs, [\pkDused])
 $}
 \\\multicolumn{2}{r}{$
 \wedge [\pkDused] \subseteq_l [\pkD]
 $}
 \\
 \xRightarrow[]{(3)}& 
 \exists [\skD] \sd
 \BBrecorded([\ballot])\in \tr
 \wedge \BBballotWoSignature([\ballotPrime]) \in \tr
 \\\multicolumn{2}{r}{$
 \wedge \BBpkD([\pkD]) \in \tr 
 \wedge \BBH([H]) \in \tr 
 \wedge [\ballot] = [\{\ballotPrime\}_{\skD}] 
 $}
 \\\multicolumn{2}{r}{$
 \wedge [\pk(\skD)]= [\pkDused] 
 \wedge [\pkDused] \subseteq_l [\pkD].
 $}
\end{alignat*}
Steps (1) holds as the bulletin board is honest by topology $\topoInst$, thus when the recorded ballots are published, then the lists $[\ballotPrime]$, $[\pkD]$, and $[H]$ have previously been published.
As $\tr \notin \verdict(\server,\ballot')$ for some ballot $\ballot' \neq \perp$, both disjunct specified in the definition of $\verdict(\server,\ballot')$ for the case $\ballot' \neq \perp$ must be false in $\tr$.
In particular, as the second disjunct is false, all recorded ballots must contain a unique valid signature, i.e., corresponding to a public key that is contained in $[\pkD]$, indicated by Step (2).
Step (3) holds by the definitions of $\verlist$. 

Let $\ballot$ be an arbitrary recorded ballot.
By the above formula, it has the form $\ballot= \{\ballotPrime\}_{\skD}$ for some signing key $\skD$ such that $\pk(\skD) \in [\pkD]$ (thus $\ver(\ballot, \pk(\skD))\neq \perp$) and such that no other recorded ballot is associated with the same key $\skD$.
As the bulletin board $\BB$ is honest by topology $\topoInst$ and as there are no parallel sessions, $\BB$ only publishes the list $[\pkD]$ from its initial knowledge in the signal $\BBpkD$.
By the shared initial knowledge and the setup of the protocol (there are no parallel sessions), for each verification key $\pkD$ in the list $[\pkD]$ from $\BB$'s initial knowledge, there is exactly one device which knows the corresponding signing key $\skD$ and one voter $H$ who is associated with this device. Moreover, the device and voter associated with each verification key are distinct.
Also, the verification key~$\pkD$ is together with the voter $H$ in the list in $\corresponds([\human,\pkD])$, and $H$ is in the list $[\human]$.
Furthermore, these list are correctly published on the bulletin board $\BB$ as $\BB$ is honest by topology $\topoInst$.
From this, all conjuncts in the property's definition (Definition~\ref{def:uniquenessMV}) except for the sends already follow.

We next establish that each $\skD$ is only known to one device $D$ at all times during the execution.
First, by the specified initial knowledge, the signing key is only known to this device at the protocol's start. 
Second, by the topology $\topoInst$ the agent instantiating the device is honest and follows its specification. 
Therefore, $D$ only sends one message, which contains the term $\skD$ as a signature.
By our equational theory, it is not possible to extract a signing key from a signed message. 
Therefore, no agent other than $D$ knows $\skD$ during the protocol execution.

We have established that $\ballot=\{\ballotPrime\}_{\skD}$ is recorded on the bulletin board and that no agent except $D$ knows the term $\skD$. Thus, $D$ must have computed $\ballot$.
As we have just argued that $D$ is honest, $D$ only sends one message to its associated voter $H$ on a secure channel.
Therefore, in all traces where $\ballot$ is recorded on the bulletin board, the unique voter $H$ must have forwarded $\ballot$, that is $H$ must have sent a message containing $\ballot$ to some agent $A$. 
Thus, $\send(H,A,m)$ is in the trace with $m \subterm \ballot$. 
The above reasoning holds for any ballot.
Furthermore, we argued that for each recorded ballot, there is a unique associated signing key and device, i.e., distinct from those associated with the other recorded ballots, and that each device has a unique associated voter.
Therefore, it follows that for any two distinct recorded ballots the required signal $\send$ is recorded with a distinct voter.
\end{proof}

\begin{lemma}\label{lemma:protocolInst_EligibilityVerif}
 \protocolnameInst satisfies \EligibilityVerif when run in the topologies $\advModel{\topoInst}{\hShH}$ and $\advModel{\topoInst}{\mShH}$.
 \end{lemma}
 
\begin{proof}
Let $\tr$ be a trace in $\TR(\protInst,\advModel{\topoInst}{\hShH})\cup \TR(\protInst,\advModel{\topoInst}{\mShH})$ such that 
$q_1=\verifyValidSig(A,\exists [\pkDused]. [\pkDused]\subseteq_l [\pkD] \wedge 
[\ballotPrime]=\verlist(\bs, [\pkDused]),\true) \in \tr$ and 
$q_2= \verifyProof(A, \verifyzkp(\zkproof,[\ballotPrime],[\vote],\pkS),\true) \in \tr$ and let $i$ and $j$ be two natural numbers in $\{1,\dots,n\}$, where $n$ is the length of the lists $\vs$ and~$\bs$.
\begin{alignat*}{1}
 &q_1 \wedge q_2 \wedge i,j \in\{1,\dots,n\}
 \\\xRightarrow[]{(1)}& 
 \exists [\skD],[r], \pi \sd
 q_1 \wedge q_2 \wedge i,j \in\{1,\dots,n\}
 \\\multicolumn{2}{r}{$
 \wedge [\ballot] = \pi [\{\encp{\vote}{\pkS}{r}\}_{\skD}] 
\wedge [\pk(\skD)]=[\pkDused]
 \wedge [\pkDused] \subseteq_l [\pkD]
 $}
 \\\multicolumn{2}{r}{$
 \wedge \BBrecorded([\ballot])\in \tr 
 \wedge \BBtallied([\vote])\in \tr 
 \wedge \BBpkD([\pkD])\in \tr
 $}
\\\xRightarrow{(2)}&
\exists [H] \sd
q_1 \wedge q_2 \wedge i,j \in\{1,\dots,n\}
 \wedge [\ballot] = \pi [\{\encp{\vote}{\pkS}{r}\}_{\skD}] 
 \\\multicolumn{2}{r}{$
\wedge [\pk(\skD)]=[\pkDused]
 \wedge [\pkDused] \subseteq_l [\pkD]
 \wedge \BBrecorded([\ballot])\in \tr
 $}
 \\\multicolumn{2}{r}{$
 \wedge \BBtallied([\vote])\in \tr 
 \wedge \BBpkD([\pkD])\in \tr
\wedge \BBballotWoSignature([\ballotPrime]) \in \tr 
 $}
 \\\multicolumn{2}{r}{$
\wedge \BBH([H]) \in \tr
 $}
\\\xRightarrow{(3)}&
\exists i', j',P_1,P_2,m_1,m_2 \sd
i, j \in \{1,\dots,n\}
\\\multicolumn{2}{r}{$
\wedge [\ballot] = \pi[\{\encp{\vote}{\pkS}{r}\}_{\skD}] 
\wedge [\pk(\skD)]=[\pkDused]
 \wedge [\pkDused] \subseteq_l [\pkD]
 $}
 \\\multicolumn{2}{r}{$
 \wedge \BBrecorded([\ballot])\in \tr 
 \wedge \BBtallied([\vote])\in \tr 
 \wedge \BBpkD([\pkD])\in \tr
 $}
 \\
 \multicolumn{2}{r}{$ 
\wedge \BBballotWoSignature([\ballotPrime]) \in \tr 
\wedge \BBH([H]) \in \tr
 $}
\\\multicolumn{2}{r}{$ 
 \wedge \send([H]_{i'},P_1,m_1)\in \tr 
 \wedge \send([H]_{j'},P_2,m_2)\in \tr 
 $}
\\\multicolumn{2}{r}{$ 
 \wedge m_1 \subterm [b]_i
 \wedge m_2 \subterm [b]_j 
 \wedge (i \neq j \implies [H]_{i'} \neq [H]_{j'}).
$}
\end{alignat*}
Step (1) holds by the same reasoning as in the proof of Lemma~\ref{lemma:protocolInst_TalliedAsRecorded}.
Next, Step (2) holds as $\BB$ is honest by the topology $\topoInst$, thus when the signals $\BBtallied$ and $\BBrecorded$ are recorded in the trace, then the signals $\BBballotWoSignature$ and $\BBH$ have also been recorded.
Step (3) holds by the reasoning in the proof of Lemma~\ref{lemma:protocolInst_uniqueness}, where we derived from Step (3) that $\uniqueness(\server)$ holds.
As for each vote $v$ in $[v]$ there is a ballot $b$ in $[b]$ such that $\ballot =\{\encp{\vote}{\pkS}{r}\}_{{\skD}}$, it follows that each such $\send$-signal in Step (3) contains a unique vote from the final tally.

As we have argued in the proof of Lemma~\ref{lemma:protocolInst_uniqueness} that only each voter's device can compute such a ballot with the valid verification key, for each voter $H$ to send a message containing the ballot, $H$ must have previously learned it from his device $D$. 
Moreover, as each device $D$ is honest by the topology $\topoInst$, $D$ only computes the ballot $\ballot =\{\encp{\vote}{\pkS}{r}\}_{{\skD}}$ when it has previously received $v$ from its associated voter $H$ on a secure channel.
Thus, for each voter $H$ to send a ballot of the form $\ballot =\{\encp{\vote}{\pkS}{r}\}_{{\skD}}$, $H$ must have previously sent the corresponding vote $\vote$ to his device. 
In particular, as there are no parallel sessions and $D$ only receives and sends one message, no message confusion is possible.
We conclude that each (distinct) vote in the final tally must have been sent by a (distinct) voter to his device.
\end{proof}

\begin{lemma}\label{lemma:protocolInst_provoterAbstain}
 \protocolnameInst satisfies $\provoterAbstain(\server)$ when run in the topologies $\advModel{\topoInst}{\hShH}$ and $\advModel{\topoInst}{\mShH}$.
\end{lemma}
\begin{proof}
Let $\tr$ be a trace in $\TR(\protInst,\advModel{\topoInst}{\hShH})\cup \TR(\protInst,\advModel{\topoInst}{\mShH})$.
As all traces in $\TR(\protInst,\advModel{\topoInst}{\hShH})\cup \TR(\protInst,\advModel{\topoInst}{\mShH})$ satisfy $\uniqueness(\server)$ by Lemma~\ref{lemma:protocolInst_uniqueness} and as the preconditions of Theorem~\ref{theorem:uniquenessImpliesDR} hold (i.e., the protocol does not allow re-voting and a voter who abstains does not send any messages), we use Theorem~\ref{theorem:uniquenessImpliesDR} to infer that $\tr$ satisfies $\provoterAbstain(\server)$.
\end{proof}

\end{document}